\documentclass[%
 aip,
 jmp,%
 amsmath,amssymb,
preprint,
]{revtex4-2}

\usepackage{graphicx}% Include figure files
\usepackage{dcolumn}% Align table columns on decimal point
\usepackage{bm}% bold math
\usepackage[utf8]{inputenc}
\usepackage[T1]{fontenc}
\usepackage{mathptmx}
\usepackage{etoolbox}
\usepackage{amsthm}
\usepackage{bbm}
\usepackage{graphicx}
\usepackage{qcircuit}
\usepackage{tikz}

\newtheorem{theorem}{Theorem}
\newtheorem{corollary}[theorem]{Corollary}
\newtheorem{lemma}[theorem]{Lemma}
\newtheorem{proposition}[theorem]{Proposition}
\newtheorem{example}[theorem]{Example}

\newtheorem{remark}[theorem]{Remark}

\newtheorem{conjecture}[theorem]{Conjecture}

\makeatletter
\def\@email#1#2{%
 \endgroup
 \patchcmd{\titleblock@produce}
 {\frontmatter@RRAPformat}
 {\frontmatter@RRAPformat{\produce@RRAP{*#1\href{mailto:#2}{#2}}}\frontmatter@RRAPformat}
 {}{}
}%
\makeatother
\begin{document}

%\preprint{AIP/123-QED}

\title{Which Bath-Hamiltonians Matter for Thermal Operations?}
% Force line breaks with \\
\author{Frederik vom Ende}
% \altaffiliation[Also at ]{Physics Department, XYZ University.}%Lines break automatically or can be forced with \\
%\author{B. Author}%
 \email{frederik.vom-ende@tum.de}
\affiliation{ 
Department of Chemistry, Technische Universit{\"a}t M{\"u}nchen, Lichtenbergstra{\ss}e 4, 85737 Garching, Germany%\\This line break forced with \textbackslash\textbackslash
}%
\affiliation{Munich Centre for Quantum Science and Technology (MCQST) \& Munich Quantum Valley (MQV), Schellingstra{\ss}e 4, 80799 M{\"u}nchen, Germany}

%\author{C. Author}
% \homepage{http://www.Second.institution.edu/~Charlie.Author.}
%\affiliation{%
%Second institution and/or address%\\This line break forced% with \\
%}%

\date{\today}% It is always \today, today,
 % but any date may be explicitly specified

\begin{abstract}
In this article we explore the set of thermal operations from a mathematical and topological point of view.
First we introduce the concept of Hamiltonians with resonant spectrum with respect to some reference Hamiltonian, followed by proving that when defining thermal operations it suffices to only consider bath Hamiltonians which satisfy this resonance property.
Next we investigate continuity of the set of thermal operations in certain parameters, such as energies of the system and temperature of the bath. 
We will see that the set of thermal operations changes discontinuously with respect to the Hausdorff metric at any Hamiltonian which has so-called degenerate Bohr spectrum, regardless of the temperature.
Finally we find a semigroup representation of the (enhanced) thermal operations in two dimensions by characterizing any such operation via three real parameters, thus allowing for a visualization of this set.
Using this, in the qubit case we show commutativity of the (enhanced) thermal operations as well as convexity of the thermal operations without the closure. The latter is done by specifying the elements of this set exactly.
\end{abstract}

\pacs{
%02.20.-a, %Group theory
%02.30.Tb, %Operator theory
02.40.Pc, %Topology
03.65.Aa, %Quantum systems with finite Hilbert space
03.67.-a %Quantum information
05.70.-a %Thermodynamics
}% PACS, the Physics and Astronomy Classification Scheme.

\keywords{
thermal operations; quantum thermodynamics; semigroup representation; Bohr spectrum
}%Use showkeys class option if keyword display desired

\maketitle

\section{Introduction}
Over the last decade, sparked by Brand\~ao et al.~\cite{Brandao15}, Horodecki \& Oppenheim 
\cite{Horodecki13}, as well as Renes \cite{Renes14}---and further pursued by others
\cite{Faist17,Gour15,Lostaglio18,Sagawa19,Mazurek19,Alhambra19}---thermo-majorization and in particular its
 resource theory approach % to quantum thermodynamics
 has been a widely discussed and researched topic in quantum physics. Here the 
 central question is:
 Given a fixed background temperature as well as initial and target states 
of a quantum system,
 % with Hamiltonian (hermitian operator) $H_S$,
 can the former be mapped to the 
 latter by means of a thermal operation? These channels are the fundamental building block of the resource theory approach to quantum thermodynamics as they, roughly speaking, are the operations which are assumed to be performable in arbitrary number without any cost; for a precise definition, cf.~Section \ref{sec_therm_op_def}.
 Thus, arguably, studying and understanding the thermal operations, their structure, and their properties is of crucial importance.

The concept of thermal operations is an attempt to formalize which operations can 
be carried out at no cost (with respect to some resource, e.g., work). Recall that in 
macroscopic systems a state transformation is thermodynamically possible if and only if the 
free energy decreases. In the quantum realm---using the currently accepted definition 
of thermal operations---this is at least necessary:
the non-equilibrium system free energy\cite{Esposito11}
$F=\operatorname{tr}(H (\cdot))-k_BTS$ 
cannot increase under any thermal operation. 
Here $H$ is the system's Hamiltonian, $S$ is the von Neumann
entropy, $T$ the temperature of the environment, and $k_B$ the Boltzmann
constant.
This property of not increasing actually holds for both the free energy of the classical (diagonal) part as well as the so-called asymmetry (relative entropy of the coherences); together these add up to the free energy \cite{Lostaglio15_2}.
However, the decrease of the free energy is not sufficient to guarantee state conversion via thermal operations (Example 6 in \cite{Lostaglio19}).
This changes once one relaxes the set of operations to those which leave not the energy, but the \textit{average} energy of system plus bath invariant as these are precisely the channels which decrease the free energy\cite{Skrzypczyk14}.
For the interconversion of classical states, considering not only the free energy but a collection of generalized free energies leads to a characterization of this problem when allowing for \textit{catalytic} thermal operations (Thm.~18 in \cite{Brandao15}).
For a comprehensive introduction to this topic we refer to the review article by Lostaglio \cite{Lostaglio19}.

These conditions imposed by the generalized free energies have been called ``the second laws of quantum thermodynamics'' in the past. 
On a related note there is also a third law of quantum thermodynamics, at least for qubits. Scharlau et al.~gave a lower bound on the population of the lowest energy level when applying any thermal operation (Thm.~9 in \cite{Scharlau18}). In particular this implies that no non-ground state can be mapped exactly to the ground state by means of thermal operations with finite heat baths.
This is a refinement of the related result that no state with trivial kernel (i.e.~$0$ is not an eigenvalue of the state) can be mapped to the ground state -- or any pure state for that matter -- by means of a Gibbs-preserving channel (Coro.~4.7 in \cite{vomEnde20Dmaj}).

The problem of characterizing state 
conversions as mentioned in the beginning is fully solved in the classical regime. This has to do with the observation
made early on that thermal operations and general Gibbs-preserving quantum maps 
are (approximately) indistinguishable on quasi-classical states.
Indeed, given a system described by $\operatorname{diag}(E_1,\ldots,E_n)$ with 
background temperature $T\in(0,\infty]$, transforming $\operatorname{diag}(y)$ into 
$\operatorname{diag}(x)$ via thermal operations is possible if and only if
$\|x-\frac{y_i}{d_i}d\|_1\leq\|y-\frac{y_i}{d_i}d\|_1$ holds for all $i=1,\ldots,n $ where 
$d:=(e^{-E_j/T})_{j=1}^n$ is the vector of Gibbs weights\cite{vomEnde22}. 
Equivalently, the so-called ``thermo-majorization curve'' (a piecewise linear 
bijection on the interval $[0,1]$) corresponding to $y$ must not lie below the 
curve corresponding to $x$ anywhere\cite{Horodecki13}.
This reduces the classical state conversion problem to a finite list of conditions, that 
is, $n$ simple $1$-norm inequalities or, using thermo-majorization curves, to $n-1$ 
inequalities each involving a minimum over a set of $n$ elements (Thm.~4 in \cite
{Alhambra16}).
For more detail as well as further characterizations we refer to Prop.~1 in \cite{vomEnde22}.

Be aware that it was also noticed early on that the thermal operations form a strict subset of the Gibbs-preserving maps as soon as coherences come into play \cite{Faist17}. 
This is one of the reasons why the state conversion problem becomes much more complicated in the quantum case: While there exists a characterization via infinitely many inequalities involving the conditional min-entropy\cite{Gour18} a simple characterization -- like in the classical case -- beyond qubits is still amiss, refer also to Section 4.2 in \cite{vomEnde20Dmaj}.
There have been different ways to deal with this problem in the past: While some authors constrained the set of thermal operations to simpler subsets, e.g., such which are experimentally implementable using current technology \cite{Perry18,Lostaglio18}, 
others\cite{Scharlau18,Ding19,Ding21} focused on learning more about the role of the bath Hamiltonian in the action of thermal operations.
In this article we follow the second line of thought.

This work is organized as follows. In Section \ref{sec_therm_op_def} we introduce the concept of bath Hamiltonians having ``resonant spectrum'' with respect to a given system, and we show that these are everything one needs to generate (approximate) all thermal operations (Proposition \ref{prop_1}).
As a special case we recover and refine a result about the structure of thermal operations if the system in question is a spin system, that is, if the Hamiltonian has equidistant eigenvalues (Corollary \ref{prop_to_spin} \& \ref{coro_spin_rational}).
These corollaries suggest that the set of thermal operations may in some sense change discontinuously at certain Hamiltonians; this we investigate in Section \ref{sec_continuity}. There we look at two particular systems where this discontinuity manifests (Example \ref{ex_discont}). These examples can be generalized to arbitrary dimensions and Hamiltonians with certain properties, thus revealing a structural problem rather than being singled-out counter-examples.
Finally in Section \ref{sec_ento} we visualize the set of qubit thermal operations as a three-dimensional shape (Figure \ref{fig1}). Using this as well as our results regarding baths with resonant spectrum we give a full answer to what elements the qubit thermal operations consist of, and what role degenerate bath Hamiltonians play (Theorem \ref{thm_equiv_qubit}).

\section{Thermal Operations: The Basics}\label{sec_therm_op_def}
%\subsection{Definition and Basic Properties}
We start by reviewing how thermal operations are defined and what basic properties they have.
Consider an $n$-level system described by some $H_S\in\mathbb C^{n\times n}$ Hermitian (``system's Hamiltonian'') as well as some $T>0$ (``fixed background temperature''). Given any $m\in\mathbb N$ we define
\begin{align*}
\Phi_{T,m}:{i}\mathfrak{u}(m)\times \mathsf{U}(mn)&\to \textsc{cptp}(n)\\
(H_B,U)&\mapsto \operatorname{tr}_B\Big(U\Big((\cdot)\otimes \frac{e^{-H_B/T}}{\operatorname{tr}(e^{-H_B/T})}\Big)U^*\Big)
\end{align*}
where $\mathsf U(m)$ is the unitary group in $m$ dimensions, $\mathfrak{u}(m)$ is its Lie algebra (so ${i}\mathfrak u(m)$ is the collection of all Hermitian $m\times m$ matrices), and $\textsc{cptp}(n)$ is the set of all completely positive, trace preserving, linear maps on $\mathbb C^{n\times n}$. 
Thus $\Phi_{T,m}(H_B,U)$ represents first coupling the system described by $H_S$ to an $m$-dimensional bath described by $H_B$ at temperature $T$, then applying the unitary channel $\operatorname{Ad}_U=U(\cdot)U^*$ to the full system, and finally discarding the bath.
Using this notation,
following Lostaglio\cite{Lostaglio19} we
define
%\begin{definition}\label{def_1}
%Given $H_S\in\mathbb C^{n\times n}$ Hermitian (``system's Hamiltonian'') and some $T>0$ (``fixed background temperature'') one defines 
the thermal operations with respect to $H_S,T$ as
\begin{align}
\mathsf{TO}(H_S,T):=&\bigcup_{m\in\mathbb N}\Big\{ \Phi_{T,m}(H_B,U) :\ \substack{H_B\in{i}\mathfrak{u}(m),U\in\mathsf{U}(mn)\\U(H_S\otimes\mathbbm{1}_B+\mathbbm{1}\otimes H_B)U^*=H_S\otimes\mathbbm{1}_B+\mathbbm{1}\otimes H_B }\Big\}\label{eq_def_TO_1}\\
=& \bigcup_{m\in\mathbb N}\Big\{ \Phi_{T,m}(H_B,e^{iH_\mathsf{tot}}) :\ \substack{
H_B\in{i}\mathfrak{u}(m), H_\mathsf{tot}\in{i}\mathfrak u(mn)\\
[H_\mathsf{tot},H_S\otimes\mathbbm{1}_B+\mathbbm{1}\otimes H_B]=0
}\Big\} \label{eq_def_TO_2}\\
=& \bigcup_{m\in\mathbb N}\Big\{ \Phi_{T,m}(H_B,e^{iH_\mathsf{tot}}):\ \substack{
H_B=\operatorname{diag}(E_1,\ldots,E_m)\text{ with }E_1\leq\ldots\leq E_m\\
H_\mathsf{tot}\in{i}\mathfrak{u}(mn),
[H_\mathsf{tot},H_S\otimes\mathbbm{1}_B+\mathbbm{1}\otimes H_B]=0
}\Big\} \,.\label{eq_def_TO_3}
\end{align}
Physically, $H_\mathsf{tot}$ includes the system-bath-interaction $H_{SB}$, that is, $H_\text{tot} = H_S \otimes  1_B + 1_S \otimes  H_B + H_{SB}$ where the commutator condition then reduces to $[H_{SB}, H_S \otimes  1_B + 1_S \otimes  H_B] = 0$, cf.~also \cite{Kosloff21}.

Now to see that the sets \eqref{eq_def_TO_1} and \eqref{eq_def_TO_2} are equal note that the subgroup of $\mathsf U(n)$ which stabilizes $H_S\otimes\mathbbm{1}_B+\mathbbm{1}\otimes H_B$ is compact and connected (because $\mathsf U(n)$ is compact and the stabilized element is Hermitian). Therefore $\mathsf{exp}$ maps onto this subgroup and we can replace the stabilizing condition on $U$ by the equivalent condition on the level of generators on $H_\mathsf{tot}$.

For equality of \eqref{eq_def_TO_2} and \eqref{eq_def_TO_3} in the definition of $\mathsf{TO}(H_S,T)$ (henceforth $\mathsf{TO}$ for short), i.e.~the fact there is some unitary degree of freedom on the ancilla despite energy-conservation, note
$
\Phi_{T,m}(H_B,U)=\Phi_{T,m}(\operatorname{Ad}_V(H_B),\operatorname{Ad}_{\mathbbm1\otimes V}(U))
$
for all $H_B\in{i}\mathfrak u(m)$, $U\in\mathsf U(mn)$, $V\in\mathsf U(m)$; this follows from the partial trace identity $\operatorname{tr}_2((A\otimes B)C(D\otimes B^{-1}))=A\operatorname{tr}_2(C)D$.
%\footnote{\label{footnote_partialtrace}
%One verifies $\operatorname{tr}_2((A\otimes B)C(D\otimes B^{-1})=A\operatorname{tr}_2(C)D$ by means of a straightforward computation: for all $X\in\mathbb C^{n\times n}$
%$$
%\operatorname{tr}\big(X \operatorname{tr}_2((A\otimes B)C(D\otimes B^{-1})) \big)=\operatorname{tr}((X\otimes\mathbbm{1}_m)(A\otimes B)C(D\otimes B^{-1})\big)=\operatorname{tr}\big((DXA\otimes\mathbbm{1}_m)C\big)=\operatorname{tr}(XA\operatorname{tr}_2(C)D)\,.
%$$
%}
Moreover $\operatorname{Ad}_{\mathbbm1\otimes V}(U)$ is energy-conserving with respect to $(H_S,\operatorname{Ad}_V(H_B))$ so in particular we can choose $V$ such that it diagonalizes $H_B$. What this implies is that the only relevant information coming from the bath is the spectrum of the associated Hamiltonian together with its degeneracies.

\begin{remark}[Thermal Operations in the High Temperature Limit]
Observing that the Gibbs state of any finite-dimensional system becomes the maximally mixed state in the limit $T\to\infty$ one can extend the definition of thermal operations to $T=\infty$ via
\begin{align*}
\mathsf{TO}(H_S,\infty):=&\bigcup_{m\in\mathbb N}\Big\{ \Phi_{1,m}(\mathbbm1_m,U):\ \substack{H_B\in{i}\mathfrak{u}(m),U\in\mathsf{U}(mn)\\U(H_S\otimes\mathbbm{1}_B+\mathbbm{1}\otimes H_B)U^*=H_S\otimes\mathbbm{1}_B+\mathbbm{1}\otimes H_B }\Big\} \\
=& \bigcup_{m\in\mathbb N}\Big\{ \Phi_{1,m}(\mathbbm1_m,e^{iH_\mathsf{tot}}):\ \substack{
H_B=\operatorname{diag}(E_1,\ldots,E_m)\text{ with }E_1\leq\ldots\leq E_m\\
H_\mathsf{tot}\in{i}\mathfrak{u}(mn),
[H_\mathsf{tot},H_S\otimes\mathbbm{1}_B+\mathbbm{1}\otimes H_B]=0
}\Big\} \,.
\end{align*}
While it might seem that the bath Hamiltonian is redundant as it does not appear in 
the argument of $\Phi_{1,m}$, waiving it from the energy conservation condition (i.e.~setting $H_B=0$)
would reduce
the set to only dephasing thermalizations (sometimes called 
``Hadamard channels'' because they are of Hadamard product form $\rho\mapsto P*\rho:=(p_{ij}\rho_{ij})_{i,j=1}^n$ for 
some positive semi-definite $P\in\mathbb C^{n\times n}$ which has only ones on the 
diagonal, cf.~Ch.~1.2 in \cite{Bhatia07}).
In particular this would disallow any non-trivial action on diagonal matrices which is, however, certainly possible within $\mathsf{TO}$.
%Equal to the set of all Hadamard channels: Lemma \ref{lemma_hadamard_channel_decomp} (Appendix \ref{app_lemma}); proof shows: diagonal bath Hamiltonians (relative phases) suffice. In particular: no mixing allowed, only scaling -- significant difference compared to $T\in(0,\infty)$
\end{remark}
%Now choose $V$ such that $\tilde H_B:=\operatorname{Ad}_{V}(H_B)$ is diagonal with non-decreasing diagonal entries, as well as $\tilde U:=\operatorname{Ad}_{\mathbbm1\otimes V}(U)$. Then $
%\tilde U(H_S\otimes\mathbbm1_B+\mathbbm1\otimes \tilde H_B)\tilde U^*=(\mathbbm1\otimes V)U(H_S\otimes\mathbbm1_B+\mathbbm1\otimes H_B)U^*(\mathbbm1\otimes V^*)=0\,.
%$
%
%
%In other words $\mathsf{TO}(H_S,T)$ is defined to be (the union of all images of) $\Phi_{T,m}$ when restricting the second argument to the stabilizer subgroup of the uncoupled Hamiltonian with respect to the adjoint representation.
Having explained how $\mathsf{TO}$ is defined for infinite temperatures let us illustrate how for all $0<T\leq\infty$ the set of thermal operations changes under some elementary transformations of the system's Hamiltonian:
\begin{lemma}\label{lemma_3}
Given $H_S\in{i}\mathfrak{u}(n)$, $T\in(0,\infty]$, and $U\in\mathsf U(n)$ the following statements hold:
\begin{itemize}
\item[(i)] $\mathsf{TO}(\lambda H_S+\mu\mathbbm1,\lambda T)=\mathsf{TO}(H_S,T)$ for all $\lambda>0,\mu\in\mathbb R$.
\item[(ii)] $ \mathsf{TO}(\operatorname{Ad}_U(H_S),T)=\operatorname{Ad}_U\circ\; \mathsf{TO}(H_S,T)\circ \operatorname{Ad}_{U^*}$
%\item[(iii)] Let any thermal operation $S$ with corresponding $H_B$ Hermitian and energy-preserving unitary $V$ be given.
%%, i.e.~$S=\operatorname{tr}_B(V((\cdot)\otimes\frac{e^{-H_B/T}}{\operatorname{tr}(e^{-H_B/T})})V^*)$ for some $m\in\mathbb N$, $H_B\in\mathbb C^{m\times m}$ Hermitian, and $V\in\mathbb C^{mn\times mn}$ unitary such that $[V,H_S\otimes\mathbbm 1_B+\mathbbm 1\otimes H_B]=0$.
%If there exists $W\in\mathbb C^{n\times n}$ unitary with $WH_BW^*=H_B$ such that $V(U\otimes W)V^*=U\otimes W$, then $[S,\operatorname{Ad}_U]=0$.
\end{itemize}
These statements continue to hold when replacing $\mathsf{TO}$ by its closure.
\end{lemma}
\noindent This is straightforward to show so we omit the proof.

Now let us have a closer look at the condition $[H_\mathsf{tot},H_S\otimes\mathbbm{1}_B+\mathbbm{1}\otimes H_B]=0$ from
%the 
Equation \eqref{eq_def_TO_2}
%definition of $\mathsf{TO}$
which encodes energy-conservation:
this imposes block-diagonal structure on $H_\mathsf{tot}$ in some eigenbasis of $H_S\otimes H_B$, and the sizes of the blocks correspond to how degenerate the eigenvalues of $H_S\otimes\mathbbm{1}_B+\mathbbm{1}\otimes H_B$ are.
Letting $\sigma(\cdot)$ henceforth denote the spectrum of any matrix: because $\sigma(H_S\otimes\mathbbm{1}_B+\mathbbm{1}\otimes H_B)=\sigma(H_S)+\sigma(H_B)$ this means that $H_\mathsf{tot}$ acts non-trivially on an energy level $E$ of the full system only if $E$ can be decomposed into a sum of elements from $\sigma(H_S)$ and $\sigma(H_B)$ in more than one way, i.e.~$E=E_i+E_l'=E_j+E_k'$ for some pairwise different $E_i,E_j\in\sigma(H_S)$, $E_k',E_l'\in\sigma(H_B)$.
But this is equivalent to $E_i-E_j=E_k'-E_l'$ which is the necessary condition for the diagonal entries $\rho_{ii},\rho_{jj}$ of the state of the system $\rho$ to mix by means of a thermal operation (cf.~Remark \ref{rem_bath_condition_prob} for details).
Because the spectrum of $\operatorname{ad}_{H_S}:=[H_S,\,\cdot\,]$ is given by $\{E_i-E_j:i,j\}$ this motivates the following definition: Given $m,n\in\mathbb N$, $H_S\in{i}\mathfrak u(n)$, $H_B\in{i}\mathfrak u(m)$
% there exists a permutation $\tau\in S_m$ such that
%$
%E'_{\tau(j+1)}-E'_{\tau(j)}\in\sigma(\operatorname{ad}_{H_S})
%$ for all $j=1,\ldots,m-1$. Here $E_1',\ldots,E_m'$ are the eigenvalues of $H_B$.
%Equivalently one can define this via translating the bath Hamiltonian into an
define an undirected graph with vertices being the eigen-energies of $H_B$, and two vertices are connected if the difference of the corresponding energies appears in the spectrum of $\operatorname{ad}_{H_S}$. We say that $H_B$ has \textit{resonant} or \textit{absorbing spectrum with respect to} $H_S$ if this graph is connected\footnote{
Equivalently $H_B$ having resonant spectrum w.r.t~$H_S$ is characterized by the following condition: For all proper (non-empty) subsets $I$ of $\{1,\ldots,m\}$ there exist $i\in I$ and $j\in\{1,\ldots,m\}\setminus I$ such that $E_i'-E_j'\in\sigma(\operatorname{ad}_{H_S})$. This means that the above graph cannot be written as the union of two (or more) disconnected components.
}.
%counted with multiplicities
%, i.e.~$H_B=V\operatorname{diag}(E_1',\ldots,E_m')V^*$ for some $V\in\mathsf U(m)$.
To illustrate this definition we refer to the examples shown in Figure \ref{figa}.
\begin{figure}[!htb]
\centering
\includegraphics[width=0.75\textwidth]{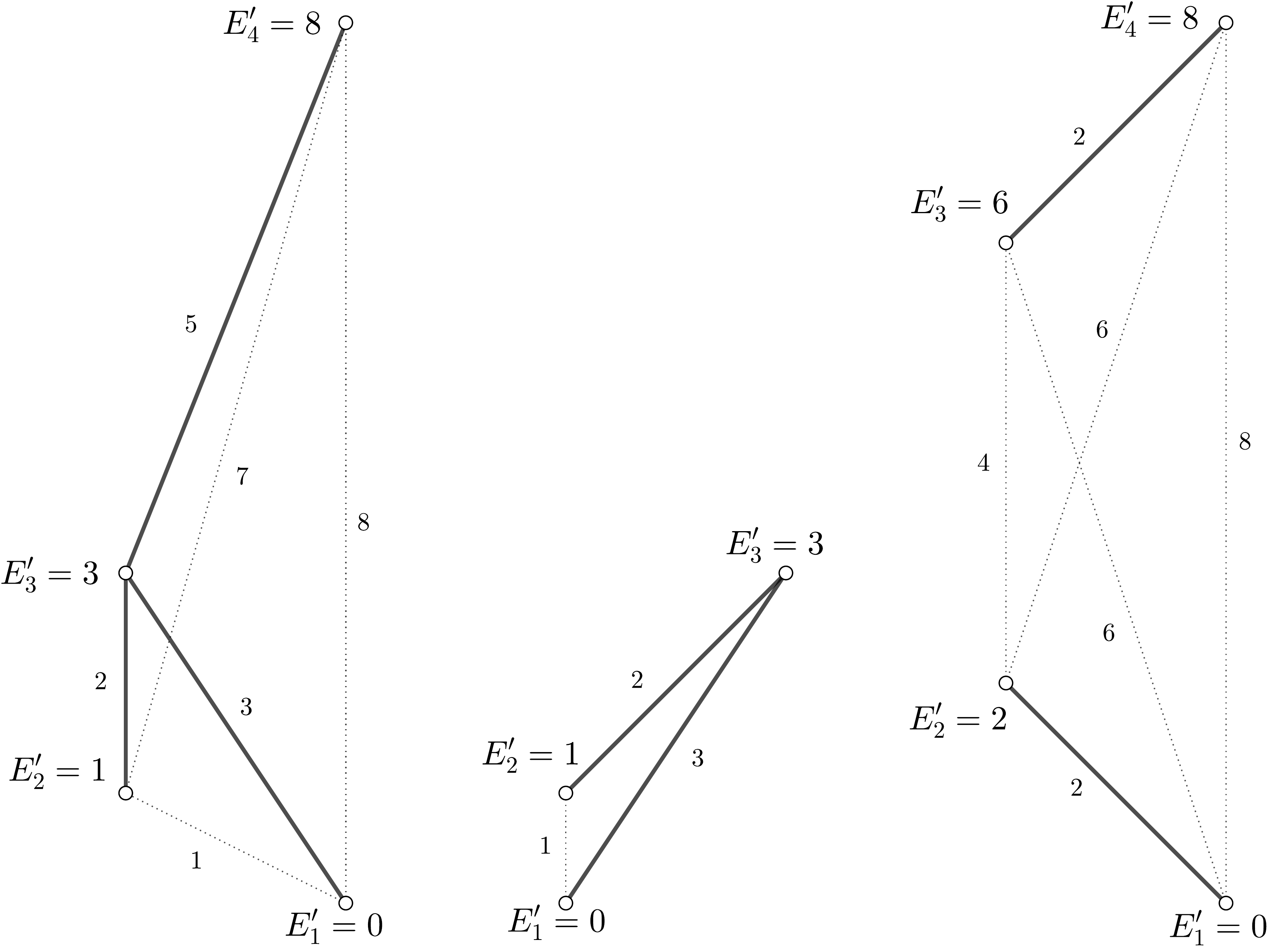}
\caption {Let us investigate whether the following bath Hamiltonians have resonant spectrum with respect to $H_S=\operatorname{diag}(0,2,5)$, that is, $|\sigma(\operatorname{ad}_{H_S})|=\{0,2,3,5\}$.\\
Left: $H_{B,1}=\operatorname{diag}(0,1,3,8)$ has resonant spectrum with respect to $H_S$ because its graph is connected.
Middle: $H_{B,2}=\operatorname{diag}(0,1,3)$ also has resonant spectrum w.r.t.~$H_S$ for the same reason.
Right: $H_{B,3}=\operatorname{diag}(0,2,6,8)$ does not have resonant spectrum 
w.r.t.~$H_S$ as it decomposes into the connected components $\{0,2\}$, $\{6,8\}$. 
These do not ``interact'' with each other because none of the energy differences between them is in $\sigma(\operatorname{ad}_{H_S})$.
}\label{figa}
\end{figure}%${}$

\begin{remark}\label{rem_bath_condition_prob}
A word of warning: The definition of a Hamiltonian having resonant spectrum is similar to---but should not be confused with---one of the assumptions on heat baths from the early works on thermal operations\cite{Horodecki13}.
There it was assumed that for any two energies $E_i,E_j$ of the system and any energy $E_k'$ of the bath there exists some energy $E_l'$ of the bath such that $E_i-E_j=E_k'-E_l'$. However no finite size heat bath can satisfy this; violations of this condition appear always at the edge, and in some cases even in the bulk of the energy band. These problems are discussed comprehensively yet in detail in Appendix A in \cite{Shiraishi20}.

This condition is related to a necessary criterion for ``interaction'' between different entries of a quantum state: Given any Hamiltonian $H_S=\sum_{j=1}^nE_j|g_j\rangle\langle g_j|$ which describes a system currently in the state $\rho$ -- represented for now in the eigenbasis of $H_S$, i.e.~$(\rho_{jk})_{j,k=1}^n$ with $\rho_{jk}:=\langle g_j,\rho g_k\rangle$ -- a thermal operation $\Phi_{T,m}(H_B,U)$ can mix $\rho_{ij}$ and $\rho_{kl}$ only if\footnote{
W.l.o.g.~let both $H_S,H_B$ be diagonal in the standard basis.
When expressed mathematically the ``mixing property'' in question reads 
\unexpanded{$\langle e_k,\Phi_{T,m}(H_B,U)(|e_i\rangle\langle e_j|)e_l\rangle\neq 0$}. This readily implies the existence of indices $\alpha,\beta$ such that neither \unexpanded{$\langle e_k\otimes e_\alpha,U(e_i\otimes e_\beta)\rangle$} nor \unexpanded{$\langle e_l\otimes e_\alpha,U(e_j\otimes e_\beta)\rangle$} vanish.
But \unexpanded{$\langle e_k\otimes e_\alpha,[U,H_S\otimes\mathbbm1+\mathbbm1\otimes H_B](e_i\otimes e_\beta)\rangle=(E_i+E_\beta'-E_k-E_\alpha')\langle e_k\otimes e_\alpha,U(e_i\otimes e_\beta)\rangle$} does always vanish due to $U$ being energy-conserving; therefore $E_i-E_k=E_\alpha'-E_\beta'$ (and similarly $E_j-E_l=E_\alpha'-E_\beta'$ for the second term).
}
$E_i-E_k=E_j-E_l\in\sigma(\operatorname{ad}_{H_B})$. Equivalently, it is necessary that the transitions corresponding to $\rho_{ij}$ and $\rho_{kl}$ coincide (that is, $E_i-E_j=E_k-E_l$) and that the difference between these transitions appears as a difference in $H_B$ (i.e.~$E_i-E_k\in\sigma(\operatorname{ad}_{H_B})$).
Be aware that simply scaling entries $\rho_{ij}$ using $\mathsf{TO}$ is independent of either of these notions.
\end{remark}

Now the concept of resonance allows us to restrict the set of bath Hamiltonians necessary for describing the set of thermal operations. This as well as some fundamental topological properties of $\mathsf{TO}$ are summarized in the following:
\begin{proposition}\label{prop_1}
Let $H_S\in{i}\mathfrak u(n)$ and $T\in(0,\infty]$ be given, and let $\overline{(\cdot)}$ henceforth denote the closure. The following statements hold:
\begin{itemize}
\item[(i)] $\mathsf{TO}(H_S,T)$ is a bounded, path-connected semigroup with identity.
\item[(ii)] $\overline{\mathsf{TO}(H_S,T)}$ is a convex, compact semigroup with identity.
\item[(iii)] $\overline{\mathsf{TO}(H_S,T)}$ is a subset of all \textsc{cptp} maps with common fixed point $e^{-H_S/T}$.
\item[(iv)] For describing the closure of all thermal operations it suffices to only consider bath-Hamiltonians with resonant spectrum w.r.t.~$H_S$, that is,
\begin{equation}\label{eq:TO_connected_spectrum}
\overline{\mathsf{TO}(H_S,T)}=\overline{\bigcup_{m\in\mathbb N}\Big\{ \Phi_{T,m}(H_B,e^{iH_\mathsf{tot}}):\ \substack{
H_B=\operatorname{diag}(E_1,\ldots,E_m)\text{ with } E_1\leq\ldots\leq E_m,\\
H_B\text{ has resonant spectrum w.r.t.~}H_S\\
H_\mathsf{tot}\in{i}\mathfrak{u}(mn),
[H_\mathsf{tot},H_S\otimes\mathbbm{1}_B+\mathbbm{1}\otimes H_B]=0
}\Big\}} \,.
\end{equation}
\end{itemize}
\end{proposition}
\noindent The only non-trivial statements in this lemma are convexity in (ii) as first 
shown in Appendix C in \cite{Lostaglio15}, and (iv) (respectively Eq.~\eqref{eq:TO_connected_spectrum}). The intuition for the latter is as 
follows: Given some bath-Hamiltonian with non-resonant spectrum 
(w.r.t.~$H_S$) one can partition said spectrum into different components which can 
not interact with each other because of energy-conservation. This implies that the 
full unitary is of similar block structure and that the associated thermal operation 
can be written as a convex combination of thermal operations generated by
bath-Hamiltonians with resonant spectra (i.e.~the connected components). The full 
proof is given in Appendix \ref{app_proof_lemma_1}.

Working with $\overline{\mathsf{TO}}$ instead of $\mathsf{TO}$ is advantageous for two reasons: On the one hand it
%(to my knowledge) 
is unknown whether $\mathsf{TO}$ itself is convex (we will answer this in the affirmative for qubits later on). Indeed a necessary step in showing that statement (iv) from Proposition \ref{prop_1} holds without the closure, i.e.~the somehow ``intuitive'' result that bath-Hamiltonians with non-resonant spectrum are not needed for describing $\mathsf{TO}$, would be a proof of convexity of $\mathsf{TO}$ which continues to hold when considering the right-hand side of \eqref{eq:TO_connected_spectrum} (without the closure).

On the other hand, more gravely, $\mathsf{TO}$ is not closed. The simplest counter-example corresponds to transforming an energy eigenstate; this is not thermally allowed \cite{Lostaglio15_2} meaning the map
$$
\begin{pmatrix} a_{11}&a_{12}\\a_{21}&a_{22} \end{pmatrix}\mapsto\begin{pmatrix}(1-e^{-1/T})a_{11}+a_{22} &0\\0&e^{-1/T}a_{11} \end{pmatrix}
%\not\in {\mathsf{TO}(-\sigma_z,T)}
$$
is not in ${\mathsf{TO}(-\sigma_z,T)}$. Yet, this map can be approximated arbitrarily well by thermal operations so it is an element of $\overline{\mathsf{TO}(-\sigma_z,T)}$, cf.~Section \ref{sec_ento}. One way to fix this is to use baths of infinite size, e.g., single-mode bosonic baths (cf.~Lemma 1 in \cite{Lostaglio18}, as well as \cite{Ding19}).
However, while such baths are able to implement the above operation,
using them to implement full dephasing (even approximately) becomes impossible once the temperature is too low, cf.~Theorem \ref{thm_equiv_qubit} (iv).

Either way the closure guarantees a ``reasonable mathematical structure''. This can also be motivated from an application or engineering point of view: At least for some questions (e.g.,~reachability in control theory) it does not matter whether one can implement an operation exactly or ``only'' with arbitrary precision. However, figuring out which results continue to hold after waiving the closures could reveal more of the structure of the thermal operations (cf.~also Section \ref{sec_concl}). 

An important consequence of Proposition \ref{prop_1} (iv) is that if $H_S$ is a spin-Hamiltonian, i.e.~$H_S$ has equidistant eigenvalues, then one can reduce the set of bath-Hamiltonians used in the definition of $\mathsf{TO}$ to spin-Hamiltonians ``of the same structure'' without changing the set (after taking the closure).
This continues to hold even if $H_S$ only is of spin-form ``up to potential gaps''. The precise statement---a weaker version of which first appeared in Lemma 1 of \cite{Ding21}---%-- at least formally -- 
reads as follows:

\begin{corollary}\label{prop_to_spin}
Given $H_S\in{i}\mathfrak{u}(n)$ assume there exist $E_1\in\mathbb R$ and an energy gap $\Delta E>0$ such that $\sigma(H_S)\subseteq\{E_1+j\Delta E:j\in\mathbb N_0\}$. Define $\mathsf{TO}^{\Delta E}_\mathsf{Spin}(H_S,T)$ as the collection of all thermal operations where $H_B$ is any spin-Hamiltonian with the same gap $\Delta E$ as $H_S$, that is,
\begin{align*}
\mathsf{TO}^{\Delta E}_\mathsf{Spin}(H_S,T):=\bigcup_{m\in\mathbb N}\Big\{ \Phi_{T,\sum_{j=1}^m\beta_j}(H_B,e^{iH_\mathsf{tot}}):\ \substack{H_B=\bigoplus_{j=1}^m j\Delta E\,\mathbbm1_{\beta_j}
\text{ with }\beta_1,\ldots,\beta_m\in\mathbb N\\
%\text{where }\sum_{i=1}^m\alpha_i>0
H_\mathsf{tot}\in{i}\mathfrak{u}(n\sum_j\beta_j),[H_\mathsf{tot},H_S\otimes\mathbbm{1}_B+\mathbbm{1}\otimes H_B]=0 }\Big\}
\end{align*}
for all $T>0$, as well as
\begin{align*}
\mathsf{TO}^{\Delta E}_\mathsf{Spin}(H_S,\infty):=\bigcup_{m\in\mathbb N}\Big\{ \Phi_{1,\sum_{j=1}^m\beta_j}(\mathbbm1,e^{iH_\mathsf{tot}}):\ \substack{H_B=\bigoplus_{j=1}^m j\Delta E\,\mathbbm1_{\beta_j}
\text{ with }\beta_1,\ldots,\beta_m\in\mathbb N\\
%\text{where }\sum_{i=1}^m\alpha_i>0
H_\mathsf{tot}\in{i}\mathfrak{u}(n\sum_j\beta_j),[H_\mathsf{tot},H_S\otimes\mathbbm{1}_B+\mathbbm{1}\otimes H_B]=0 }\Big\}\,.
\end{align*}
One finds
\begin{equation}\label{eq:equiv_TO_spin}
\begin{split}
\forall_{T>0}\quad \overline{\mathsf{TO}^{\Delta E}_\mathsf{Spin}(H_S,T)}&=\overline{\bigcup_{m\in\mathbb N}\Big\{ \Phi_{T,\sum_{j=1}^m\beta_j}(H_B,e^{iH_\mathsf{tot}}):\ \substack{H_B=\bigoplus_{j=1}^m j\Delta E\,\mathbbm1_{\beta_j}
\text{ with }\beta_1,\ldots,\beta_m\in\mathbb N_0,
%\text{where }
\sum_{i=1}^m\beta_i>0\\
H_\mathsf{tot}\in{i}\mathfrak{u}(n\sum_j\beta_j),
[H_\mathsf{tot},H_S\otimes\mathbbm{1}_B+\mathbbm{1}\otimes H_B]=0 }\Big\}}\,,\\
\overline{\mathsf{TO}^{\Delta E}_\mathsf{Spin}(H_S,\infty)}&=\overline{\bigcup_{m\in\mathbb N}\Big\{ \Phi_{1,\sum_{j=1}^m\beta_j}(\mathbbm1,e^{iH_\mathsf{tot}}):\ \substack{H_B=\bigoplus_{j=1}^m j\Delta E\,\mathbbm1_{\beta_j}
\text{ with }\beta_1,\ldots,\beta_m\in\mathbb N_0,
%\text{where }
\sum_{i=1}^m\beta_i>0\\
H_\mathsf{tot}\in{i}\mathfrak{u}(n\sum_j\beta_j),
[H_\mathsf{tot},H_S\otimes\mathbbm{1}_B+\mathbbm{1}\otimes H_B]=0 }\Big\}}\,,
\end{split}
\end{equation}
and
$
\overline{\mathsf{TO}(H_S,T)}=\overline{ \mathsf{TO}^{\Delta E}_\mathsf{Spin}(H_S,T) }$ for all $T\in(0,\infty]$.
\end{corollary}
While this result---for the most part---is a corollary of Proposition \ref{prop_1} we nevertheless present a proof in Appendix \ref{app_proof_prop_to_spin}.
This immediately yields the following:
\begin{corollary}\label{coro_spin_rational}
Given $T\in(0,\infty]$, if $H_S\in{i}\mathfrak{u}(n)$ has rational Bohr spectrum up to a global constant---i.e.~there exists real $r>0$ such that $\sigma(\operatorname{ad}_{H_S})\in r\mathbb Z$---then
$
\overline{\mathsf{TO}(H_S,T)}=\overline{ \mathsf{TO}^{r}_\mathsf{Spin}(H_S,T) }
$ from \eqref{eq:equiv_TO_spin}.
\end{corollary}
This is a direct application of Corollary \ref{prop_to_spin} because if $H_S$ (with eigenvalues $E_1\leq\ldots\leq E_n$) has rational Bohr spectrum, then
\begin{align*}
\sigma(H_S)=\Big\{ E_1+r\cdot\frac{E_j-E_1}r:j=1,\ldots,n\Big\}
\subseteq \{E_1+rj:j\in\mathbb N_0\}\,.
\end{align*}
There does not seem to be an obvious generalization of the previous two results to arbitrary Hamiltonians. For this consider $H_S=\operatorname{diag}(0,1,\sqrt{2})$ as system and $H_B=\operatorname{diag}(0,\sqrt{2}-1,1)$ as bath Hamiltonian; then $H_B$
does not have rational Bohr spectrum up to any constant, yet $H_B$ has resonant spectrum w.r.t.~$H_S$ so there is no ``obvious'' decomposition as in Proposition \ref{prop_1} / Corollary \ref{prop_to_spin} into baths $H_{B,1},H_{B,2}$ of spin form.

However one may ask whether the (somewhat unphysical) condition of the Bohr spectrum being rational up to a constant can be waived if one only demands approximation instead of equality in Corollary \ref{coro_spin_rational}. This essentially boils down to whether $\mathsf{TO}$ is continuous in the system's Hamiltonian.
%, and we will come back to this question in Section \ref{sec_continuity}.
%; there is a relation between the first and the third as well as the first and the second energy level of $H_B$ with respect to the Bohr spectrum of $H_S$ (cf.~Figure \ref{fig2}). Yet if a thermal operation w.r.t.~$H_B$ can be decomposed into a convex combination of thermal operations w.r.t.~spin-Hamiltonians compatible with the Bohr spectrum of $H_S$ (e.g., $H_{B,1}=\operatorname{diag}(0,1)$, $H_{B,2}=\operatorname{diag}(1,\sqrt{2})$ because those are the two realized transitions) then this at least seems not obvious.
%%
%%
%\begin{figure}[!htb]
%\centering
%\includegraphics[width=0.8\textwidth]{counterex_bohr.pdf}
%\caption{Visualization of the spectrum of $H_S\otimes\mathbbm1_B+\mathbbm1\otimes H_B$ for $H_S=\operatorname{diag}(0,1,\sqrt{2})$ as well as $H_B=\operatorname{diag}(0,\sqrt{2}-1,1)$. .. }\label{fig2}
%\end{figure}%${}$
%
%
%
%corollary: for all $H_S$ one has a strict inclusion (even of the action)? e.g., strong regularity of $H_S$ will not resolve the fact that there are $d$-stochastic matrices which cannot be generated
\section{Thermal Operations and Continuity -- Or Lack Thereof}\label{sec_continuity}
A natural question from a physics perspective
% -- also keeping Corollary 
%\ref{coro_spin_rational} in mind -- 
is how robust the set of thermal operations is
to small changes in temperature or in the energy levels of the system. This 
question already has a partial answer for inhomogeneous reservoirs and diagonal 
states from the perspective of work generation and $\alpha$-free energies 
\cite{Shu19}. 
Others have also studied characterizing approximate thermodynamic state 
transitions via smoothed generalized free energies \cite{vanderMeer17}, as well as 
the general effect of imperfections (such as finite-time and finite-size) on work 
extraction and the second law \cite{Baeumer19,Richens18}.
However it seems that a rigorous study of how the \textit{set} of all thermal 
operations depends on parameters such as the temperature or (the spectrum of) the 
system's Hamiltonian is still amiss.

For this we introduce a notion of distance between sets of quantum maps. 
One way to do this is to use the Hausdorff metric $\delta$ (here w.r.t.~$\|\cdot\|_{1\to 1}$, that is, the usual operator norm if domain and range are equipped with the trace norm $\|\cdot\|_1:=\operatorname{tr}(\sqrt{(\cdot)^*(\cdot)})$) which -- given non-empty, compact sets $A,B\subset\mathcal L(\mathbb C^{n\times n})$ -- is defined to be
\begin{equation}\label{eq:def_delta}
\delta(A,B):=\max\big\{\max_{S_1\in A}\min_{S_2\in B} \|S_1-S_2\|_{1\to 1},\max_{S_2\in B}\min_{S_1\in A} \|S_1-S_2\|_{1\to 1}\big\}\,.
\end{equation}
Here and henceforth, given any vector space $V$ we write $\mathcal L(V)$ for the collection of all linear maps $:V\to V$.
The expression \eqref{eq:def_delta} indeed is a metric on $\mathcal P_c\big(\mathcal L(\mathbb C^{n\times n}))$, the latter denoting the space of all non-empty compact 
subsets of $\mathcal L(\mathbb C^{n\times n})$, cf.~§21.VII in \cite{Kuratowski66}. In 
particular this allows one to define a distance between any non-empty sets 
$A,B\subset\textsc{cptp}(n)$ via $\delta(\overline{A},\overline{B})$. 

Based on this definition we will show that whenever $H_S\in{i}\mathfrak u(n)$ has 
degenerate Bohr spectrum (Supplementary Note 2 in \cite{Cwiklinski15}) -- 
i.e.~$\operatorname{ad}_{H_S}$ has less than $n^2-n+1$ different eigenvalues --
then the map
\begin{equation}\label{eq:TO_n}
\begin{split}
\mathsf{TO}_n:{i}\mathfrak{u}(n)\times(0,\infty]&\to \big(\mathcal P_c\big(\mathcal L(\mathbb C^{n\times n}),\|\cdot\|_{1\to 1}\big),\delta\big)\\
(H,T)&\mapsto\overline{\mathsf{TO}(H,T)}
\end{split}
\end{equation}
is discontinuous in $(H_S,T)$ for all temperatures $T\in(0,\infty]$.
%Here denotes the collection 
%of all non-empty compact subsets of $\mathcal L(\mathbb C^{n\times n})$.
%
Note that $H_S$ has degenerate Bohr spectrum iff either $H_S$ itself is degenerate 
($|\sigma(H_S)|<n$), or -- assuming $\sigma(H_S)=\{E_1,\ldots,E_n\}$ for some 
$E_1<\ldots<E_n$ -- if some of the energy transitions which $H_S$ admits coincide, i.e.~if the map $(j,k)\mapsto E_j-E_k$ with domain
$\{(j,k):1\leq j<k\leq n\}$ is not injective. With this in mind we will present two examples which 
illustrate how the map \eqref{eq:TO_n} can fail to be continuous:

\begin{example}\label{ex_discont}
\item[(i)] First we consider the simplest case of a degenerate system's Hamiltonian, 
that is, $n=2$ and $H_S=0$.
%Then for arbitrary $T\in(0,\infty]$,
%$\overline{\mathsf{TO}(0,T)}$ equals the set of all \textsc{cptp} maps on
%$\mathbb C^{2\times 2}$ for which $\mathbbm1$ is a fixed point (denoted b
% $\mathsf{cptp}_{\mathbbm 1}(2)$ for now):
%\begin{align*}
%\overline{\mathsf{TO}(0,T)}\subseteq \mathsf{cptp}_{\mathbbm 1}(2)&=\overline{ \operatorname{conv}\{ \operatorname{Ad}_U:U\in\mathsf U(2) \} }\\
%&\subseteq \overline{ \operatorname{conv}\mathsf{TO}(0,T) }\subseteq\overline{ \operatorname{conv}\overline{\mathsf{TO}(0,T) } }=\operatorname{conv}\overline{\mathsf{TO}(0,T) }=\overline{\mathsf{TO}(0,T) }\,.
%\end{align*}
%Here we -- in logical order -- used the fixed point property of $\mathsf{TO}$ (Lemma 
%\ref{prop_1} (iii)), the fact that every unital qubit channel is a convex combination of 
%unitary channels \cite{LS93}, the fact that $\mathsf{TO}(0,T)$ contains all unitary 
%channels (definition of $\mathsf{TO}$ for $m=1$, $H_S=0$), compactness being 
%preserved in finite-dimensional normed spaces when taking the convex 
%hull\cite[Thm.~3.20]{Rudin91}, and finally convexity of 
%$\overline{\mathsf{TO}}$ (Lemma \ref{prop_1} (ii)).
Given arbitrary $T\in(0,\infty]$ we will show that
$\delta(\overline{\mathsf{TO}(0,T)},\overline{ \mathsf{TO}(\operatorname{diag}(0,
\varepsilon),T)})\geq 1$ for all $\varepsilon>0$ which clearly violates continuity 
of $\mathsf{TO}_2$ in $(0,T)$. The reason for this, 
roughly speaking, is that no thermal operation corresponding to a
non-degenerate Hamiltonian can 
mix diagonal and off-diagonal elements. This is prohibited by the known fact 
that it has to commute with $\operatorname{ad}_{\operatorname{diag}(0,\varepsilon)}
$.

While it is easy to see that $\overline{\mathsf{TO}(0,T)}$ equals the set of all unital qubit maps (that is, all \textsc{cptp} maps on $\mathbb C^{2\times 2}$ for which $\mathbbm1$ is a fixed point, cf.~Appendix \ref{app_proof_thm_equiv_qubit} and related footnotes) for our purposes it suffices to define a map $S\in\mathcal L(\mathbb C^{2\times 2})$ via $S:=\Phi_{T,2}(H_S,U)$ where 
$$
U:=\frac{1}{\sqrt2}\begin{pmatrix}
1&0&1&0\\0&1&0&-1\\1&0&-1&0\\0&1&0&1
\end{pmatrix}\in\mathsf U(4)\,.
$$
Indeed $S\in\mathsf{TO}(0,T)$ for all $T\in(0,\infty]$ and one readily verifies that the action of $S$ is given by
$$
\begin{pmatrix}
a_{11}&a_{12}\\a_{21}&a_{22}
\end{pmatrix}\mapsto\frac12\begin{pmatrix}
a_{11}+a_{22}&a_{11}-a_{22}\\a_{11}-a_{22}&a_{11}+a_{22}
\end{pmatrix}\,.
$$
%One readily verifies $S\in \mathsf{cptp}_{\mathbbm 1}(2)=\overline{\mathsf{TO}(0,T)}$.
Thus for all $\varepsilon>0$, $T\in(0,\infty]$ we compute
\begin{align*}
\delta(\overline{\mathsf{TO}(0,T)},\overline{ \mathsf{TO}(\operatorname{diag}(0,\varepsilon),T)})&\geq \min_{\tilde S\in \overline{ \mathsf{TO}(\operatorname{diag}(0,\varepsilon),T)}}\sup_{A\in\mathbb C^{2\times 2},\|A\|_1=1} \|S(A)-\tilde S(A)\|_{1}\\
&\geq \min_{\tilde S\in \overline{ \mathsf{TO}(\operatorname{diag}(0,\varepsilon),T)}}\|S(|e_1\rangle\langle e_1|)-\tilde S(|e_1\rangle\langle e_1|)\|_1\\
&=\min_{\lambda\in[0,1]} \Big\| \frac12\begin{pmatrix}1&1\\1&1\end{pmatrix} -\begin{pmatrix}
1-\lambda e^{-\varepsilon/T}&0\\0&\lambda e^{-\varepsilon/T}
\end{pmatrix} \Big\|_1\\
&\geq \min_{\lambda'\in\mathbb R}\Big\|\begin{pmatrix}\lambda'&\frac12\\\frac12&-\lambda'\end{pmatrix}\Big\|_1=\min_{\lambda'\in\mathbb R} \sqrt{4(\lambda')^2+1}=1\,.
\end{align*}
In the third step we used Eq.~\eqref{eq:EnTO_canonical} as well as Thm.~\ref{thm_equiv_qubit} (i).
%, and in the fourth step we applied the dual space formula for the trace norm \cite[Prop.~16.24]{MeiseVogt97en}.
\item[(ii)] Having dealt with degenerate Hamiltonians let us now look at the other possible case: Hamiltonians which are non-degenerate but have degenerate Bohr spectrum. This can only occur in $3$ or more dimensions so consider $H_S=\operatorname{diag}(0,1,2)$. 
Given arbitrary $T\in(0,\infty]$ we will show that
$\delta(\overline{\mathsf{TO}(H_S,T)},\overline{ \mathsf{TO}(H_S+\varepsilon|e_3\rangle\langle e_3|, T)})\geq \frac23$ for all $\varepsilon>0$ which again violates continuity -- the reason for this being similar to the reason from (i).
Choose $H_B:=H_S$ and define
$$
U:=\begin{pmatrix}
1&0&0&0&0&0&0&0&0\\
0&0&0&1&0&0&0&0&0\\
0&0&0&0&1&0&0&0&0\\
0&1&0&0&0&0&0&0&0\\
0&0&0&0&0&0&1&0&0\\
0&0&0&0&0&0&0&1&0\\
0&0&1&0&0&0&0&0&0\\
0&0&0&0&0&1&0&0&0\\
0&0&0&0&0&0&0&0&1
\end{pmatrix}
$$
which corresponds to the permutation (in cycle notation) $(1)(2\,4)(3\,5\,7)(6\,8)(9)$. Therefore $U$ is unitary and satisfies the stabilizer condition $U(H_S\otimes\mathbbm{1}_B+\mathbbm{1}\otimes H_B)U^*=H_S\otimes\mathbbm{1}_B+\mathbbm{1}\otimes H_B$ because matching diagonal entries of $H_S\otimes\mathbbm{1}_B+\mathbbm{1}\otimes H_B$ precisely correspond to the cycles of $U$. 
With this 
$S:=\Phi_{T,3}(H_S,U)$
acts on any $A\in\mathbb C^{3\times 3}$ as follows:
\begin{align*}
\frac{1}{1+e^{-1/T}+e^{-2/T}}
\begin{pmatrix}
a_{11}+a_{22}(1+e^{-1/T})&a_{23}(1+e^{-1/T})&0\\
a_{32}(1+e^{-1/T})&a_{11}e^{-1/T}+a_{33}(1+e^{-1/T})&0\\
0&0&(a_{11}+a_{22}+a_{33})e^{-2/T}
\end{pmatrix}\,.
\end{align*}
%as is verified by direct computation. 
One for all $\varepsilon>0$, $T\in(0,\infty]$ finds (similar to (i))
\begin{align*}
\delta(\overline{\mathsf{TO}(H_S,T)},\overline{ \mathsf{TO}(H_S+\varepsilon|e_3\rangle\langle e_3|, T)})&\geq \min_{\tilde S\in \overline{ \mathsf{TO}(H_S+\varepsilon|e_3\rangle\langle e_3|, T)}}\|S(|e_2\rangle\langle e_3|)-\tilde S(|e_2\rangle\langle e_3|)\|_1\,.
\end{align*}
Now $H_S+\varepsilon|e_3\rangle\langle e_3|$ has non-degenerate Bohr spectrum for all $\varepsilon>0$, meaning the only thing $\overline{ \mathsf{TO}(H_S+\varepsilon|e_3\rangle\langle e_3|, T)}$ can do to off-diagonal entries is scale them by a factor $\gamma\in\mathbb C$, $|\gamma|\leq 1$ -- this follows from the fact that every thermal operation (w.r.t.~$H_S+\varepsilon|e_3\rangle\langle e_3|$ and any $T$) has to commute with $\operatorname{ad}_{H_S+\varepsilon|e_3\rangle\langle e_3|}\ $\cite{Lostaglio15_2}.
With this one obtains the lower bound
\begin{align*}
\delta(\overline{\mathsf{TO}(H_S,T)},\overline{ \mathsf{TO}(H_S+\varepsilon|e_3\rangle\langle e_3|, T)})&\geq \min_{|\gamma|\leq 1}\Big\|\frac{1+e^{-1/T}}{1+e^{-1/T}+e^{-2/T}}|e_1\rangle\langle e_2|-\gamma|e_2\rangle\langle e_3|\Big\|_1\\
&= \min_{|\gamma|\leq 1}\Big(\frac{1+e^{-1/T}}{1+e^{-1/T}+e^{-2/T}}+|\gamma|\Big)\\
&= 1-\frac{e^{-2/T}}{1+e^{-1/T}+e^{-2/T}} \geq1-\sup_{T>0}\frac{1}{e^{2/T}+e^{1/T}+1}=\frac23\,.
\end{align*}
\end{example}
\noindent It is not difficult to generalize these examples to any Hamiltonians with degenerate Bohr spectrum in arbitrary dimensions.

%POTENTIALLY WEAKEN THEOREM if proof does not work out: ``while one has continuity in $T$, this breaks down as soon as one perturbs (eigenvalues of) $H_S$''
%explain how it can be generalized to arbitrary $H_S$ with non-degen.~Bohr spectrum

The reason for discontinuity in either example was the condition $[S,\operatorname{ad}_{H_S}]=0$ for all $S\in\overline{\mathsf{TO}}$ which comes solely from $U(H_S\otimes\mathbbm{1}_B+\mathbbm{1}\otimes H_B)U^*=H_S\otimes\mathbbm{1}_B+\mathbbm{1}\otimes H_B $. This suggests two things: first, to restore the (physically reasonable) requirement of continuity one has to somehow relax or alter this condition -- more on this in Section \ref{sec_concl}. Second, as the temperature does not appear here it seems reasonable to conjecture the following:
\begin{conjecture}\label{prop_cont_in_T}
%There exists a continuous function $c:[0,\infty)\times[0,\infty)\to(0,\infty)$ such that 
For all $H_S\in{i}\mathfrak u(n)$
%one has
%\begin{align*}
%\delta\big(\overline{\mathsf{TO}(H_S,T)},\overline{\mathsf{TO}(H_S,T')}\big)\leq n\Big( 
%%\max\Big\{ \Big|1-\sqrt{\frac{T}{T'}}\Big| ,
%\frac{|\frac1T-\frac1{T'}|}{\frac1T+\frac1{T'}}
%%\Big\}
%+\Big|\ln\frac{1}{T}-\ln\frac{1}{T'}\Big|\Big)\,.
%\end{align*}
%In particular 
the map $T\mapsto\overline{\mathsf{TO}(H_S,T)}$ is continuous
if the domain $(0,\infty]$ is equipped with the
%``inverse temperature'' 
metric $d_{-1}(T,T'):=|\frac1T-\frac1{T'}|$, and
%on $(0,\infty]$ if
the co-domain is equipped with the Hausdorff metric $\delta$ w.r.t.~$(\mathcal L(\mathbb C^{n\times n}),\|\cdot\|_{1\to 1})$.
\end{conjecture}
For a simple yet (so far) unsuccessful attempt to prove this, see Appendix \ref{proof_att_conj_cont_T}.
Either way this property would be necessary for the role of the temperature in the definition of thermal operations to be ``correct'' in the sense that it accurately models the behavior of real physical systems.

As a final remark the case $T=0$ is excluded from the above continuity considerations for two reasons: first, the concept of zero temperature and achieving it with finite resources (e.g., time, heat baths) is problematic in the classical \cite{Nernst12,Loebl60} as well as the quantum case (at least for qubits, cf.~Lemma 9 in \cite{Scharlau18}). 
Second, letting the temperature tend to zero reveals a lack of continuity already in the classical case, cf.~Appendix D, Example 3 in \cite{vomEnde22}.
More precisely, there exist classical states (probability vectors) $x$ such that the map $T\mapsto\{Ax:A\text{ Gibbs-stochastic}$\footnote{
Recall that for $n\in\mathbb N$ and some positive $d\in\mathbb R_{++}^n$ the set of Gibbs-stochastic (or ``$d$-stochastic'' in the mathematics literature, cf.~Ch.~14.B in \cite{MarshallOlkin}) $n\times n$ matrices is defined as the set of all $A\in\mathbb R^{n\times n}$ with non-negative entries such that $Ad=d$ and $\mathbbm{e}^\top A=\mathbbm{e}^\top$ where $\mathbbm{e}:=(1,\ldots,1)^\top$.
}${}^\text{,}$\cite{MarshallOlkin} $\text{ w.r.t.~}H_S,T\}$ is discontinuous in $T=0$.
\section{The Qubit Case: Overview, Semigroup Representation, and Visualization}\label{sec_ento}

Two core features of thermal operations are preservation of the Gibbs state and the covariance law (in generator form) $[S,\operatorname{ad}_{H_S}]=0$ for all $S\in \overline{\mathsf{TO}(H_S,T)}$ \,\!\cite{Lostaglio15_2}. This motivates the following definition \cite{Cwiklinski15}:
%\begin{definition}
Given $H_S\in{i}\mathfrak u(n)$ and some $T>0$ the set of all covariant Gibbs-preserving maps is defined to be
\begin{align*}
\mathsf{EnTO}(H_S,T):=\{S\in \textsc{cptp}(n)\,:\,S(e^{-H_S/T})=e^{-H_S/T}\,\wedge\,[S,\operatorname{ad}_{H_S}]=0\}%\\
%=&\{T\in Q_{\frac{e^{-H_S/T}}{\operatorname{tr}(e^{-H_S/T})}}(n):[T,\operatorname{ad}_{H_S}]=0\}
\end{align*}
where $\mathsf{EnTO}$ is short for ``enhanced thermal operations''. This definition naturally extends to $T=\infty$ by replacing the fixed point $e^{-H_S/T}$ by $\mathbbm1_n$.
%\end{definition}
It is straightforward to see that for all $H_S\in{i}\mathfrak u(n)$, $T\in(0,\infty]$, $\mathsf{EnTO}$ is a convex, compact semigroup with identity, and $\mathsf{EnTO}$ satisfies the same transformation rules as $\mathsf{TO}$ and $\overline{\mathsf{TO}}$ (Lemma \ref{lemma_3}).
%$\mathsf{EnTO}(\lambda H_S+\mu\mathbbm1,\lambda T)=\mathsf{EnTO}(H_S,T)$ as well as $\mathsf{EnTO}(\operatorname{Ad}_{U}(H_S),T)=\operatorname{Ad}_{U}\circ\;\mathsf{EnTO}(H_S,T)\circ\operatorname{Ad}_{U^*}$ for all $U\in\mathbb C^{n\times n}$ unitary, $\lambda>0$, $\mu\in\mathbb R$.
Moreover $\overline{\mathsf{TO}}\subseteq\mathsf{EnTO}$, and the action of $\overline{\mathsf{TO}}$ and $\mathsf{EnTO}$ on any classical state $\rho$ (i.e.~on any state with $[H_S,\rho]=0$) even coincides, cf.~Sec.~3 in \cite{Shiraishi20}. 

A set of necessary and sufficient (implicit) conditions for state conversion under enhanced thermal operations was given by Gour et al.~\cite{Gour18}.
However for general systems, $\mathsf{TO}$ and $\mathsf{EnTO}$ do not agree, even in closure and when restricted to their respective action: Choosing $H_S=\operatorname{diag}(0,1,2)$ there exist temperature $T>0$, quantum states $\rho,\rho'$, and $S\in\mathsf{EnTO}(H_S,T)$ such that $S(\rho)=\rho'$ but
$\rho'\not\in\overline{\mathsf{TO}(H_S,T)}(\rho)\;$\cite{Ding21}.
This, however, is only true beyond two dimensions because for qubits it is known that the two sets coincide.
Before we review the many results on qubit thermal operations let us investigate the basic structure of $\mathsf{TO}$ in two dimensions; this will simplify things later on.

The qubit case is particularly nice because there $\mathsf{EnTO}$ is characterized by three real parameters (i.e.~one real and one complex number) so in particular we can visualize it. Indeed given $H_S\in{i}\mathfrak{u}(2)$ non-degenerate (i.e.~$H_S=\sum_{i=1}^2E_i|g_i\rangle\langle g_i|$ with $E_1<E_2$ for some orthonormal basis $\{g_1,g_2\}$ of $\mathbb C^2$) as well as $T\in(0,\infty]$, one finds that a linear map $S:\mathbb C^{2\times 2}\to\mathbb C^{2\times 2}$ is in $\mathsf{EnTO}(H_S,T)$ if and only if there exist $\lambda\in[0,1]$, $r\in\big[0,\sqrt{ (1-\lambda)(1-\lambda e^{- \Delta E /T}) }\big]$, and $\phi\in[-\pi,\pi)$ such that the Choi matrix\cite{Choi75} of $S$ (w.r.t.~$\{g_1,g_2\}$) reads
\begin{equation}\label{eq:EnTO_canonical}
\begin{pmatrix} 1-\lambda e^{- \Delta E /T} &0&0&re^{i\phi} \\
0&\lambda e^{- \Delta E /T}&0&0\\
0&0&\lambda&0\\
re^{-i\phi}&0&0&1-\lambda \end{pmatrix}\,.
\end{equation}
Here $\Delta E:=E_2-E_1>0$ and, if $T=\infty$, then $e^{-\Delta E/T}$ gets replaced by $1$.
%where $\{g_1,g_2\}$ is an orthonormal basis of $\mathbb C^2$ such that $H_S=E_1|g_1\rangle\langle g_1|+E_2|g_2\rangle\langle g_2|$ with $E_1<E_2$, the matrices $G_1:=|g_1\rangle\langle g_1|$, $G_2:=|g_2\rangle\langle g_2|$, $G_3:=|g_1\rangle\langle g_2|$, $G_4:=|g_2\rangle\langle g_1|$ form an orthonormal basis of $(\mathbb C^{2\times 2},\langle\cdot,\cdot\rangle_\mathsf{HS})$, and $d:=e^{- \Delta E /T}\in (0,1)$ is the ratio of Gibbs weights corresponding to $H_S,T$.
The basic structure of \eqref{eq:EnTO_canonical} -- meaning the position of the zeros -- is solely due to $[S,\operatorname{ad}_{H_S}]=0$, while preservation of the Gibbs state is encoded in the diagonal action being a Gibbs-stochastic $2\times 2$ matrix
where $d=(e^{-E_1/T},e^{-E_2/T})$. In two dimensions the latter set is well known to equal
$$
\operatorname{conv}{ \Big\{
\begin{pmatrix}
 1-\frac{d_2}{d_1} & 1 \\
\frac{d_2}{d_1} & 0
\end{pmatrix},\begin{pmatrix}
1&0\\0&1
\end{pmatrix}\Big\}
 }=
\operatorname{conv}{ \Big\{
\begin{pmatrix}
 1-e^{-\Delta E/T} & 1 \\
e^{-\Delta E/T} & 0
\end{pmatrix},\begin{pmatrix}
1&0\\0&1
\end{pmatrix}\Big\}
 }
$$
so it is characterized by one parameter $\lambda\in[0,1]$. Finally the scaling of $S$ on the off-diagonal is only restricted by complete positivity (so positive semi-definiteness of \eqref{eq:EnTO_canonical}, cf.~again \cite{Choi75}) which is equivalent to $|r|^2\leq(1-\lambda)( 1-\lambda e^{- \Delta E /T} )$.
This allows us to define the linear map
\begin{equation}\label{eq:rep_TO_map}
\begin{split}
\Psi_T:\mathcal L(\mathbb C^{2\times 2})&\to\mathbb R\times\mathbb C\\
S&\mapsto\begin{pmatrix}
\langle g_1,S(|g_2\rangle\langle g_2|)g_1\rangle\\
\langle g_1,S(|g_1\rangle\langle g_2|)g_2\rangle
\end{pmatrix}
\end{split}
\end{equation}
which maps \eqref{eq:EnTO_canonical} to $(\lambda,re^{i\phi})$.
This becomes a faithful semigroup representation if the domain of $\Psi_T$ is restricted to $D(\Psi_T):=$ ``all linear maps on $\mathbb C^{2\times 2}$ whose Choi matrix is given by \eqref{eq:EnTO_canonical} for some $\lambda\in\mathbb R$, $r\geq 0$, $\phi\in[-\pi,\pi)$'' and if the codomain of $\Psi_T$ is equipped with the associative operation
\begin{align*}
\circ_T:(\mathbb R\times\mathbb C)\times(\mathbb R\times\mathbb C)&\to (\mathbb R\times\mathbb C)\\
\Big(
\begin{pmatrix}
\lambda_1\\c_1
\end{pmatrix},
\begin{pmatrix}
\lambda_2\\c_2
\end{pmatrix}\Big)&\mapsto\begin{pmatrix}
\lambda_1+\lambda_2-\lambda_1\lambda_2(1+e^{-1/T})\\c_1c_2
\end{pmatrix}
\end{align*}
which has neutral element $(0,1)^\top$. The image of $\Psi_T:D(\Psi_T)\to(\mathbb R\times\mathbb C,\circ_T)$
%when restricting its domain to $\mathsf{EnTO}(H_S,T)$
is depicted in Fig.~\ref{fig1}.
\begin{figure}[!htb]
\centering
\includegraphics[width=0.5\textwidth]{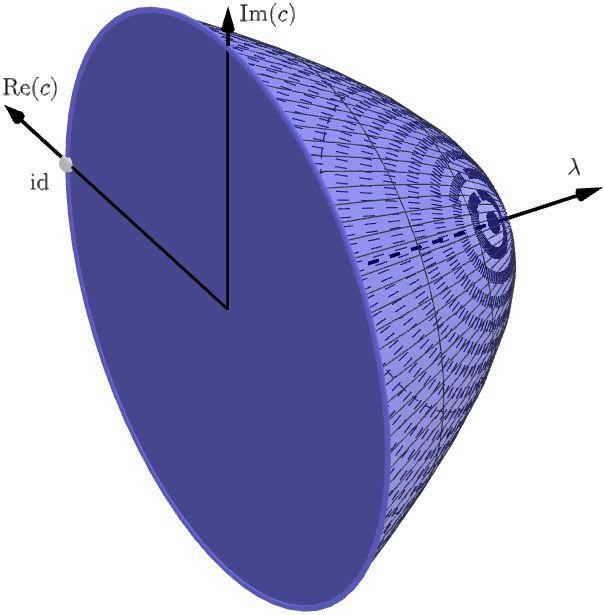}
\caption{Graph of $\Psi_T$ from \eqref{eq:rep_TO_map} for non-degenerate $H_S$ when restricting the domain to $\mathsf{EnTO}(H_S,T)$. The identity is mapped to $(0,1)^\top$. The ``classical'' channels, i.e.~the channels which set all coherences to zero are located on the $\lambda$-axis. The outer curve in $\lambda$-direction is described by $\sqrt{(1-\lambda)(1-\lambda e^{- \Delta E /T})}$, here depicted for $\Delta E=1$ and $e^{-1/T}=0.2$.
%under the assumption $\langle g_1,S(|g_1\rangle\langle g_2|)g_2\rangle\geq 0$ (i.e.~phase factor $e^{i\phi}$ equals one).
%as a subset of $\mathbb R\times\mathbb C\simeq\mathbb R^3$.
%
%
%To get a better feeling for this graph
%and the relation between $\mathsf{EnTO}$ and $\overline{\mathsf{TO}}$
%let us make some simple observations:
%
%While the top right point of the graph corresponds to the identity, the boundary on the right (i.e.~$\lambda=1$) corresponds to pure dephasing. Note that this boundary is also a part of $\overline{\mathsf{TO}}$ (even of the thermal operations without the closure, cf.~Example \ref{ex_therm_op}).
%
%The boundary on the bottom (i.e.~$r=0$) corresponds to the maps which only encode $d$-stochastic matrices (and the coherences are mapped to zero). 
%
%Note that if any point $(\lambda,r)$ of the above graph corresponds to a thermal operation, then so does every point below it, i.e.~$\{(\lambda,cr):c\in[0,1]\}$ because the thermal operations form a semigroup which in particular contains every dephasing map (cf.~Example \ref{ex_therm_op} (i)). 
%
%The upper boundary $r(\lambda)=\sqrt{ (1-\lambda)(1-\lambda d) }$ for all $d\in(0,1)$ satisfies $\lambda< r(\lambda)<\sqrt{\lambda}$ on $(0,1)$, and one even has the uniform limits $r(\lambda)\to\sqrt{\lambda}$ as $d\to 0^+$ and $r(\lambda)\to\lambda$ as $d\to 1^-$.
}
\label{fig1}
\end{figure}

Interestingly $\circ_T$ operates commutatively which -- because $\Psi_T$ is a faithful semigroup representation of $\mathsf{EnTO}$ -- yields the following:
\begin{corollary}\label{coro_commutative}
Let $T\in(0,\infty]$ and let $H_S\in{i}\mathfrak u(2)$ be non-degenerate. For all 
$S_1,S_2\in\mathsf{EnTO}(H_S,T)$ one has $S_1\circ S_2=S_2\circ S_1$. In 
particular every subset of $\mathsf{EnTO}(H_S,T)$ is commutative, as well.
\end{corollary}
Be aware that this result does not hold if $H_S$ is degenerate: for example the group $\operatorname{Ad}_{\mathsf{SU}(2)}$ is a non-commutative subset of $\mathsf{TO}(\mathbbm1_2,T)\subseteq\mathsf{EnTO}(\mathbbm1_2,T)$.
Moreover, unsurprisingly, this result does not generalize to higher dimensions because already the Gibbs-stochastic matrices form a non-commutative semigroup in three and more dimensions (Appendix A in \cite{vomEnde22}).

Also this semigroup representation turns into a group representation if points of the form $(\frac{1}{1+e^{-1/T}},*)$, $(*,0)$ are excluded from the domain of $\Psi_T$. Then the inverse of any $(\lambda,c)$ from the restricted domain of $\Psi_T$ under $\circ_T$ is given by $( \frac{\lambda}{\lambda(1+e^{-1/T})-1} , \frac{1}{c} )$. Indeed if one defines the map $x\circ_ay:=x+y-xya$ on $\mathbb R\times\mathbb R$ for any non-zero $a$, then
\begin{align*}
\mathfrak I_a:(\mathbb R\setminus \{a^{-1}\},\circ_a)&\to(\mathbb R\setminus\{0\},\cdot)\\
x&\mapsto 1-ax
\end{align*}
is a group isomorphism because $(1-ax)(1-ay)=1-a(x\circ_ay)$.
In particular the map $\mathfrak I_a$ transfers commutativity of $(\mathbb R\setminus\{0\},\cdot)$ over to $(\mathbb R\setminus \{a^{-1}\},\circ_a)$, and thus ultimately to $\circ_T$ because $\circ_T\equiv\circ_a|_{a=1+e^{-1/T}}\times\,\cdot\,|_\mathbb C$. Thus in a way, Corollary \ref{coro_commutative} is of the same fundamental structure as the statement that multiplying real numbers is commutative.

With this we are ready to show how $\mathsf{TO}$ sits inside $\mathsf{EnTO}$ in two dimensions.
It has first been shown by \'Cwikli\'nski et al.~that for all $H_S\in{i}\mathfrak{u}(2)$ and all $T\in(0,\infty]$ the sets $\overline{\mathsf{TO}(H_S,T)}$ and $\mathsf{EnTO}(H_S,T)$ coincide \cite{Cwiklinski15}. Using the above semigroup representation we outlined their proof in Appendix \ref{app_proof_thm_equiv_qubit}. 
Be aware that their proof relied on bath Hamiltonians with exponentially-growing degeneracies of the energy levels. This requirement was eventually shown to be unnecessary: the observation that it suffices to consider truncated single-mode bosonic baths (i.e.~$H_B=\Delta E\operatorname{diag}(0,1,\ldots,m)$) was first done by Scharlau et al.~for classical states (Sec.~IV in \cite{Scharlau18}), followed by Hu et al.~\cite{Ding19} for the general case.
We summarize and extend on their results in the following theorem. 
Most notably the only advantage degenerate baths give over non-degenerate ones is generating full dephasing at low temperatures:
\begin{theorem}\label{thm_equiv_qubit}
%Given $H_S\in{i}\mathfrak{u}(n)$ and $T,\Delta E>0$ define $\mathsf{TO}^{\Delta E}_\mathsf{Spin,0}(H_S,T)$ as the collection of all thermal operations with non-degenerate and equidistant bath Hamiltonians, i.e.
%$$
%\mathsf{TO}^{\Delta E}_\mathsf{Spin,0}(H_S,T) :=\bigcup_{m\in\mathbb N}\Big\{ \Phi_{T,m}(H_B,e^{iH_\mathsf{tot}}):\ \substack{H_B=\operatorname{diag}(j\Delta E)_{j=1}^m,
%%\text{where }\sum_{i=1}^m\alpha_i>0
%H_\mathsf{tot}\in{i}\mathfrak{u}(mn),\\
%[H_\mathsf{tot},H_S\otimes\mathbbm{1}_B+\mathbbm{1}\otimes H_B]=0 }\Big\}\,.
%$$
Let $H_S\in{i}\mathfrak u(2)$ and $T\in(0,\infty]$ be given. Defining $\Delta E:=\max\sigma(H_S)-\min\sigma(H_S)$ the following statements hold:
\begin{itemize}
\item[(i)] $\overline{\mathsf{TO}(H_S,T)}=\mathsf{EnTO}(H_S,T)$
\item[(ii)] If $H_S$ has non-degenerate spectrum, then $\mathsf{TO}(H_S,T)$ is convex as it equals
%\begin{align*}
%\mathsf{TO}(H_S,T)&=\mathsf{TO}^{\Delta E}_\mathsf{Spin}(H_S,T)
%%= \Psi^{-1}\big\{ (\lambda,re^{i\phi}) :\substack{ \lambda\in[0,1],\phi\in[-\pi,\pi)\\
%r\in\big[0,\sqrt{ (1-\lambda)(1-\lambda e^{- \Delta E /T}) } \big)} \big\} \\
\begin{equation}\label{eq:TO_qubit_set}
\mathsf{EnTO}(H_S,T)\setminus \Psi_T^{-1}\big(\Big\{ \begin{pmatrix}\lambda\\\sqrt{ (1-\lambda)(1-\lambda e^{- \Delta E /T}) } e^{i\phi}\end{pmatrix} : \lambda\in(0,1],\phi\in[-\pi,\pi)\Big\}\big)\,.
\end{equation}
In other words an enhanced thermal operation can be implemented via a thermal operation if and only if it is a dephasing map ($\lambda=0$) or if it lies in the relative interior\footnote{
By this we mean the interior relative to the subset $\Psi_T^{-1}(\mathbb R^3)$ of $\mathcal L(\mathbb C^{2\times 2})$, that is, the usual interior when considering Fig.~\ref{fig1}.
}
of $\mathsf{EnTO}(H_S,T)$. In particular the difference between $\mathsf{TO}(H_S,T)$ and $\mathsf{EnTO}(H_S,T)$ occurs only on the relative boundary of the enhanced thermal operations.
%and
%\begin{align*}
%\big(\mathsf{EnTO}(H_S,T)\big)^\circ=\big({\mathsf{TO}(H_S,T)}\big)^\circ=\big({\mathsf{TO}^{\Delta E}_\mathsf{Spin}(H_S,T)}\big)^\circ=\langle{\mathsf{TO}^{\Delta E}_\mathsf{Spin,0}(H_S,T)}\rangle_\mathsf{Semigroup}^\circ
%\end{align*}
%where $(\cdot)^\circ$ denotes the interior relative to ??????????????
\item[(iii)] If $H_S$ has non-degenerate spectrum, then the semigroup 
%\marginpar{specify superset of missing maps}
generated by thermal operations with bath Hamiltonians $\operatorname{diag}(0,\Delta E,\ldots,m\Delta E)$ for some $m\in\mathbb N_0$ equals
%$$
%\mathsf{TO}(H_S,T)\setminus \Psi^{-1}\Big\{ \begin{pmatrix}\lambda\\0\end{pmatrix} : \lambda\in\Big[0, \max\big\{ 0,\min\{ \tfrac{1-e^{-\Delta E/T}}{1+e^{-\Delta E/T}} ,\tfrac{(1-e^{-\Delta E/T})^4-(e^{-\Delta E/T})^4}{(1-e^{-\Delta E/T})^3} \}\big\} \Big)\Big\}\,.
%$$
%Thus it coincides with 
$\mathsf{TO}(H_S,T)$
%once the closure is applied to both sets, and $\langle{\mathsf{TO}^{\Delta E}_\mathsf{Spin,0}(H_S,T)}\rangle_\mathsf{Semigroup}= \mathsf{TO}(H_S,T)$ 
if and only if $T\in(\frac{\Delta E}{\ln 2},\infty]$. Should the two sets not be the same (i.e.~if $T\leq\frac{\Delta E}{\ln 2}$) 
their difference is a subset of $\Psi_T^{-1}(\{(\lambda,0)^\top:\lambda\in[0,1]\} )$ (i.e.~the $\lambda$ axis in Fig.~\ref{fig1}); in particular the two sets coincide after taking the closure.
\item[(iv)] Every enhanced thermal operation can be realized by a thermal operation with a single-mode bosonic bath, that is, $H_B=\operatorname{diag}(j\Delta E)_{j=0}^\infty\,$, if and only if $T\in[\frac{\Delta E}{\ln 2},\infty]$. Indeed if $T<\frac{\Delta E}{\ln 2}$ the difference between the two sets has measure strictly larger than zero. 
This is due to the fact that for small enough $\lambda$ the action of thermal operations on the off-diagonal elements cannot become arbitrarily small anymore.
\end{itemize}
\end{theorem}
\noindent For convenience, proofs of these results are given in Appendix \ref{app_proof_thm_equiv_qubit}. Be aware that most of these results do not generalize to more than two dimensions.

%\begin{proof}
%The case of degenerate $H_S$ has been treated in Example \ref{example_10_1} (i). Lemma \ref{lemma_pauli_gen_{H_S}} allows us to further reduce the statement from arbitrary non-degenerate $H_S$ to $\operatorname{diag}(0,1)$. All we have to show then is $\overline{\mathsf{TO}(\operatorname{diag}(0,1),T)}\supseteq\mathsf{EnTO}(\operatorname{diag}(0,1),T)$ because the converse holds by Prop.~\ref{prop_modes_sep}.
%
%Now let $\lambda\in[0,1]$, $T>0$ be given. Following \cite[Supplementary Note 4, Section IV]{Cwiklinski15} the idea is -- given any rational approximation $\mu$ of $e^{1/T}$, i.e.~$\mu\in(1,e^{1/T})\cap\mathbb Q$ -- to explicitly construct a family of maps $(S^\mu_m)_{m\in\mathbb N}\subset \mathsf{TO}(\operatorname{diag}(0,1),T)$ such that
%$$
%\lim_{\substack{\mu\to (e^{1/T})^-\\\mu\in\mathbb Q}}\lim_{m\to\infty} C(S_m^\mu)=\begin{pmatrix}
%1-\lambda e^{-1/T}&0&0&\sqrt{(1-\lambda)(1-\lambda e^{-1/T})}\\
%0&\lambda e^{-1/T}&0&0\\
%0&0&\lambda&0\\
%\sqrt{(1-\lambda)(1-\lambda e^{-1/T})}&0&0&1-\lambda
%\end{pmatrix}\,.
%$$
%This would conclude the proof because ..
%
%state idea of construction. ``Details in Appendix \ref{app_A} (after Lemma \ref{lemma_degenerate_spin_TO})''
%\end{proof}
One conclusion ``hidden'' in the proof of Theorem \ref{thm_equiv_qubit} concerns the dimension of bath Hamiltonians. Recall that the Stinespring dilation $S=\operatorname{tr}_B(U((\cdot)\otimes|\psi\rangle\langle\psi|)U^*)$ of an arbitrary quantum channel $S\in\textsc{cptp}(n)$ can always be chosen such that $\psi\in\mathbb C^k$ for some $k\leq n^2$ (Thm.~6.18 in \cite{Holevo12}). This result breaks down for thermal operations, that is, if $|\psi\rangle\langle\psi|$ is replaced by a Gibbs state and $U$ is required to be energy-preserving: For every $\lambda\in(0,1)$, $m\in\mathbb N$ there exists $\varepsilon_m>0$ and $r\in (  \sqrt{(1-\lambda)(1-\lambda e^{- \Delta E /T})} -\varepsilon_m, \sqrt{(1-\lambda)(1-\lambda e^{- \Delta E /T})}  )$ such that the thermal operation $\Psi_T^{-1}(\lambda,r)$ can be implemented by a bath Hamiltonian $H_B$ only if it is of size $m\times m$ or larger. However $\varepsilon_m$ goes to zero as $m\to\infty$ because $\mathsf{TO}(H_S,T)$ gets arbitrarily close to the (relative) boundary of $\mathsf{EnTO}$ for all $T\in(0,\infty]$.
The final observation we want to make is that the ``geometry'' of the qubit thermal operations
pertains to the set $\{\Phi(\rho):\Phi\in\overline{\mathsf{TO}(H_S,T)}\}$ -- which is sometimes referred to as the (future) thermal 
cone \cite{Lostaglio15_2,Korzekwa17,Oliveira22} --
as is depicted in Figure \ref{fig_bloch}.
This recovers what has already been observed in \cite{Lostaglio15_2}, specifically Figures 1 \& 6 in said article.
In particular the boundary of the thermal cone is not linear, meaning that even after factoring out the rotational symmetry inflicted by the thermal operations $\rho\mapsto\operatorname{diag}(1,e^{i\phi})\rho\operatorname{diag}(1,e^{-i\phi})$ the set of extreme points is still infinite.
\begin{figure}[!htb]
\centering
\includegraphics[width=0.32\textwidth]{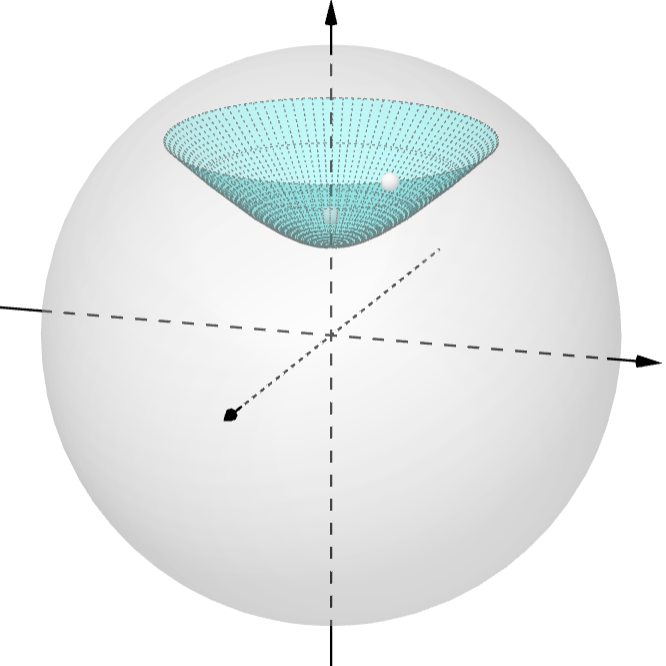}
\includegraphics[width=0.32\textwidth]{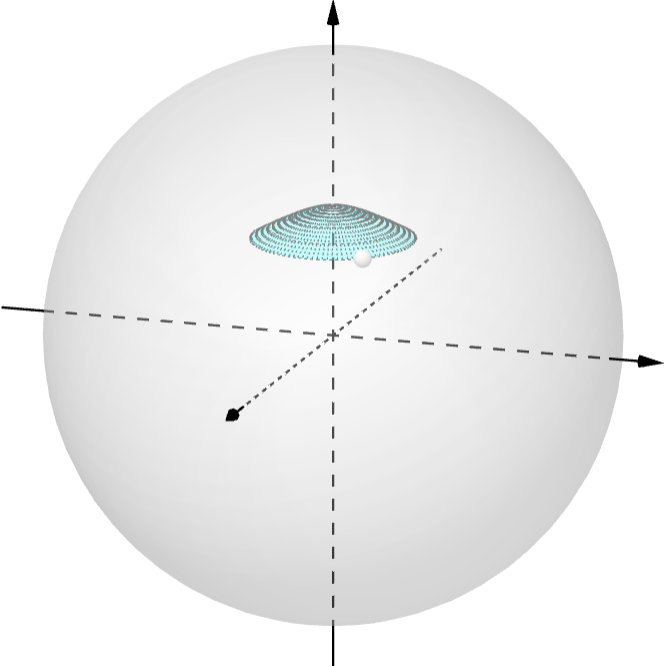}
\includegraphics[width=0.32\textwidth]{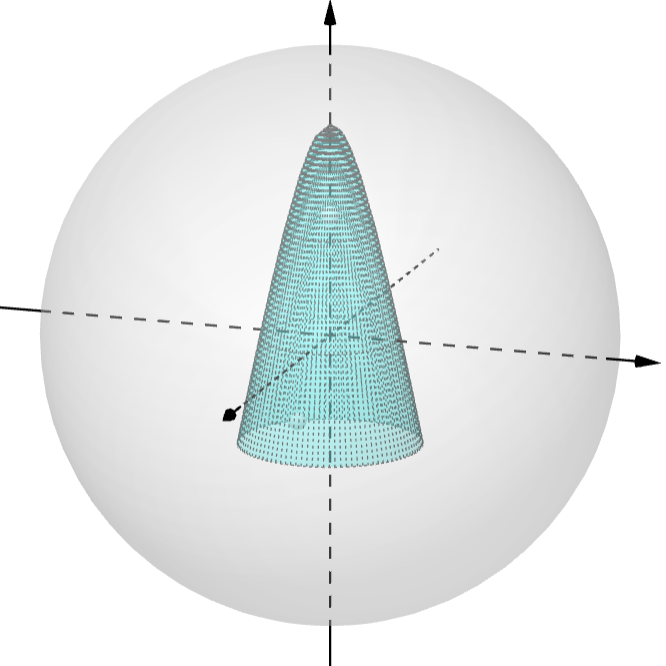}
\caption{
Future thermal cones of initial states with Bloch vector 
$c\cdot(0.5,0.4,\sqrt{0.59})^\top$ for different values of $c$. Left: $c=0.9$.
Middle: $c=0.45$. Right: $c=-0.5$.
The ``system parameter`` we chose is $e^{-H_S/T}=0.45$.
Here the point on the boundary of the cone is the Bloch vector of the initial 
state, and the point in the interior of the cone (on the $z$-axis) corresponds to 
the Gibbs state.
%Graph of $\Psi_T$ from \eqref{eq:rep_TO_map} for non-degenerate $H_S$ when restricting the domain to $\mathsf{EnTO}(H_S,T)$. The identity is mapped to $(0,1)^\top$. The ``classical'' channels, i.e.~the channels which set all coherences to zero are located on the $\lambda$-axis. The outer curve in $\lambda$-direction is described by $\sqrt{(1-\lambda)(1-\lambda e^{- \Delta E /T})}$, here depicted for $\Delta E=1$ and $e^{-1/T}=0.2$.
}
\label{fig_bloch}
\end{figure}
\section{Conclusion and Open Questions}\label{sec_concl}
We reduced the set of bath Hamiltonians needed in the definition of thermal operations to those which have so-called ``resonant spectrum'' with respect to the system. This resonance condition is about the spectrum of the bath forming a connected graph
with respect to the possible energy transitions of the system (cf.~Figure \ref{figa}). We saw that the action coming from any bath which does not satisfy this condition decomposes into the convex sum of two or more thermal operations with resonant bath.
%What this notion of resonance ensures intuitively is that sufficiently many energy levels of the bath are allowed to ``communicate'' with each other once coupled to the system.
Be aware that this notion is logically independent from a bath Hamiltonian containing \textit{all} possible transitions of the system (i.e.~$\sigma(\operatorname{ad}_{H_S})\subseteq\sigma(\operatorname{ad}_{H_B})$). The latter is a necessary condition for the diagonal action of a thermal operation to be represented by a strictly positive Gibbs-stochastic matrix (cf.~also Remark \ref{rem_bath_condition_prob}).

Either way, as a consequence of the new-found resonance we showed that if any multiple of the system's Hamiltonian has rational Bohr spectrum, then there exists an energy gap such that the set of thermal operations is fully characterized by spin Hamiltonians w.r.t.~this energy gap (Corollary \ref{coro_spin_rational}).
As rational numbers are a key concept of this statement this suggests that the thermal operations behaves discontinuously at certain Hamiltonians.
Indeed we were able to show that if either the spectrum or the transitions of the system's Hamiltonian are degenerate, then the set of thermal operations changes discontinuously with respect to the Hausdorff metric (Example \ref{ex_discont}).
Taking the nature of our two counter-examples into account it seems reasonable to conjecture that $\mathsf{TO}$ is at least continuous in the temperature (for fixed Hamiltonian), as well as $\mathsf{TO}$ admitting some form of semi-continuity in the joint argument $(H_S,T$).

An idea to restore continuity -- inspired by the concept of average energy conservation\cite{Skrzypczyk14} -- could be to allow for an error in the energy conservation condition: Given any $\varepsilon\geq 0$ define $\mathsf{TO}^\varepsilon(H_S,T)$ by adjusting $U(H_S\otimes\mathbbm{1}_B+\mathbbm{1}\otimes H_B)U^*=H_S\otimes\mathbbm{1}_B+\mathbbm{1}\otimes H_B$ in the definition of $\mathsf{TO}(H_S,T)$ to
$\|U(H_S\otimes\mathbbm{1}_B+\mathbbm{1}\otimes H_B)U^*-H_S\otimes\mathbbm{1}_B+\mathbbm{1}\otimes H_B\|_\infty \leq 2\varepsilon\|U-\mathbbm1\|_\infty(\|H_S\|_\infty+\|H_B\|_\infty)$.
This is motivated by the simple estimate
\begin{align*}
\| U(H_S\otimes\mathbbm{1}_B+\mathbbm{1}\otimes &H_B)U^*-H_S\otimes\mathbbm{1}_B+\mathbbm{1}\otimes H_B  \|_\infty\\
&\leq \|  U(H_S\otimes\mathbbm{1}_B+\mathbbm{1}\otimes H_B)U^*-(H_S\otimes\mathbbm{1}_B+\mathbbm{1}\otimes H_B)U^* \|_\infty\\
&+\| (H_S\otimes\mathbbm{1}_B+\mathbbm{1}\otimes H_B)U^*-H_S\otimes\mathbbm{1}_B+\mathbbm{1}\otimes H_B  \|_\infty\\
&\leq 2\| U-\mathbbm1  \|_\infty\|H_S\otimes\mathbbm{1}_B+\mathbbm{1}\otimes H_B  \|_\infty=2\| U-\mathbbm1  \|_\infty(\| H_S  \|_\infty+\|  H_B \|_\infty)\,.
\end{align*}
Note that $\mathsf{TO}^0$ recovers $\mathsf{TO}$, while $\mathsf{TO}^\varepsilon$ for $\varepsilon\geq1$ renders energy conservation obsolete because then every unitary satisfies the condition in question.

In some way introducing such $\varepsilon$ (small enough) ``smoothens out'' the binary nature of energy-conservation by allowing for unitaries which are $\varepsilon$-close to conserving the energy of the full uncoupled system. As a result new transitions which were previously forbidden do not appear instantly once the Bohr spectrum becomes degenerate, but the norm of the corresponding diagonal block in the unitary correlates to the error $\varepsilon$.
Now in order to get a collection of maps which is ``physically reasonable'' one may have to intersect $\mathsf{TO}^\varepsilon$ with the Gibbs-preserving maps, or maybe consider the semigroup generated by $\mathsf{TO}^\varepsilon$.

Finally we reviewed what is known about thermal operations in the qubit case and, using our results on baths with resonant spectrum, extended on this knowledge by specifying how the set of thermal qubit operations looks \textit{exactly}.
We did so by means of a faithful semigroup representation which translates the (enhanced) thermal operations into a subset of ordinary 3D space, thus allowing for a visualization from which intuition benefits as well.
Interestingly, our proof of the main qubit results (Theorem \ref{thm_equiv_qubit}) gave two different families of energy-preserving unitaries for approximating the extreme points of the enhanced thermal operations, depending on whether the temperature is finite or infinite. For now, finding a (temperature-dependent) family of unitaries which continues to do the job in the limit $T\to\infty$ is an open problem.

In any case these qubit results readily leads us to a number of open questions for general (finite-dimensional) systems, two of the more obvious ones being the following:
\begin{itemize}
\item Is $\mathsf{TO}(H_S,T)$ is convex for all $H_S\in{i}\mathfrak u(n)$, $T\in(0,\infty]$? We showed that this holds true in two dimensions; however our proof -- as well as the proof of convexity of $\overline{\mathsf{TO}}$ -- are very much unsuited to tackle this question in higher dimensions, because they are either too complicated or they fundamentally rely on rational numbers and approximations. Should convexity hold in general (without the closure) it seems likely that proving so requires some deeper knowledge about thermal operations.
\item Can one specify an upper bound for how degenerate the energy levels of the bath need to be?
For qubits, our proof of Prop.~\ref{prop_1} (i) shows that every qubit thermal operation is the composition of something ``close to an extreme point'' (i.e.~non-degenerate bath) and a partial dephasing which can always be implemented by a trivial two-level bath. 
Then the proof of the semigroup property shows that the energy levels of the bath Hamiltonian of the composite operation has degeneracy at most two.
Thus one may conjecture that the definition of $\mathsf{TO}$ can be restricted to such (resonant) bath Hamiltonians which have degeneracy at most dimension of the system. This claim is further supported by the fact that full dephasing in $n$ dimensions can always be realized by choosing $H_B=\mathbbm1_n$ (together with $U=\bigoplus_{j=1}^n \operatorname{diag}(e^{2\pi ijk/n})_{k=0}^{n-1}$), cf.~p.~88 in \cite{Bhatia07}.
\end{itemize}
Generally speaking, settling which results regarding $\overline{\mathsf{TO}}$ continue to hold once the closure is waived should be a future line of research. We expect that any progress in this direction will reveal more of the intrinsic structure the set of thermal operations has.
\begin{acknowledgments}
I would like to thank Emanuel Malvetti, Gunther Dirr, Thomas Schulte-Herbr\"uggen, and Amit Devra
%and the anonymous referee
for valuable and constructive comments during the preparation of this manuscript.
Moreover I am grateful to the anonymous referee for their valuable comments which led to an improved presentation of the material.
This research is part of the Bavarian excellence network \textsc{enb}
via the International PhD Programme of Excellence
\textit{Exploring Quantum Matter} (\textsc{exqm}), as well as the \textit{Munich Quantum Valley} of the Bavarian
State Government with funds from Hightech Agenda \textit{Bayern Plus}.
\end{acknowledgments}

\appendix
\section{Proof of Proposition \ref{prop_1}}\label{app_proof_lemma_1}
%The fact that the thermal operations form a semigroup is a straightforward computation, where the key tool one uses is the following: given Hilbert spaces $\mathcal H_1,\mathcal H_2,\mathcal H_3$ and trace class operators $A\in\mathcal B^1(\mathcal H_1)$, $C\in\mathcal B^1(\mathcal H_2\otimes\mathcal H_3)$ the identity $ \operatorname{tr}_{\mathcal H_3}(A\otimes C)=A\otimes\operatorname{tr}_{\mathcal H_3}(C) $ holds true. Then given ..
\noindent We will only prove the case $T\in(0,\infty)$ as the case $T=\infty$ is done analogously.

(i): While the semigroup property is well-known we still give a proof here for the sake of completeness. Given thermal operations $S_i$, $i=1,2$ with associated bath-Hamiltonian $H_{B,i}\in\mathbb C^{m_i\times m_i}$ and energy-preserving unitary $U_i\in\mathbb C^{m_in\times m_in}$, respectively, we claim that $S_1\circ S_2=\Phi_{T,m_1m_2}(H_B,U)$ with $H_B:=H_{B,1}\otimes\mathbbm1_{m_2}+\mathbbm1_{m_1}\otimes H_{B,2}$ and
\begin{equation}\label{eq:def_U_semigroup}
U:=(U_1\otimes\mathbbm1_{m_2})(\mathbbm1_{n}\otimes\mathbb F^*)(U_2\otimes\mathbbm1_{m_1})(\mathbbm1_{n}\otimes\mathbb F)\in\mathbb C^{ n \times n }\otimes\mathbb C^{ m_1 \times m_1 }\otimes\mathbb C^{ m_2 \times m_2 }\,.
\end{equation}
Here $\mathbb F:\mathbb C^{m_1}\otimes\mathbb C^{m_2}\to\mathbb C^{m_2}\otimes\mathbb C^{m_1}$ is the flip operator, i.e.~the unique linear operator which satisfies $\mathbb F(x\otimes y)=y\otimes x$ for all $x\in\mathbb C^{m_1}$, $y\in\mathbb C^{m_2}$. Note that $\mathbb F$ also ``generates'' the matrix flip, that is, $\mathbb F^*(B\otimes A)\mathbb F=A\otimes B$ for all $A\in\mathbb C^{m_1 \times m_1}, B\in\mathbb C^{ m_2 \times m_2}$. Now the idea as to why $S_1\circ S_2$ can be described in such a way is depicted in Figure \ref{fig0}.
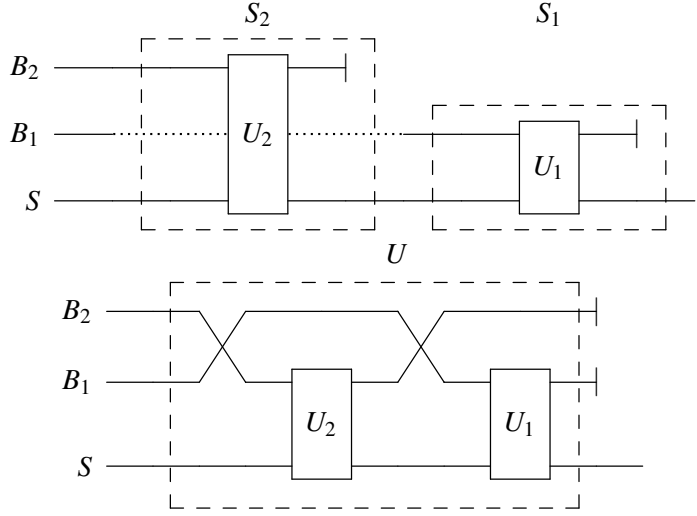
\begin{figure}[!htb]
$$\Qcircuit @C=2em @R=1.4em {
& && S_2& & & & S_1&&\\
\lstick{B_2}& \qw& \qw & \multigate{2}{U_2} &{|}\qw &&& &&\\
\lstick{B_1}& \qw& \tikz[baseline]{\draw[dotted][thick] (0,.3ex)--++(1.52,0) ;} & \nghost{U_2} & \tikz[baseline]{\draw[dotted][thick] (0,.3ex)--++(1.52,0) ;}&& \qw & \multigate{1}{U_1}& {|}\qw & \\
\lstick{S}& \qw& \qw & \ghost{U_2} & \qw &\qw&\qw & \ghost{U_1}&\qw &\qw\gategroup{2}{3}{4}{5}{2em}{--}\gategroup{3}{7}{4}{9}{2em}{--}
}
$$
$$
\Qcircuit @C=1.6em @R=2em {
& & & &&U& & &&\\
\lstick{B_2}& \qw & \qw &\link{1}{-1}& \qw &\qw &\link{1}{-1} & \qw & {|}\qw&\\
\lstick{B_1}& \qw& \qw &\link{-1}{-1} & \multigate{1}{U_2}& \qw &\link{-1}{-1} & \multigate{1}{U_1} & {|}\qw &\\
\lstick{S}& \qw& \qw & \qw & \ghost{U_2} & \qw &\qw& \ghost{U_1} & \qw&\qw
\gategroup{2}{3}{4}{8}{2em}{--}
}
$$
\caption{
Top: Circuit of first applying $S_2$ followed by applying $S_1$ (i.e.~$S_1\circ S_2$). Bottom: Circuit of $\Phi_{T,m_1m_2}(H,U)$ with $U$ from \eqref{eq:def_U_semigroup}. The idea as to why the action of these circuits coincides is that tracing out the bath $B_2$ commutes with applying $U_1$ because the action of the latter on $B_2$ is trivial.
%Note that for convenience we attached the first bath ``from the left''; in the computation below we will have to compensate for this change in order by applying a local flip operator in some places.
}
\label{fig0}
\end{figure}
%
%\begin{figure}[!htb]
%\centering
%\includegraphics[width=0.6\textwidth]{1b.pdf}
%\caption{
%The idea is that applying $S_1$ after $S_2$ and tracing out the baths \textit{after} running both unitary evolutions on the respective subsystem amount to the same thing.
%This is the case because tracing out the second bath commutes with applying $U_1$ (i.e.~the unitary operation on the subsystem $S+B_1$) because its action on the second bath is trivial.
%Note that for convenience we attached the first bath ``from the left''; in the computation below we will have to compensate for this change in order by applying a local flip operator in some places.
%}
%\label{fig0}
%\end{figure}${}$
The key tool one uses in this proof is the following partial trace identity: Given Hilbert spaces $\mathcal H_1,\mathcal H_2$, a trace-class operator $A$ on $\mathcal H_1\otimes\mathcal H_2$, and bounded linear operators $C,D$ on $\mathcal H_1$ and $B$ on $\mathcal H_2$ such that $B$ is invertible one readily verifies
\begin{equation}\label{eq:partial_trace_two_baths_0}
\operatorname{tr}_{\mathcal H_2}\big((C\otimes B)A(D\otimes B^{-1})\big)=C \operatorname{tr}_{\mathcal H_2}(A)D\,.
\end{equation}
Consequently, given Hilbert spaces $\mathcal H_1,\mathcal H_2,\mathcal H_3$ and trace-class operators $X$ on $\mathcal H_1\otimes\mathcal H_2$ and $Y$ on $\mathcal H_3$ one finds
\begin{equation}\label{eq:partial_trace_two_baths}
 \operatorname{tr}_{\mathcal H_2}(X)\otimes Y=\operatorname{tr}_{\mathcal H_2}\big((\mathbbm1\otimes\mathbb F^*)(X\otimes Y)(\mathbbm1\otimes\mathbb F)\big)
\end{equation}
where $\operatorname{tr}_{\mathcal H_2}$ on the right-hand side of Eq.~\eqref{eq:partial_trace_two_baths} is the partial trace on $(\mathcal H_1\otimes\mathcal H_3)\otimes\mathcal H_2$ and, again, $\mathbb F:\mathcal H_3\otimes\mathcal H_2\to\mathcal H_2\otimes\mathcal H_3$ is the flip operator. Note that \eqref{eq:partial_trace_two_baths_0} implies \eqref{eq:partial_trace_two_baths} by choosing $C=D=\mathbbm1$, $B=\mathbb F$, and $A=X\otimes Y$.
With all of this we compute that $S_1\circ S_2$ equals
\begin{align*}
&\operatorname{tr}_{B_1}\circ \operatorname{Ad}_{U_1}\circ \Big( (\cdot) \otimes\frac{e^{-H_{B,1}/T}}{\operatorname{tr}(e^{-H_{B,1}/T})}\Big)\circ \operatorname{tr}_{B_2}\circ \operatorname{Ad}_{U_2}\circ \Big( (\cdot) \otimes\frac{e^{-H_{B,2}/T}}{\operatorname{tr}(e^{-H_{B,2}/T})}\Big)\\
\overset{\eqref{eq:partial_trace_two_baths}}=& \operatorname{tr}_{B_1}\circ \operatorname{Ad}_{U_1}\circ \operatorname{tr}_{B_2}\circ\operatorname{Ad}_{\mathbbm1\otimes\mathbb F^*}\circ \Big( (\cdot) \otimes\frac{e^{-H_{B,1}/T}}{\operatorname{tr}(e^{-H_{B,1}/T})}\Big)\circ \operatorname{Ad}_{U_2}\circ \Big( (\cdot) \otimes\frac{e^{-H_{B,2}/T}}{\operatorname{tr}(e^{-H_{B,2}/T})}\Big) \\
\overset{\eqref{eq:partial_trace_two_baths_0}}=& \operatorname{tr}_{B_1,B_2}\circ \operatorname{Ad}_{U_1\otimes\mathbbm1}\circ\operatorname{Ad}_{\mathbbm1\otimes\mathbb F^*}\circ \Big( (\cdot) \otimes\frac{e^{-H_{B,1}/T}}{\operatorname{tr}(e^{-H_{B,1}/T})}\Big)\circ \operatorname{Ad}_{U_2}\circ \Big( (\cdot) \otimes\frac{e^{-H_{B,2}/T}}{\operatorname{tr}(e^{-H_{B,2}/T})}\Big) \\
=& \operatorname{tr}_{B_1,B_2}\circ \operatorname{Ad}_{U_1\otimes\mathbbm1}\circ\operatorname{Ad}_{\mathbbm1\otimes\mathbb F^*}\circ \operatorname{Ad}_{U_2\otimes\mathbbm1}\circ \Big( (\cdot) \otimes\frac{e^{-H_{B,1}/T}}{\operatorname{tr}(e^{-H_{B,1}/T})}\Big)\circ \Big( (\cdot) \otimes\frac{e^{-H_{B,2}/T}}{\operatorname{tr}(e^{-H_{B,2}/T})}\Big) \\
\overset{\eqref{eq:def_U_semigroup}}=& \operatorname{tr}_{B_1,B_2}\circ \operatorname{Ad}_{U}\circ\operatorname{Ad}_{\mathbbm1\otimes\mathbb F^*}\circ\Big( (\cdot) \otimes\frac{e^{-H_{B,2}/T}}{\operatorname{tr}(e^{-H_{B,2}/T})}\otimes\frac{e^{-H_{B,1}/T}}{\operatorname{tr}(e^{-H_{B,1}/T})}\Big) =\Phi_{T,m_1m_2}(H,U)\,.
%=&\operatorname{tr}_{B_1,B_2}\Big(U\Big((\cdot)\otimes \frac{e^{-H_{B,1}/T}}{\operatorname{tr}(e^{-H_{B,1}/T})} \otimes\frac{e^{-H_{B,2}/T}}{\operatorname{tr}(e^{-H_{B,2}/T})}\Big)U^*\Big)=\Phi_{T,m_1m_2}(H,U)\,.
\end{align*}
Finally let us sketch why $U$ is energy-preserving by tracking how each of the three components of $H_S\otimes\mathbbm1+\mathbbm1\otimes H_B=H_S\otimes\mathbbm1 \otimes\mathbbm1 +\mathbbm1 \otimes H_{B,1}\otimes\mathbbm1 +\mathbbm1 \otimes\mathbbm1 \otimes H_{B,2}$ change with the factors of $U$:
\begin{align*}
\begin{bmatrix}
H_S\otimes\mathbbm1 \otimes\mathbbm1 \\
\mathbbm1 \otimes H_{B,1}\otimes\mathbbm1 \\
\mathbbm1 \otimes\mathbbm1 \otimes H_{B,2}
\end{bmatrix}
\overset{\operatorname{Ad}_{\mathbbm1\otimes\mathbb F}}\longrightarrow
\begin{bmatrix}
H_S\otimes\mathbbm1 \otimes\mathbbm1 \\
\mathbbm1 \otimes\mathbbm1 \otimes H_{B,1}\\
\mathbbm1 \otimes H_{B,2}\otimes\mathbbm1
\end{bmatrix}
\overset{\operatorname{Ad}_{U_2\otimes\mathbbm1}}\longrightarrow
\begin{bmatrix}
U_2(H_S\otimes\mathbbm1 )U_2^*\otimes\mathbbm1 \\
\mathbbm1 \otimes\mathbbm1 \otimes H_{B,1}\\
U_2(\mathbbm1 \otimes H_{B,2})U_2^*\otimes\mathbbm1
\end{bmatrix}
\end{align*}
But the sum of these three matrices is equal to $H_S\otimes\mathbbm1+\mathbbm1\otimes H_B$ because $U_2$ is energy-preserving w.r.t.~$(H_S,H_{B,2})$, that is, $U_2(H_S\otimes\mathbbm1+\mathbbm1 \otimes H_{B,2} )U_2^*=H_S\otimes\mathbbm1+\mathbbm1 \otimes H_{B,2}$. Similarly one finds
\begin{align*}
\begin{bmatrix}
H_S\otimes\mathbbm1 \otimes\mathbbm1 \\
\mathbbm1 \otimes\mathbbm1 \otimes H_{B,1}\\
\mathbbm1 \otimes H_{B,2}\otimes\mathbbm1
\end{bmatrix}
\overset{\operatorname{Ad}_{\mathbbm1\otimes\mathbb F^*}}\longrightarrow
\begin{bmatrix}
H_S\otimes\mathbbm1 \otimes\mathbbm1 \\
\mathbbm1 \otimes H_{B,1}\otimes\mathbbm1 \\
\mathbbm1 \otimes\mathbbm1 \otimes H_{B,2}
\end{bmatrix}
\overset{\operatorname{Ad}_{U_1\otimes\mathbbm1}}\longrightarrow
\begin{bmatrix}
U_1(H_S\otimes\mathbbm1 )U_1^*\otimes\mathbbm1 \\
U_1(\mathbbm1 \otimes H_{B,1})U_1^*\otimes\mathbbm1\\
\mathbbm1 \otimes\mathbbm1 \otimes H_{B,2}\\
\end{bmatrix}\simeq 
\begin{bmatrix}
H_S\otimes\mathbbm1 \otimes\mathbbm1 \\
\mathbbm1 \otimes H_{B,1}\otimes\mathbbm1 \\
\mathbbm1 \otimes\mathbbm1 \otimes H_{B,2}
\end{bmatrix}\,.
\end{align*}

Boundedness of $\mathsf{TO}$ comes from the known fact\cite{PG06} that the \textsc{cptp} maps form a subset of the unit sphere w.r.t.~$\|\,\cdot\,\|_{1\to1}$.
%, that is, the usual operator norm if domain and range are equipped with the trace norm $\|\cdot\|_1:=\operatorname{tr}(\sqrt{(\cdot)^*(\cdot)})$.
Path-connectedness follows from the fact that every thermal operation can be connected to the identity in a continuous manner: Given $S= \operatorname{tr}_B(e^{iH_\mathsf{tot}}((\cdot)\otimes \frac{e^{-H_B/T}}{\operatorname{tr}(e^{-H_B/T})})e^{-iH_\mathsf{tot}})$ with $H_\mathsf{tot}$ energy-preserving,
%, by Lemma \ref{lemma_gen_energy_pres} there exists $A\in\mathbb C^{mn\times mn}$ Hermitian such that $U=e^{iA}$ and that $e^{itA}$ is energy-preserving for all $t\in\mathbb R$. Therefore
\begin{align*}
t\mapsto\operatorname{tr}_B\Big(e^{itH_\mathsf{tot}}\Big((\cdot)\otimes \frac{e^{-H_B/T}}{\operatorname{tr}(e^{-H_B/T})}\Big)e^{-itH_\mathsf{tot}}\Big) 
\end{align*}
is a continuous curve in $\mathsf{TO}(H_S,T)$ which connects $S$ ($t=1$) with $\operatorname{id}$ ($t=0$). 

(ii): The only non-trivial things here are convexity and the semigroup property. First $\overline{\mathsf{TO}}$ is a semigroup because it is the closure of a semigroup in a space where left- and right-multiplication are continuous. This is a general fact: Given any Hausdorff topological 
space $(X,\tau)$ with a binary operation $\circ:X\times X\to X$ which is left- and right-
continuous, i.e.~$x\mapsto x\circ y$ and $x\mapsto y\circ x$ are continuous for all 
$y\in X$, if $S\subset X$ is a semigroup (w.r.t.~$\circ$) then $\overline{S}$ is a 
semigroup, as well. The idea is to first show that $\overline{S}\circ 
S\subseteq\overline{S}$ by means of nets using continuity of right-multiplication. 
Based on this one sees $\overline{S}\circ \overline{S}\subseteq\overline{\overline{S}}
=\overline{S}$ in a similar fashion.

For convexity one first shows that for any two thermal operations $S_1,S_2$ and 
any $\lambda\in(0,1)\cap\mathbb Q$ the convex combination $\lambda S_1+(1-
\lambda)S_2$ again is a thermal operation, cf.~Appendix C in \cite{Lostaglio15}. Indeed let such $\lambda$ as well as $S_i\in\mathsf{TO}(H_S,T)$ generated by $H_{B,i}\in\mathbb C^{m_i\times m_i}$ Hermitian 
and $U_i\in\mathsf U(m_in)$ energy-preserving, respectively, for $i=1,2$ be given. There exist $k,d\in\mathbb N$, $k<d$ such that $
\lambda=\frac{k}{d}$.
We claim that $\lambda S_1+(1-\lambda)S_2=\Phi_{T,m_1m_2d}(H_B,U)$ where
$
H_B=H_{B,1}\otimes\mathbbm{1}\otimes\mathbbm{1}_d+\mathbbm{1}\otimes H_{B,2}\otimes\mathbbm{1}_d
$
as well as
%---in abuse of notation\footnote{
%Keeping in mind that $U$ acts on the system $\mathbb C^{n}\otimes \mathbb C^{m_1}\otimes \mathbb C^{m_2}\otimes \mathbb C^d$ this is short for the following action: the first summand applies $U_1$ to the first and second tensor factor and simultaneously $\Pi_1$ on the fourth tensor factor, and the second summand applies $U_2$ to the first and third tensor factor and $\Pi_2$ acts simultaneously on the fourth tensor factor. When considering the isomorphic system $\mathbb C^{m_1}\otimes \mathbb C^{n}\otimes \mathbb C^{m_2}\otimes \mathbb C^d$ (assuming $U_1$ belongs to the ``swapped'' thermal operation $\operatorname{tr}_B(U(\frac{e^{-H_B/T}}{\operatorname{tr}(e^{-H_B/T})}\otimes (\cdot))U^*)$) one would have $U= U_1\otimes\mathbbm{1}_2\otimes\Pi+\mathbbm{1}_1\otimes U_2\otimes(\mathbbm{1}_d-\Pi) $.\label{footnote_1}}---
%the unitary
$
U= U_1\otimes\mathbbm1\otimes\Pi+(\mathbbm1\otimes\mathbb F^*\otimes\mathbbm1)(U_2\otimes\mathbbm1\otimes(\mathbbm{1}_d-\Pi))(\mathbbm1\otimes\mathbb F\otimes\mathbbm1)
$. Here $\Pi$ is any orthogonal projection on $\mathbb C^d$ of rank $k$ and $\mathbb F$ is the flip operator from earlier. Using \eqref{eq:partial_trace_two_baths} as well as the fact that $\Pi(1-\Pi)=0$ one directly computes that $U$ is unitary and energy-preserving,
%, i.e.~$[U,H_S\otimes\mathbbm{1}_B+\mathbbm{1}\otimes H_B]=0$ \footnote{
%In the notation of footnote \ref{footnote_1} one finds that the commutator in question is equal to $$[U_1,\mathbbm{1}_1\otimes H_S+H_{B,1}\otimes\mathbbm{1}]\otimes\mathbbm{1}_2\otimes\Pi+\mathbbm{1}_1\otimes [U_2, H_S\otimes\mathbbm{1}_2+\mathbbm{1}\otimes H_{B,2} ]\otimes(\mathbbm{1}_d-\Pi)$$ which, because $U_1$ and $U_2$ are both energy-preserving, is equal to zero as desired.}
and that
$$
\Phi_{T,m_1m_2d}(H_B,U)=\operatorname{tr}\Big( \frac{e^{-H_{B,2}/T}}{\operatorname{tr}(e^{-H_{B,2}/T})} \Big)\frac{\operatorname{tr}(\Pi)}{d}S_1+\operatorname{tr}\Big( \frac{e^{-H_{B,1}/T}}{\operatorname{tr}(e^{-H_{B,1}/T})} \Big)\frac{\operatorname{tr}(\mathbbm{1}_d-\Pi)}{d}S_2
%=\lambda S_1+(1-\lambda)S_2
$$
as desired. This intermediate result will carry over to $\overline{\mathsf{TO}}$ simply by combining it with two approximation arguments. Indeed let $T_1,T_2\in \overline{\mathsf{TO}(H_S,T)}$, $\lambda\in(0,1)$, as well as $\varepsilon>0$ be given. On the one hand there exist thermal operations $S_1,S_2$ with $\|T_i-S_i\|_{1\to1}<\min\{\frac{\varepsilon}{8\lambda},\frac{\varepsilon}{2}\}$ for $i=1,2$ and on the other hand one finds $\mu\in(0,1)\cap\mathbb Q$ with $|\mu-\lambda|<\frac{\varepsilon}{8}$. By our previous considerations we know that $\mu S_1+(1-\mu)S_2$ is a thermal operation which---as we will compute now---is $\varepsilon$-close to $\lambda T_1+(1-\lambda)T_2$. Because $\varepsilon$ is arbitrary this would show $\lambda T_1+(1-\lambda)T_2\in \overline{\mathsf{TO}(H_S,T)}$. Indeed
\begin{align*}
\|(\lambda T_1&+(1-\lambda)T_2)-( \mu S_1+(1-\mu)S_2 )\|_{1\to1}\\
&\leq \|\lambda T_1-\mu S_1\|_{1\to1}+\|T_2-S_2\|_{1\to1}+\|\lambda T_2-\mu S_2\|_{1\to1}\\
&< \|\lambda T_1-\lambda S_1\|_{1\to1}+\|\lambda S_1-\mu S_1\|_{1\to1}+\frac{\varepsilon}{2}+\|\lambda T_2-\lambda S_2\|_{1\to1}+\|\lambda S_2-\mu S_2\|_{1\to1}\\
&< \lambda\cdot\frac{\varepsilon}{8\lambda}+\frac{\varepsilon}{8}\|S_1\|_{1\to1}+\frac{\varepsilon}{2}+\lambda\cdot\frac{\varepsilon}{8\lambda}+\frac{\varepsilon}{8}\|S_2\|_{1\to1}=\varepsilon\,.
\end{align*}
Here we used that $\|S_1\|_{1\to 1},\|S_2\|_{1\to 1}=1$ as stated previously.

(iii): This is implied by energy conservation because $e^{-(H_S\otimes\mathbbm{1}_B+\mathbbm{1}\otimes H_B)/T}=e^{-H_S/T}\otimes e^{-H_B/T}$. Also the subset of all \textsc{cptp} maps which have $e^{-H_S/T}$ as common fixed point is closed so the inclusion continues to hold when replacing $\mathsf{TO}$ by its closure.

(iv): There are two steps to this proof: first we show that the r.h.s.~of \eqref{eq:TO_connected_spectrum} is convex, followed by proving the more important fact that $\mathsf{TO}(H_S,T)$ is a subset of the convex hull of the right-hand side of \eqref{eq:TO_connected_spectrum}. The statement in question then follows from
\begin{align*}
\overline{\mathsf{TO}(H_S,T)}\subseteq \overline{\operatorname{conv}{(\text{r.h.s.~of }\eqref{eq:TO_connected_spectrum})}}&=\operatorname{conv}{\overline{(\text{r.h.s.~of }\eqref{eq:TO_connected_spectrum})}}\\
&=\operatorname{conv}{(\text{r.h.s.~of }\eqref{eq:TO_connected_spectrum})}=\text{r.h.s.~of }\eqref{eq:TO_connected_spectrum}\subseteq \overline{\mathsf{TO}(H_S,T)}\,.
\end{align*}
In the first equality we used that every bounded subset $A$ of a finite-dimensional vector space satisfies $\operatorname{conv}\overline{A}=\overline{\operatorname{conv}A}$.

Step 1: Convexity is proven just like convexity of $\overline{\mathsf{TO}}$ in (ii): first one shows that any rational convex combination is in the set exactly, and for irrational convex coefficients one at least ends up in the closure. The only thing that changes about the proof: one has to show that if $H_{B,1}\in{i}\mathfrak{u}(m_1)$, $H_{B,2}\in{i}\mathfrak{u}(m_2)$ have resonant spectrum w.r.t.~$H_S$, then so does $H_{B,1}\otimes\mathbbm1+\mathbbm1\otimes H_{B,2}$. 
%\marginpar{\color{red}An neue Definition anpassen!}
Indeed if $H_{B,i}=\operatorname{diag}((E_1')_i,\ldots,(E'_{m_i})_i)$, $i=1,2$ let us
%$\tau_i$ be a permutation such that $(E'_i)_{\tau(j+1)}-(E'_i)_{\tau(j)}$ is in $\sigma(\operatorname{ad}_{H_S})$ for all $j=1,\ldots,m_i-1$. Then we
write out $\sigma(H_{B,1}\otimes\mathbbm1+\mathbbm1\otimes H_{B,2})=\sigma(H_{B,1})+\sigma(H_{B,2})$ in the following way:
\begin{align*}
\big( (E'_{i_1})_{1}+(E'_{i_2})_{2}\big)_{i_1=1, i_2=1}^{m_1,m_2}=\begin{pmatrix}
(E'_1)_{1}+(E'_1)_{2}& (E'_1)_{1}+(E'_2)_{2}&\cdots&(E'_1)_{1}+(E'_{m_2})_{2}\\
(E'_2)_{1}+(E'_1)_{2}&(E'_2)_{1}+(E'_2)_{2}&\ddots&\vdots\\
\vdots&\ddots&\ddots&\vdots\\
(E'_{m_1})_{1}+(E'_1)_{2}&\cdots&\cdots&(E'_{m_1})_{1}+(E'_{m_2})_{2}
\end{pmatrix}
\end{align*}
Now given an arbitrary proper non-empty subset $I$ of $\{1,\ldots,m_1\}\times\{1,\ldots,m_2\}$ there certainly exists either a row or a column in the above matrix which features indices $(i_1,i_2)$ from $I$ as well as indices from the complement $I^{\mathsf{c}}$ of $I$. If we assume w.l.o.g.~that this property is satisfied by the row $k\in\{1,\ldots,m_2\}$, this means that
$$
I_k:=I\cap\{(j,k):j=1,\ldots,m_1\}\neq\emptyset\quad\text{ and }\quad \{(j,k):j=1,\ldots,m_1\}\not\subset I\,.
$$
In particular $I_k$ is a proper non-empty subset of $\{1,\ldots,m_1\}\times\{k\}$ so because $H_{B_1}$ is resonant w.r.t.~$H_S$ there exist $(i_1,k)\in I_k$, $(j_1,k)\in I_k^{\mathsf{c}}$ such that $(E_{i_1}')_1-(E_{j_1}')_1\in\sigma(\operatorname{ad}_{H_S})$. Therefore
$$
\big((E_{i_1}')_1+(E_k')_2\big)-\big((E_{j_1}')_1+(E_k')_2\big)=(E_{i_1}')_1-(E_{j_1}')_1\in\sigma(\operatorname{ad}_{H_S})
$$
which -- because $(i_1,k)\in I$, $(j_1,k)\in I^{\mathsf{c}}$ -- shows that $H_{B,1}\otimes\mathbbm1+\mathbbm1\otimes H_{B,2}$ is resonant with respect to $H_S$ as claimed.

%Every column constitutes something connected because $H_{B,1}$ is resonant, and every row is connected because $H_{B,2}$ is resonant. Therefore $H_{B,1}\otimes\mathbbm1+\mathbbm1\otimes H_{B,2}$ is resonant w.r.t.~$H_S$, as well.

Step 2: By Lemma \ref{lemma_3} we can assume w.l.o.g.~that $H_S$ is diagonal in the standard basis, i.e.~$H_S=\operatorname{diag}(E_1,\ldots,E_n)$. This will make defining certain objects less tedious. Now let $H_B=\operatorname{diag}(E_1',\ldots,E_m')\in{i}\mathfrak u(m)$, $U\in\mathsf U(mn)$ be given such that $H_B$ does not have resonant spectrum w.r.t.~$H_S$ (else we would be done). Hence there exists $I\subsetneq\{1,\ldots,m\}$, $I\neq\emptyset$ such that
$$
\big\{E_i'-E_j':i\in I,j\in\{1,\ldots,m\}\setminus I\big\}\cap\sigma(\operatorname{ad}_{H_S})=\emptyset\,.
$$
We will show that this partition of the index set $\{1,\ldots,m\}$ implies a decomposition of $H_B,U$ into smaller submatrices such that the resulting thermal operations recreate the original map via a convex combination. More precisely we define
$$
H_{B,1}:=\operatorname{diag}(E_i')_{i\in I}\quad,\quad H_{B,2}=\operatorname{diag}(E_j')_{j\in\{1,\ldots,m\}\setminus I}
$$
as well as
$$
U_1:=(\langle e_i,Ue_j\rangle)_{i,j\in I'}\quad,\quad U_2:=(\langle e_i,Ue_j\rangle)_{i,j\in\{1,\ldots,mn\}\setminus I'}
$$
where $I':= \{j+(k-1)m:j\in I,k=1,\ldots,n\} $. We claim that
\begin{equation}\label{eq:proof_lemma_iv_decomp}
\Phi_{T,m}(H_B,U)=\frac{\operatorname{tr}( e^{-H_{B,1}/T} )}{\operatorname{tr}( e^{-H_{B}/T})}\Phi_{T,|I|}(H_{B,1},U_1)+\frac{\operatorname{tr}( e^{-H_{B,2}/T} )}{\operatorname{tr}( e^{-H_{B}/T})}\Phi_{T,m-|I|}(H_{B,2},U_2)\,.
\end{equation}
The easiest way to see this is by decomposing $H_B$ into blocks. Define $\Pi_I:=\sum_{i\in I}|e_i\rangle\langle e_i|$ and note that $[\Pi_I,H_B]=0$. We compute
\begin{align*}
H_B&=(\Pi_I+(\mathbbm1-\Pi_I))H_B(\Pi_I+(\mathbbm1-\Pi_I))\\
&=\Pi_I H_B\Pi_I + [H_B,\Pi_I]\Pi_I+\Pi_I[\Pi_I ,H_B] +(\mathbbm1-\Pi_I )H_B(\mathbbm1-\Pi_I )\\
&=\Pi_I H_B\Pi_I +(\mathbbm1-\Pi_I )H_B(\mathbbm1-\Pi_I )\,.
\end{align*}
This ``block structure'' of $H_B$ carries over to $U$ via energy conservation: Given $a,b\in\{1,\ldots,n\}$ as well as $i\in I$, $j\in\{1,\ldots,m\}\setminus I$ the energy-conservation condition implies
\begin{align*}
0&=\langle e_a\otimes e_i,[U,H_S\otimes\mathbbm1+\mathbbm1\otimes H_B]e_b\otimes e_j\rangle\\
&=\big((E_b-E_a)-(E_i'-E_j')\big)\langle e_a\otimes e_i,U(e_b\otimes e_j)\rangle\,.
\end{align*}
But $E_i'-E_j'\not\in\sigma(\operatorname{ad}_{H_S})$ by assumption while $E_b-E_a\in\sigma(\operatorname{ad}_{H_S})$ meaning the prefactor is non-zero; hence $\langle e_a\otimes e_i,U(e_b\otimes e_j)\rangle=0$. Therefore
\begin{align*}
(\mathbbm1\otimes \Pi_I)U(\mathbbm1-(\mathbbm1\otimes \Pi_I))=\sum_{a,b=1}^n\sum_{\substack{i\in I\\j\in\{1,\ldots,m\}\setminus I}}\langle e_a\otimes e_i,U(e_b\otimes e_j)\rangle|e_a\otimes e_i\rangle\langle e_b\otimes e_j|=0
\end{align*}
and similarly for $(\mathbbm1-(\mathbbm1\otimes \Pi_I))U(\mathbbm1\otimes \Pi_I)$. This -- just as for $H_B$ -- yields the block-decomposition 
%\begin{align*}
$
U
%&= \big((\mathbbm1\otimes \Pi_I)+(\mathbbm1-(\mathbbm1\otimes \Pi_I))\big)U\big((\mathbbm1\otimes \Pi_I)+(\mathbbm1-(\mathbbm1\otimes \Pi_I))\big)\\
%&=(\mathbbm1\otimes \Pi_I)U(\mathbbm1\otimes \Pi_I)+(\mathbbm1\otimes (1-\Pi_I))U(\mathbbm1\otimes (1-\Pi_I))\\
%&=(\mathbbm1\otimes \Pi_I)U(\mathbbm1\otimes \Pi_I)+(\mathbbm1-(\mathbbm1\otimes \Pi_I))U(\mathbbm1-(\mathbbm1\otimes \Pi_I))\\
=\Pi_{I'}U\Pi_{I'}+(\mathbbm1-\Pi_{I'})U(\mathbbm1-\Pi_{I'})
$
%\end{align*}
where
\begin{align*}
\Pi_{I'}:=\sum_{i\in I'}|e_i\rangle\langle e_i|&=\sum_{j\in I}\sum_{k=1}^n|e_{j+(k-1)m}\rangle\langle e_{j+(k-1)m}|\\
&=\sum_{j\in I}\sum_{k=1}^n|e_k\otimes e_j\rangle\langle e_k\otimes e_j|\\
&=\Big(\sum_{k=1}^n|e_k\rangle\langle e_k|\Big)\otimes\Big(\sum_{j\in I}|e_j\rangle\langle e_j|\Big)=\mathbbm1\otimes\Pi_I\,.
\end{align*}
Inserting this decomposition of $U,H_B$ into the definition of the associated thermal operation yields
\begin{align*}
\Phi_{T,m}(H_B,U)&= \operatorname{tr}_B\circ\operatorname{Ad}_{\Pi_{I'}U\Pi_{I'}+(\mathbbm1-\Pi_{I'})U(\mathbbm1-\Pi_{I'})}\circ \Big((\cdot)\otimes \tfrac{\Pi_Ie^{-H_B/T}\Pi_I+(\mathbbm1-\Pi_I)e^{-H_B/T}(\mathbbm1-\Pi_I)}{\operatorname{tr}(e^{-H_B/T})}\Big)\\
&= \operatorname{tr}_B\circ\operatorname{Ad}_{\Pi_{I'}U\Pi_{I'}}\circ \Big((\cdot)\otimes \tfrac{\Pi_I e^{-H_B/T}\Pi_I }{\operatorname{tr}(e^{-H_B/T})}\Big)\\
&+ \operatorname{tr}_B\circ\operatorname{Ad}_{(\mathbbm1-\Pi_{I'})U(\mathbbm1-\Pi_{I'})}\circ \Big((\cdot)\otimes \tfrac{(\mathbbm1-\Pi_I )e^{-H_B/T}(\mathbbm1-\Pi_I )}{\operatorname{tr}(e^{-H_B/T})}\Big)\\
&=\tfrac{\operatorname{tr}(\Pi_I e^{-H_B/T}\Pi_I )}{\operatorname{tr}(e^{-H_B/T})} \operatorname{tr}_B\circ\operatorname{Ad}_{\Pi_{I'}U\Pi_{I'}}\circ \Big((\cdot)\otimes \tfrac{\Pi_I e^{-H_B/T}\Pi_I }{\operatorname{tr}(\Pi_I e^{-H_B/T}\Pi_I )}\Big)\\
&+(1-\tfrac{\operatorname{tr}(\Pi_I e^{-H_B/T}\Pi_I )}{\operatorname{tr}(e^{-H_B/T})} )\operatorname{tr}_B\circ\operatorname{Ad}_{(\mathbbm1-\Pi_{I'})U(\mathbbm1-\Pi_{I'})}\circ \Big((\cdot)\otimes \tfrac{(\mathbbm1-\Pi_I )e^{-H_B/T}(\mathbbm1-\Pi_I )}{\operatorname{tr}((\mathbbm1-\Pi_I )e^{-H_B/T}(\mathbbm1-\Pi_I ))}\Big)
\end{align*}
where we used again that $\Pi_{I'}(\mathbbm1\otimes(\mathbbm1-\Pi_I))=0=(\mathbbm1-\Pi_{I'})(\mathbbm1\otimes\Pi_I)$. But now the second-to-last channel (without the pre-factor) is indistinguishable from $\Phi_{T,|I|}(H_{B,1},U_1)$, and the same holds for the channel in the last line and $\Phi_{T,m-|I|}(H_{B,2},U_2)$. The reason for this is that $\Pi_Ie^{-H_B/T}\Pi_I$ and $e^{-H_{B,1}/T}$ as well as $\Pi_I'U\Pi_I'$ and $U_1$ have the same non-zero entries in the same ``order'', hence everything else in $\Pi_Ie^{-H_B/T}\Pi_I$, $\Pi_I'U\Pi_I'$ can be disregarded without changing the map. More precisely one may use the decompositions
$$
H_{B,1}=\operatorname{diag}(\langle e_i,H_Be_i\rangle)_{i\in I}=\sum_{j=1}^{|I|}\langle e_{\iota(j)},H_Be_{\iota(j)}\rangle|e_j\rangle\langle e_j|
$$
and
$$
U_1=\sum_{a,b=1}^n\sum_{i,j=1}^{|I|}\langle e_a\otimes e_{\iota(i)},U(e_b\otimes e_{\iota(j)}\rangle|e_a\rangle\langle e_b|\otimes|e_i\rangle\langle e_j|
$$
when enumerating $I=\{\iota(1),\ldots,\iota(|I|)\}$, $\iota(1)<\ldots<\iota(|I|)$ so $\iota:\{1,\ldots,|I|\}\to I$ is bijective and order-preserving. The same is done for $H_{B,2},U_2$. In total this proves \eqref{eq:proof_lemma_iv_decomp} by means of a direct computation.
The proof is concluded by the observation that $U_1,U_2$ are unitary because $\Pi_{I'}U\Pi_{I'}$, $(\mathbbm1-\Pi_{I'})U(\mathbbm1-\Pi_{I'})$ are partial isometries due to the ``block-form'' of $U$, as well as the fact that $\Pi_{I'}[U,H_S\otimes\mathbbm1+\mathbbm1\otimes H_B]\Pi_{I'}=[\Pi_{I'}U\Pi_{I'},H_S\otimes\mathbbm1+\mathbbm1\otimes H_B]$ has the same non-zero entries as $[U_1,H_S\otimes\mathbbm1+\mathbbm1\otimes H_{B,1}]$; but the former is zero due to energy conservation meaning $U_1$ is energy-conserving w.r.t.~$H_{B,1}$ (similarly for $U_2$, $H_{B,2}$).\qed
\section{Proof of Corollary \ref{prop_to_spin}}\label{app_proof_prop_to_spin}
W.l.o.g.~$H_S=\bigoplus_{j=0}^{n-1}\,(E_1+j\Delta E)\,\mathbbm1_{\alpha_j}$ for some $E_1\in\mathbb R$, $\Delta E>0$, and $\alpha_0,n\in\mathbb N$, $\alpha_1,\ldots,\alpha_{n-1}\in\mathbb N_0$; the rest is just a basis change which Lemma \ref{lemma_3} takes care of.

All we have to prove is identity \eqref{eq:equiv_TO_spin} for $T\in(0,\infty)$ (as the case $T=\infty$ is done analogously) because the second statement of the proposition is a special case of Proposition \ref{prop_1} (iv): if something has resonant spectrum w.r.t.~a spin Hamiltonian (with gaps) then it has to be of the same ``spin form''. Indeed if $H_B\in{i}\mathfrak u(m)$ is resonant w.r.t.~the above $H_S$ then the difference of any pair of eigenvalues of $H_B$ is a multiple of $\Delta E$ so there exist $\alpha_1,\ldots,\alpha_m\in\mathbb N_0$, $c\in\mathbb R$ such that $\sigma(H_B)=\{\alpha_j\cdot\Delta E:j=1,\ldots,m\}+c$ (the global shift $c$ can be disregarded as such a shift does not change the corresponding thermal operation). We will show this via contraposition: Assume that $H_B\in{i}\mathfrak u(m)$ has two eigenvalues $E_a',E_b'$ the difference of which is not a multiple of $\Delta E$. Then the set $I:=\{1,\ldots,m:\exists_{\alpha\in\mathbb Z}\ E_i'-E_b'=\alpha\Delta E\}$ is a proper ($a\not\in I$) non-empty ($b\in I$) subset of $\{1,\ldots,m\}$. But by definition of $I$, for all $i\in I$, $j\in\{1,\ldots,m\}\setminus I$ one finds that
$$
E_i'-E_j'=\underbrace{E_i'-E_b'}_{\text{multiple of }\Delta E}+\underbrace{E_b'-E_j'}_{\text{not a multiple of }\Delta E}
$$
is not a multiple of $\Delta E$. Hence $E_i'-E_j'$ cannot be an element of $\sigma(\operatorname{ad}_{H_S})$ which shows that $H_B$ is not resonant w.r.t.~$H_S$.

As for the first equation in \eqref{eq:equiv_TO_spin}: while ``$\subseteq$'' is obvious, for ``$\supseteq$'' we have to show that it is possible to approximate any thermal operation with bath Hamiltonian $H_B=\bigoplus_{j=1}^m j\Delta E\,\mathbbm1_{\beta_j}$ (where some of the $\beta_j$ can be zero) using a Hamiltonian where $\beta_j\geq 1$ for all $j$. Given arbitrary $\alpha\in\mathbb N$ define $J:=\{j\in\{1,\ldots,m\}:\beta_j=0\}$ and 
$H_{B,\alpha}':=\underline{\tau_\alpha}((\bigoplus_{k=1}^\alpha H_B)\oplus\operatorname{diag}(j\Delta E)_{j\in J})\underline{\tau_\alpha^{-1}}$ where $\tau_\alpha$ is any permutation\footnote{
Given some permutation ${\pi}\in S_n$ the corresponding permutation matrix \unexpanded{$\underline{\pi}$} is 
given by \unexpanded{$\sum_{i=1}^n |e_i\rangle\langle e_{{\pi}(i)}|$}. In particular the identities
\unexpanded{$(\underline{{\pi}}x)_j=x_{{\pi}(j)}$} and \unexpanded{$\underline{{\pi}\circ\tau}=\underline{\tau}\cdot\underline{{\pi}}$} hold for all $\pi,\tau\in S_n$, $x\in\mathbb C^n$, $j\in\{1,\ldots,n\}$.
%\label{footnote_permutation}
}
such that the diagonal of $H_{B,\alpha}'$ is sorted increasingly; thus $H_{B,\alpha}'$ is of the required form. The total Hamiltonian will be decomposed into $n\times n$ blocks\footnote{
Given $m,n\in\mathbb N$, $A\in\mathbb C^{nm\times nm}$ and any orthonormal basis $(g_j)_{j=1}^n$ of $\mathbb C^n$ one can decompose \unexpanded{$A =\sum_{i,j=1}^n|g_i\rangle\langle g_j|\otimes A_{ij}$} where \unexpanded{$A_{ij}:=\operatorname{tr}_1((|g_j\rangle\langle g_i|\otimes\mathbbm1)A)$} for all \unexpanded{$i,j=1,\ldots,n$}. 
}
of equal size $(\sum_{j=1}^m\beta_j)\times (\sum_{j=1}^m\beta_j)$, that is, $H_\mathsf{tot}=\sum_{i,j=1}^n |e_i\rangle\langle e_j|\otimes (H_\mathsf{tot})_{ij}$. With this we define
$$
H_{\mathsf{tot},\alpha}':=\sum_{i,j=1}^n |e_i\rangle\langle e_j|\otimes \underline{\tau_\alpha}\Big(\Big(\bigoplus_{k=1}^\alpha (H_\mathsf{tot})_{ij}\Big)\oplus 0_{|J|}\Big)\underline{\tau_\alpha^{-1}}
$$
and claim that the thermal operations generated by $H_B,H_\mathsf{tot}$ and by $H_{B,\alpha}',H_{\mathsf{tot},\alpha}'$, respectively, coincide in the limit $\alpha\to\infty$. To see that the latter actually generates a thermal operation one readily verifies
\begin{align*}
\big([H'_{\mathsf{tot},\alpha},H_S\otimes\mathbbm{1}+\mathbbm{1}\otimes H_{B,\alpha}']\big)_{ij}
%&=\big([(\mathbbm1\otimes\underline{\tau_\alpha})H_\mathsf{tot}(\mathbbm1\otimes\underline{\tau_\alpha^{-1}}),H_S\otimes\mathbbm{1}+\mathbbm{1}\otimes \underline{\tau_\alpha}H_B\underline{\tau_\alpha^{-1}}]\big)_{jk}\oplus 0_{|J|}\\
&=\underline{\tau_\alpha}\Big(\bigoplus_{k=1}^\alpha([H_\mathsf{tot},H_S\otimes\mathbbm{1}+\mathbbm{1}\otimes H_B])_{ij}\oplus 0_{|J|}\Big)\underline{\tau_\alpha^{-1}} 
\end{align*}
for all $i,j=1,\ldots,n$ so $[H'_{\mathsf{tot},\alpha},H_S\otimes\mathbbm{1}+\mathbbm{1}\otimes H_{B,\alpha}']=0$ is equivalent to $[H_\mathsf{tot},H_S\otimes\mathbbm{1}+\mathbbm{1}\otimes H_B]=0$. Moreover
\begin{align*}
S_\alpha:=& \operatorname{tr}_{B'}\Big(e^{iH'_{\mathsf{tot},\alpha}}\Big((\cdot)\otimes \frac{e^{-H_{B,\alpha}'/T}}{\operatorname{tr}(e^{-H_{B,\alpha}'/T})}\Big)e^{-iH'_{\mathsf{tot},\alpha}}\Big) \\
=& \frac{\sum_{i,j,k,l=1}^n \operatorname{tr}_{B'}( (|e_i\rangle\langle e_j|\otimes (e^{iH'_{\mathsf{tot},\alpha}})_{ij})((\cdot)\otimes e^{-H_{B,\alpha}'/T})(|e_k\rangle\langle e_l|\otimes(e^{iH'_{\mathsf{tot},\alpha}})_{kl}) ) }{\alpha(\sum_{j=1}^m e^{-j\Delta E/T}\beta_j)+\sum_{j\in J}e^{-j\Delta E/T}}\\
=& \frac{\sum_{i,j,k,l=1}^n \langle e_j,(\cdot)e_k\rangle\operatorname{tr}( (e^{iH'_{\mathsf{tot},\alpha}})_{ij} e^{-H_{B,\alpha}'/T})(e^{iH'_{\mathsf{tot},\alpha}})_{kl})|e_i\rangle\langle e_l| }{\alpha(\sum_{j=1}^m e^{-j\Delta E/T}\beta_j)+\sum_{j\in J}e^{-j\Delta E/T}}\\
=& \frac{\sum_{i,j,k,l=1}^n \langle e_j,(\cdot)e_k\rangle(\alpha\operatorname{tr}( (e^{iH_\mathsf{tot}})_{ij} e^{-H_B/T} (e^{-iH_\mathsf{tot}})_{kl} +\sum_{j\in J}e^{-j\Delta E/T} )|e_i\rangle\langle e_l| }{\alpha(\sum_{j=1}^m e^{-j\Delta E/T}\beta_j)+\sum_{j\in J}e^{-j\Delta E/T}}
\end{align*}
which in the limit $\alpha\to\infty$ yields
\begin{align*}
 \lim_{\alpha\to\infty}S_\alpha&=\frac{\sum_{i,j,k,l=1}^n \langle e_j,(\cdot)e_k\rangle\operatorname{tr}( (e^{iH_\mathsf{tot}})_{ij} e^{-H_B/T} (e^{-iH_\mathsf{tot}})_{kl} ) |e_i\rangle\langle e_l| }{\sum_{j=1}^m e^{-j\Delta E/T}\beta_j}\\
 =\ldots&= \operatorname{tr}_{B}\Big(e^{iH_{\mathsf{tot}}}\Big((\cdot)\otimes \frac{e^{-H_{B}/T}}{\operatorname{tr}(e^{-H_{B}/T})}\Big)e^{-iH_{\mathsf{tot}}}\Big)
\end{align*}
as claimed. Here we used that $(e^{iH'_{\mathsf{tot},\alpha}})_{ij}=\underline{\tau_\alpha}((\bigoplus_{k=1}^\alpha (e^{iH_\mathsf{tot}})_{ij})\oplus \mathbbm1_{|J|})\underline{\tau_\alpha^{-1}}$ which follows from the identity $ (H'_{\mathsf{tot},\alpha})^l_{ij}=\underline{\tau_\alpha}((\bigoplus_{k=1}^\alpha (H_\mathsf{tot})^l_{ij})\oplus 0^l_{|J|})\underline{\tau_\alpha^{-1}} $ for all $l\in\mathbb N_0$. The latter is readily verified by means of the block structure of $H'_{\mathsf{tot},\alpha}$.\qed

\section{Attempted Proof of Conjecture \ref{prop_cont_in_T}}\label{proof_att_conj_cont_T}
%, $\delta(\overline{\mathsf{TO}(H_S,T)},\overline{\mathsf{TO}(H_S,T')})$ equals
%\begin{align*}
%&\max\Big\{ \max_{S\in\overline{\mathsf{TO}(H_S,T)}}\min_{S'\in\overline{\mathsf{TO}(H_S,T')}}\|S-S'\|_{1\to 1},\max_{S'\in\overline{\mathsf{TO}(H_S,T')}}\min_{S\in\overline{\mathsf{TO}(H_S,T)}}\|S-S'\|_{1\to 1} \Big\}
%%\\
%%\leq &\,n\max\Big\{ \max_{S\in\overline{\mathsf{TO}(H_S,T)}}\min_{S'\in\overline{\mathsf{TO}(H_S,T')}}\|S-S'\|_{\infty\to \infty},\max_{S'\in\overline{\mathsf{TO}(H_S,T')}}\min_{S\in\overline{\mathsf{TO}(H_S,T)}}\|S-S'\|_{\infty\to \infty} \Big\}
%\,.
%\end{align*}
%Here we used that $\|\cdot\|_{1\to 1}\leq n\|\cdot\|_{\infty\to\infty}$ which follows from $\|A\|_1\leq n\|A\|_\infty$ for all $A\in\mathbb C^{n\times n}$.
It would suffice to find a map $c:(0,\infty]\times(0,\infty]\to[0,\infty)$ such that the following holds:
\begin{itemize}
\item $c(T,T)=0$ for all $T\in(0,\infty]$
\item $c$ is continuous with respect to $d_{-1}$
\item Given any $\varepsilon>0$, $T,T'\in(0,\infty]$ there for all $S\in\overline{\mathsf{TO}(H_S,T)}$ exists $S'\in\overline{\mathsf{TO}(H_S,T')}$ such that
$
\|S-S'\|_{1\to 1}< \varepsilon+ 
c(T,T')
$
(and vice versa)
\end{itemize}
Then by definition of the Hausdorff metric
$
\delta\big(\overline{\mathsf{TO}(H_S,T)},\overline{\mathsf{TO}(H_S,T')}\big)\leq c(T,T')
$
which for all $T\in(0,\infty]$ would imply
$$
\delta\big(\overline{\mathsf{TO}(H_S,T)},\overline{\mathsf{TO}(H_S,T')}\big)\to 0\qquad\text{ as }\qquad d_{-1}(T,T')\to 0
$$
as desired.
Now given $\varepsilon>0$, $S\in\overline{\mathsf{TO}(H_S,T)}$ there exist $H_B\in{i}\mathfrak{u}(m)$ and $U\in\mathsf{U}(mn)$ such that
$$
\Big\|S- \operatorname{tr}_B\Big(U\Big((\cdot)\otimes \frac{e^{-H_B/T}}{\operatorname{tr}(e^{-H_B/T})}\Big)U^*\Big)\Big\|_{1\to 1}<\varepsilon
$$
where $U$ satisfies $U(H_S\otimes\mathbbm{1}_B+\mathbbm{1}\otimes H_B)U^*=H_S\otimes\mathbbm{1}_B+\mathbbm{1}\otimes H_B $.
The simplest way of picking an element in $ \mathsf{TO}(H_S,T')$ which is ``close to'' this approximation of $S$ is to define the channel $S':= \operatorname{tr}_B(U((\cdot)\otimes \frac{e^{-H_B/T'}}{\operatorname{tr}(e^{-H_B/T'})})U^*)$.
This yields the estimate
\begin{align}
\|S-S'\|_{1\to 1}&\leq \Big\|S- \operatorname{tr}_B\Big(U\Big((\cdot)\otimes \frac{e^{-H_B/T}}{\operatorname{tr}(e^{-H_B/T})}\Big)U^*\Big)\Big\|_{1\to 1}\notag\\
&+\Big\| \operatorname{tr}_B\Big(U\Big((\cdot)\otimes \frac{e^{-H_B/T}}{\operatorname{tr}(e^{-H_B/T})}\Big)U^*\Big)-\operatorname{tr}_B\Big(U\Big((\cdot)\otimes \frac{e^{-H_B/T'}}{\operatorname{tr}(e^{-H_B/T'})}\Big)U^*\Big)\Big\|_{1\to 1}\notag\\
&<\varepsilon+\Big\| \operatorname{tr}_B\Big(U\Big((\cdot)\otimes \Big(\frac{e^{-H_B/T}}{\operatorname{tr}(e^{-H_B/T})}-\frac{e^{-H_B/T'}}{\operatorname{tr}(e^{-H_B/T'})}\Big)U^*\Big)\Big\|_{1\to 1}\notag\\
&<\varepsilon+\|\operatorname{tr}_B\|_{1\to 1}\|\operatorname{Ad}_U\|_{1\to 1}\Big\|\frac{e^{-H_B/T}}{\operatorname{tr}(e^{-H_B/T})}-\frac{e^{-H_B/T'}}{\operatorname{tr}(e^{-H_B/T'})}\Big\|_1\label{eq:S_Sprime_final_estimate}
%\\
%&\leq \varepsilon+m\Big\|\frac{e^{-H_B/T}}{\operatorname{tr}(e^{-H_B/T})}-\frac{e^{-H_B/T'}}{\operatorname{tr}(e^{-H_B/T'})}\Big\|_\infty
\,.
\end{align}
Keeping in mind that
%$\|A\|_1\leq n\|A\|_\infty$ for all $A\in\mathbb C^{n\times n}$ as well as the fact that
$\operatorname{tr}_B$, $\operatorname{Ad}_U$ have operator norm one (w.r.t.~the trace norm) because they are (completely) positive and trace-preserving (Thm.~2.1 in \cite{PG06})
it seems reasonable to consider
\begin{equation}\label{eq:def_c_map}
\begin{split}
c:(0,\infty]\times(0,\infty]&\to[0,\infty)\\
(T,T')&\mapsto\sup_{m\in\mathbb N}\sup_{H_B\in{i}\mathfrak u(m)}\Big\|\frac{e^{-H_B/T}}{\operatorname{tr}(e^{-H_B/T})}-\frac{e^{-H_B/T'}}{\operatorname{tr}(e^{-H_B/T'})}\Big\|_1
\end{split}
\end{equation}
for an upper bound.
In other words the above estimate \textit{would} reduce the problem of continuity of the thermal operations to continuity of certain Gibbs states in the temperature. 
However \eqref{eq:def_c_map} already looks unsuited for the task as it does not feature the system's Hamiltonian anymore.
Indeed as soon as $T,T'$ do not coincide then $c(T,T')$ takes the largest possible value:
\begin{lemma}\label{lemma_c_max}
For all $T,T'\in(0,\infty]$ one has $c(T,T')=2\delta_{T,T'}$.
%$$
%\sup_{H_B\in{i}\mathfrak u(m)}\Big\| \frac{e^{-H_B/T}}{\operatorname{tr}(e^{-H_B/T})}-\frac{e^{-H_B/T'}}{\operatorname{tr}(e^{-H_B/T'})} \Big\|_1=2\,.
%$$
\end{lemma}
\begin{proof}
W.l.o.g.~$T<T'$ so one has $e^{-E/T}<e^{-E/T'}$ for all $E>0$. 
Now given any $E>0$ define the Hamiltonian
$
H_B(E):=0\oplus (E\cdot\mathbbm1_{\lfloor e^{E/2T}e^{E/2T'}\rfloor})\in{i}\mathfrak u( 1+\lfloor e^{E/2T}e^{E/2T'}\rfloor )
$.
We claim that
\begin{equation}\label{eq:conv_energy}
\lim_{E\to\infty}\Big(2-\Big\| \frac{e^{-H_B(E)/T}}{\operatorname{tr}(e^{-H_B(E)/T})}-\frac{e^{-H_B(E)/T'}}{\operatorname{tr}(e^{-H_B(E)/T'})} \Big\|_1\Big)=0\,.
\end{equation}
Indeed a straightforward computation shows
\begin{align*}
2-\Big\| \frac{e^{-H_B/T}}{\operatorname{tr}(e^{-H_B/T})}-\frac{e^{-H_B/T'}}{\operatorname{tr}(e^{-H_B/T'})} \Big\|_1&=2\Big(1- \frac{\lfloor e^{E/2T}e^{E/2T'}\rfloor(e^{-E/T'}-e^{-E/T})}{(1+\lfloor e^{E/2T}e^{E/2T'}\rfloor e^{-E/T})(1+\lfloor e^{E/2T}e^{E/2T'}\rfloor e^{-E/T'})} \Big)\\
&=2\Big( \frac{1}{1+\lfloor e^{E/2T}e^{E/2T'}\rfloor e^{-E/T'}}+\frac{1}{\lfloor e^{E/2T}e^{E/2T'}\rfloor^{-1}e^{E/T}+1} \Big)\\
&\leq 2\Big(\frac{1}{1+e^{E/2T}e^{-E/2T'}-e^{-E/T'}}+\frac{1}{e^{E/2T}e^{-E/2T'}+1} \Big)\,.
\end{align*}
But $e^{E/2T}e^{-E/2T'}=e^{\frac{E}2(\frac1T-\frac1{T'})}\to\infty$ as $E\to\infty$ because $\frac1T-\frac1{T'}>0$ by assumption. Moreover the expression in \eqref{eq:conv_energy} is non-negative (by the triangle inequality) so because the upper bound we found vanishes as $E\to\infty$, \eqref{eq:conv_energy} holds. This concludes the proof.
\end{proof}
There are two ways out of this dilemma: On the one hand one could restrict the supremum in \eqref{eq:def_c_map} to a smaller generating set
%${\bf H}\subseteq\bigcup_{m\in\mathbb N}{i}\mathfrak u(m)$ 
of the thermal operations (resp.~its closure), for example ${\bf H}:=\bigcup_{m\in\mathbb N}\{H_B\in{i}\mathfrak u(m): H_B\text{ is resonant w.r.t.~}H_S \}$.
This would invalidate the current proof of Lemma \ref{lemma_c_max}, the key of which was to let the gaps between neighboring eigenvalues of $H_B$ become arbitrarily large -- but the resonance condition prohibits this. However, modifying $c(T,T')$ to be $\sup_{H_B\in{\bf H}}\|\frac{e^{-H_B/T}}{\operatorname{tr}(e^{-H_B/T})}-\frac{e^{-H_B/T'}}{\operatorname{tr}(e^{-H_B/T'})}\|_1$ -- while seeming more suited to be an upper bound as $H_S$ now appears at least implicitly in $c$ -- does make it more difficult to study $c$ due to the more complicated structure of $\bf H$.

%$$
%c(T,T'):=\sup_{m\in\mathbb N}\sup_{\{H_B\in{i}\mathfrak u(m): H_B\text{ is resonant w.r.t.~}H_S \}}\Big\|\frac{e^{-H_B/T}}{\operatorname{tr}(e^{-H_B/T})}-\frac{e^{-H_B/T'}}{\operatorname{tr}(e^{-H_B/T'})}\Big\|_1
%$$

On the other hand \eqref{eq:S_Sprime_final_estimate} might be too poor an estimate 
for studying continuity of $T\mapsto\overline{ \mathsf{TO}(H,T) }$: Although $T$ only 
ever appears in the Gibbs state it may be that separating the fundamental building 
blocks of the thermal operations -- as done in \eqref{eq:S_Sprime_final_estimate} -- 
loses too much of its structure, even if one is only interested in the effect of the 
temperature.
Either way it seems that proving Conjecture \ref{prop_cont_in_T}---if true at all---requires a more careful analysis of the effect which changing the temperature can have on the set of thermal operations.
\section{Proof of Theorem \ref{thm_equiv_qubit}}\label{app_proof_thm_equiv_qubit}
The following lemma which is indispensable for qubit computations is verified directly:
\begin{lemma}\label{lemma_degenerate_spin_TO}
%\marginpar{\color{red}$\mathbb N_{\to 0\leftarrow}$ necessary? also lemma nach vorne?}
Let $T\in(0,\infty]$, $\Delta E>0$, $m\in\mathbb N$, $\alpha_0,\ldots,\alpha_{m-1}\in\mathbb N$, as well as unitaries $U_0\in\mathbb C^{\alpha_0\times\alpha_0}$, $U_m\in\mathbb C^{\alpha_{m-1}\times\alpha_{m-1}}$, and $U_j\in\mathbb C^{( \alpha_{j-1}+\alpha_{j} )\times(\alpha_{j-1}+\alpha_{j} )}$, $j=1,\ldots,m-1$ be given. Decompose
\begin{equation}\label{eq:U_j_unitary}
U_j=\begin{pmatrix} A_j&B_j\\C_j&D_j \end{pmatrix}
\end{equation}
with $A_j\in\mathbb C^{\alpha_j\times\alpha_j}, B_j\in\mathbb C^{\alpha_j\times\alpha_{j-1}},C_j\in\mathbb C^{\alpha_{j-1}\times\alpha_j},D_j\in\mathbb C^{\alpha_{j-1}\times\alpha_{j-1}}$ for all $j=1,\ldots,m-1$. Defining $H_B:=\bigoplus_{j=0}^{m-1}j\Delta E\cdot\mathbbm1_{\alpha_j}$,
$$
U:=\begin{pmatrix} U_0&0&0&0\\0&\bigoplus_{j=1}^{m-1}A_j&\bigoplus_{j=1}^{m-1}B_j&0\\0&\bigoplus_{j=1}^{m-1}C_j&\bigoplus_{j=1}^{m-1}D_j&0\\0&0&0&U_m \end{pmatrix}\in\mathbb C^{( 2\alpha_0+\ldots+2\alpha_{m-1} )\times( 2\alpha_0+\ldots+2\alpha_{m-1} )}\,,
$$
and $S_U:=\operatorname{tr}_B(U((\cdot)\otimes\frac{e^{-H_B/T}}{\operatorname{tr}(e^{-H_B/T})})U^*)$ (resp.~$S_U:=\operatorname{tr}_B(U((\cdot)\otimes\frac{\mathbbm1}{\sum_{j=0}^{m-1}\alpha_j})U^*)$ if $T=\infty$) one finds that $U$ is unitary, $S_U\in\mathsf{TO}(\operatorname{diag}(0,1),T)$, and the Choi matrix of $S_U$ reads
\begin{equation}\label{eq:Choi_S_U}
\begin{pmatrix} 1-\lambda e^{-\Delta E/T}&0&0&c\\0&\lambda e^{-\Delta E/T}&0&0\\0&0&\lambda&0\\c^*&0&0&1-\lambda \end{pmatrix}
\end{equation}
with
\begin{align*}
\lambda&= \frac{\sum_{j=0}^{m-2}\operatorname{tr}(B_{j+1}B_{j+1}^*)e^{-j\Delta E/T}}{\sum_{j=0}^{m-1}\alpha_j e^{-j\Delta E/T}}\\
c&= \frac{\operatorname{tr}(U_0D_1^*)+\sum_{j=1}^{m-2}\operatorname{tr}(A_jD_{j+1}^*)e^{-j\Delta E/T}+\operatorname{tr}(A_{m-1}U_m^*)e^{-(m-1)\Delta E/T}}{\sum_{j=0}^{m-1}\alpha_j e^{-j\Delta E/T}}\,.
\end{align*}
If $T=\infty$, then Eq.~\eqref{eq:Choi_S_U} and the succeeding formulae continue to hold if $e^{-\Delta E/T}$ is replaced by $1$.
\end{lemma}
With this we are ready to prove Theorem \ref{thm_equiv_qubit}:
\begin{proof}
(i): First let us review how \'Cwikli\'nski et al.~showed (in Supplementary Note 4, Section IV in \cite{Cwiklinski15}) that $
\mathsf{EnTO}(H_S,T)=\overline{\mathsf{TO}(H_S,T)}$ for all non-degenerate Hamiltonians\footnote{
The non-degenerate case, that is, $H_S=\mathbbm1_2$ is a direct consequence of the fact that in two dimensions the set of all unital quantum maps ($\mathsf{EnTO}(\mathbbm1,T)$) equals the convex hull of all unitary channels ($\operatorname{conv}{\overline{\mathsf{TO}(\mathbbm1,T)}}=\overline{\mathsf{TO}(\mathbbm1,T)}$, Proposition \ref{prop_1}), cf.~\cite{LS93}.
}${}^\text{,}$\cite{LS93}
in two dimensions.
This will allow us to highlight how 
%their proof differs from ours and how 
one gets around using (highly) degenerate bath Hamiltonians (statement (ii) \& (iii) of this theorem).

By Lemma \ref{lemma_3} w.l.o.g.~$H_S=\operatorname{diag}(0,\Delta E)$ with $\Delta E>0$. Given $T\in(0,\infty)$ (we treat $T=\infty$ separately), $\lambda\in[0,1]$ what they do is construct a family $(S_m^\mu)_{m\in\mathbb N,\mu\in(1,e^{1/T})\cap\mathbb Q}\in\mathsf{TO}(H_S,T)$ such that
$$
\big(\lim_{\mu\to(e^{\Delta E/T})^-}\lim_{m\to\infty} S_m^\mu\big)(A)=\begin{pmatrix}a_{11}(1-\lambda e^{-\Delta E/T})+\lambda a_{22} &\sqrt{(1-\lambda)(1-\lambda e^{-\Delta E/T})}a_{12}\\
\sqrt{(1-\lambda)(1-\lambda e^{-\Delta E/T})}a_{21}&\lambda e^{-\Delta E/T}a_{11}+(1-\lambda)a_{22} \end{pmatrix}
$$
for all $A\in\mathbb C^{2\times 2}$. What this means is that -- together with the fact that the channel which only applies a phase to the off-diagonals is in $\mathsf{TO}$ ($m=1$ in the definition) -- the extreme points of $\mathsf{EnTO}$ (i.e.~the boundary in Fig.~\ref{fig1} without the inner area of the circle at the bottom) are in $\overline{\mathsf{TO}}$. From this one can deduce that the two sets have to coincide: either one uses that $\mathsf{EnTO}$, $\overline{\mathsf{TO}}$ are convex and compact, so
$$
\mathsf{EnTO}=\operatorname{conv}\big(\mathsf{ext}(\mathsf{EnTO})\big)\subseteq\operatorname{conv}\big(\mathsf{ext}(\overline{\mathsf{TO}})\big)=\overline{\mathsf{TO}}\subseteq\mathsf{EnTO}
$$
by Minkowski's theorem (Thm.~5.10 in \cite{Brondsted83}, where $\mathsf{ext}$ is the set of extreme points of a convex set), or one can show that any dephasing channel
\begin{equation}\label{eq:map_deph}
A\mapsto\begin{pmatrix} a_{11}&\overline{\gamma} a_{12}\\\gamma a_{21}&a_{22} \end{pmatrix}
\end{equation}
for $\gamma\in\mathbb C$, $|\gamma|\leq 1$ is in $\mathsf{TO}$ because then every thermal operation can be written as a composition of an extreme point of $\mathsf{EnTO}$ and a dephasing channel: simply choose $U\in\mathsf U(2)$ such that $\operatorname{tr}(U)=2\gamma$ because then $\Phi_{T,2}(\mathbbm1_2,\mathbbm1_2\oplus U)$ is in $\mathsf{TO}(H_S,T)$ for all $T\in(0,\infty]$. Then its action precisely given by \eqref{eq:map_deph}, cf.~also Chapter 8.3.6 in Nielsen \& Chuang \cite{NC10}. 

Now the construction of the maps $S_m^\mu$ goes as follows: Given $T\in(0,\infty)$, $\mu\in(1,e^{\Delta E/T})\cap\mathbb Q$, and $m\in\mathbb N$ define the following:
\begin{itemize}
\item $\alpha_0=\alpha_0(m,\mu)\in\mathbb N$ is the smallest integer such that $\alpha_0\mu^{m-1}\in\mathbb N$.
%Note that the Choi matrix of $S_m^\mu$ does not feature $\alpha_0$ for any choice of $m\in\mathbb N$, $\mu\in(1,e^{\Delta E/T})\cap\mathbb Q$, so
The only role of $\alpha_0$ is to ensure that the ratio of the size of consecutive blocks which make up the unitary matrix equals $\mu$, thus approximating $e^{\Delta E/T}$. Indeed $\alpha_0$ will not appear in the explicit action of $S_m^\mu$
\item $H_{B,m}:=\bigoplus_{j=0}^{m-1}j\Delta E\cdot\mathbbm1_{\alpha_0\mu^j}$
\item $D_1:=\mathbbm1_{\alpha_0}$ and, recursively, $A_j:=D_j\oplus\mathbbm1_{ \alpha_0\mu^{j-1}(\mu-1)}$ for all $j=1,\ldots,m-1$ as well as $D_j:=\sqrt{\frac{1-\lambda}{1-\frac{\lambda}{\mu}}} A_{j-1}$ for all $j=2,\ldots,m-1$
\end{itemize}
Because $\|A_j\|_\infty,\|D_j\|_\infty\leq 1$ (where $\|\cdot\|_\infty$ is the usual operator norm, that is, the largest singular value) it is easy to see that for all $j=1,\ldots,m-1$ one can choose $B_j,C_j$ such that

$$
U_j:=\begin{pmatrix} A_j&B_j\\C_j&D_j \end{pmatrix}
$$
is unitary, i.e.~$U_j\in\mathsf U( \alpha_0\mu^{j-1}(\mu+1))$. With this one defines
$$
U_m^\mu:=\begin{pmatrix} \mathbbm1_{\alpha_0}&0&0&0\\0&\bigoplus_{j=1}^{m-1}A_j&\bigoplus_{j=1}^{m-1}B_j&0\\0&\bigoplus_{j=1}^{m-1}C_j&\bigoplus_{j=1}^{m-1}D_j&0\\0&0&0&\mathbbm1_{\alpha_0\mu^{m-1}}\end{pmatrix}
$$
and $S_m^\mu:=\Phi_{T,(\alpha_0\sum_{j=0}^{m-1}\mu^j) }(H_{B,m},U_m^\mu)$. All one has to do now is compute the limit as stated above which using the representation $\Psi_T$ from Section \ref{sec_ento} comes out to be
\begin{align*}
\lim_{m\to\infty} \Psi_T(S_m^\mu)
=\begin{pmatrix}
 \frac{\lambda\mu (\mu-1)e^{-\Delta E/T}}{\mu-\lambda-\mu e^{-\Delta E/T}(1-\lambda)} \\
\sqrt{\mu (1-\lambda)(\mu-\lambda) }e^{-\Delta E/T}+( 1-\mu e^{-\Delta E/T})\big( 1+\frac{\sqrt{\mu(1-\lambda)(\mu-\lambda)}\lambda e^{-\Delta E/T}}{\mu-\lambda-\mu e^{-\Delta E/T}(1-\lambda)} \big)
\end{pmatrix}
\end{align*}
so, as claimed,
$$
\lim_{\mu\to(e^{\Delta E/T})^-}\lim_{m\to\infty} \Psi_T(S_m^\mu)
 =\begin{pmatrix} \lambda\\\sqrt{(1-\lambda)(1-\lambda e^{- \Delta E /T})} \end{pmatrix}\,.
$$
A particularly useful identity for verifying this is $\frac{\mu(1-\gamma^2)}{\mu-\gamma^2}=\lambda$ where $\gamma:=\sqrt{\frac{1-\lambda}{1-\frac{\lambda}{\mu}}}$.

This construction breaks down once $T$ is infinite for two reasons: first, the interval $(1,e^{\Delta E/T})$ from which we pick the rational approximation $\mu$ becomes empty and, more importantly, even if we just set $\mu=1$, then $U_m^\mu=U_m^1=\mathbbm1$ for all $m$; thus the corresponding thermal operation becomes trivial.
This is why we have to treat the case $T=\infty$ separately. Indeed given any $\lambda\in[0,1]$, $\phi\in[0,2\pi)$ choose $H_{B,m}:=\operatorname{diag}(j\Delta E)_{j=0}^{m-1}$ and
\begin{equation}\label{eq:T_infty_U}
U_m^\phi:=\begin{pmatrix}
1&0&0&0\\
0& \sqrt{1-\lambda}\mathbbm1_{m-1} & \sqrt{\lambda}\mathbbm1_{m-1} &0\\
0& -\sqrt{\lambda}e^{-i\phi}\mathbbm1_{m-1} & \sqrt{1-\lambda}e^{-i\phi}\mathbbm1_{m-1} &0\\
0&0&0&1
\end{pmatrix}\,.
\end{equation}
Obviously $U_m^\phi$ is energy-preserving w.r.t.~$(H_S,H_{B,m})$ for all $m\in\mathbb N$, and
\begin{equation}\label{eq:T_infty_extreme_points}
\lim_{m\to\infty}\Psi_T( \Phi_{1,m}(\mathbbm1,U_m^\phi) )=\lim_{m\to\infty}\begin{pmatrix}
\frac{\lambda(m-1)}{m}\\
(1-\lambda)e^{i\phi}+\frac{\sqrt{1-\lambda}(1 +e^{i\phi}( 1-2\sqrt{1-\lambda} ) )}{m}
\end{pmatrix}=\begin{pmatrix}
\lambda\\(1-\lambda)e^{i\phi}
\end{pmatrix}\,.
\end{equation}
Thus $\Psi_T^{-1}(\lambda,(1-\lambda)e^{i\phi})\in\overline{ \mathsf{TO}(H_S,\infty) }$ for all $\phi\in[0,2\pi)$ as desired.

(ii): We only have to prove \eqref{eq:TO_qubit_set} because then convexity of $\mathsf{TO}(H_S,T)$ follows directly. First let us see that $\mathsf{TO}(H_S,T)$ is a subset of \eqref{eq:TO_qubit_set}. For this we present a slight modification of the proof of Thm.~1 from \cite{Ding19}: The idea is to find a family of subsets $(\mathcal S_m)_{m}$ of
$\mathsf{TO}(H_S,T)$ such that in the limit $m\to\infty$ their convex hull
$(\mathsf{conv}(\mathcal S_m))_{m}$ exhausts the cone of enhanced 
thermal operations from Figure \ref{fig1}.
The exact form of $\mathcal S_m$ will let us conclude that for every $S\in\mathsf{EnTO}(H_S,T)$ \textit{not} on the boundary there exists $m$ and $S_m\in\mathcal S_m$ such that $S$ is the composition of $S_m$ and a partial dephasing map. Hence $S\in\mathsf{TO}(H_S,T)$ as it is the composition of two thermal operations (Proposition \ref{prop_1} (i)).

Now for the details. Given $T\in(0,\infty)$ (a note on the case $T=\infty$ later), $\lambda\in[0,1]$, $m\in\mathbb N\setminus\{1\}$, $\phi\in[-\pi,\pi)$ define a thermal operation as follows: $H_{B,m}:=\operatorname{diag}(0,\Delta E,\ldots,(m-1)\Delta E)\in {i}\mathfrak u(m)$ is the bath Hamiltonian, and the energy-preserving unitary $U_m^{\lambda,\phi}\in\mathsf U(2m)$ is given by
$$
\begin{pmatrix}
1&0&0&0\\
0& \operatorname{diag}\big( (\frac{1-\lambda}{1-\lambda e^{-\Delta E/T}} )^{j/2} \big)_{j=1}^{m-1} & \operatorname{diag}\big(i(1-(\frac{1-\lambda}{1-\lambda e^{-\Delta E/T}})^j)^{1/2}\big)_{j=1}^{m-1} &0\\
0& \operatorname{diag}\big( ie^{-i\phi}(1-(\frac{1-\lambda}{1-\lambda e^{-\Delta E/T}})^j)^{1/2}\big)_{j=1}^{m-1} & \operatorname{diag}\big(e^{-i\phi} (\frac{1-\lambda}{1-\lambda e^{-\Delta E/T}} )^{j/2} \big)_{j=1}^{m-1} &0\\
0&0&0&e^{-i\phi}\\
\end{pmatrix}\,.
%\in\mathsf U(2m)
$$
Be aware that a variation of this unitary has also appeared in Appendix B of \cite{Lostaglio18} (cf.~also references therein). However the unitary matrix which Lostaglio et al.~use leads to a cone that -- while containing all classical channels $\{\Psi_T^{-1}(\lambda,0):\lambda\in[0,1)\}$ (which was their goal) -- is always a strict subset of $\mathsf{TO}(H_S,T)$, even in the closure.

Now let us collect all maps with the same $\phi$ via
$
\mathcal S_{m,\phi}:=\big\{\Phi_{T,m}(H_{B,m},U_m^{\lambda,\phi}):\lambda\in[0,1]\big\}
$
so the set we are looking for which exhausts $\mathsf{EnTO}(H_S,T)$ in the convex hull as $m$ goes to infinity is $\mathcal S_m:=\bigcup_{\phi\in[-\pi,\pi)}\mathcal S_{m,\phi}$. 

\textit{Claim:} For any $m\in\mathbb N$, $\phi\in[-\pi,\pi)$, applying $\Psi_T$ to the set $\mathcal S_{m,\phi}$ yields a strictly convex curve with end points
\begin{equation}\label{eq:end_points_curve}
\begin{pmatrix}
0\\e^{i\phi}
\end{pmatrix}\text{ (for }\lambda=0)\qquad\text{ and }\qquad
\begin{pmatrix}
\frac{1-( e^{-\Delta E/T} )^{m-1}}{1-( e^{-\Delta E/T} )^{m}}\\0
\end{pmatrix}\text{ (for }\lambda=1)\,.
\end{equation}
and $\Psi_T(\mathcal S_{m,\phi})$ converges to $\{ (\lambda,e^{i\phi}\sqrt{ (1-\lambda)(1-\lambda e^{- \Delta E /T}) } ) :\lambda\in[0,1]\}$ in the Hausdorff metric.
This follows from a direct computation using Lemma \ref{lemma_degenerate_spin_TO}:
\begin{align*}
\Psi_T\big(\Phi_{T,m}(H_{B,m},U_m^{\lambda,\phi})\big)=\begin{pmatrix}
\frac{1-(e^{-\Delta E/T})^{m-1}}{1-(e^{-\Delta E/T})^{m}}-\gamma^2\frac{1-e^{-\Delta E/T}}{1-\gamma^2 e^{-\Delta E/T}}\frac{1-(\gamma^2 e^{-\Delta E/T})^{m-1}}{1-(e^{-\Delta E/T})^{m}}
\\
e^{i\phi}(1-e^{-\Delta E/T})\big( \frac{\gamma}{1-\gamma^2 e^{-\Delta E/T}}\frac{1-(\gamma^2 e^{-\Delta E/T})^{m-1}}{1-(e^{-\Delta E/T})^{m}}-\frac{(\gamma e^{-\Delta E/T})^{m-1}}{1-(e^{-\Delta E/T})^{m}} \big)
\end{pmatrix}
\end{align*}
where $\gamma:=\sqrt{\frac{1-\lambda}{1-\lambda e^{-\Delta E/T}}}$. Note that $\gamma$ is strictly monotonically decreasing in $\lambda$ and $\gamma|_{\lambda=0}=1$, $\gamma|_{\lambda=1}=0$; hence $\gamma$ is bijective on $[0,1]$ as a function of $\lambda$. In particular setting $\lambda\in\{0,1\}$ ($\gamma\in\{0,1\}$) reproduces \eqref{eq:end_points_curve}. Taking the limit $m\to\infty$ yields
\begin{align*}
\lim_{m\to\infty}\Psi_T\big(\Phi_{T,m}(H_{B,m},U_m^{\lambda,\phi})\big)&=\begin{pmatrix} 1-\gamma^2\frac{1-e^{-\Delta E/T}}{1-\gamma^2e^{-\Delta E/T}}\\e^{-i\phi}\frac{1-e^{-\Delta E/T}}{1-\gamma^2e^{-\Delta E/T}}\gamma \end{pmatrix}\\
&=\begin{pmatrix}
1-\gamma^2\frac{1-\lambda}{\gamma^2}\\e^{i\phi}\frac{1-\lambda}{\gamma^2}\gamma
\end{pmatrix}=\begin{pmatrix}
\lambda\\
e^{i\phi}\sqrt{(1-\lambda)(1-\lambda e^{-\Delta E/T})}
\end{pmatrix}
\end{align*}
as claimed. Here we used the readily verified identity $\frac{1-e^{-\Delta E/T}}{1-\gamma^2e^{-\Delta E/T}}=\frac{1-\lambda}{\gamma^2}$. Note that the case $T=\infty$ is proven analogously once the unitary $U_m^{\lambda,\phi}$ is given by \eqref{eq:T_infty_U} (as the computation in \eqref{eq:T_infty_extreme_points} shows). 

Now let $\lambda\in(0,1)$, $r\in[0,\sqrt{(1-\lambda)(1-\lambda e^{-\Delta E/T})})$, and $\phi\in[-\pi,\pi)$ be given, that is, $(\lambda,re^{i\phi})$ does not lie on the relative boundary of $\Psi_T(\mathsf{EnTO}(H_S,T))$ (we can exclude the case $\lambda=0$ as we already know that all partial dephasings are elements of $\mathsf{TO}$, cf.~\eqref{eq:map_deph}). Because $\operatorname{conv}(\mathcal S_m)$ is strictly monotonically increasing in $m$ and because $\lim_{m\to\infty}\delta(\operatorname{conv}(\mathcal S_m),\mathsf{EnTO}(H_S,T))=0$ there exists $m\in\mathbb N$ such that $\Psi_T^{-1}(\lambda,re^{i\phi})\in \operatorname{conv}(\mathcal S_m)$. But by construction of $\mathcal S_m$ this means that $\Psi_T^{-1}(\lambda,r'e^{i\phi})\in\mathcal S_m$ for some $r'\geq r$ so
\begin{align*}
\Psi_T^{-1}\begin{pmatrix}
\lambda\\re^{i\phi}
\end{pmatrix}=\Psi_T^{-1}\begin{pmatrix}
0\\\frac{r}{r'}
\end{pmatrix}\circ\Psi_T^{-1}\begin{pmatrix}
\lambda\\r'e^{i\phi}
\end{pmatrix}&\in\mathsf{TO}(H_S,T)\circ\mathcal S_m\\
&\subseteq\mathsf{TO}(H_S,T)\circ\mathsf{TO}(H_S,T)=\mathsf{TO}(H_S,T)\,.
\end{align*}
Here we used again that all partial dephasings are thermal operations (\eqref{eq:map_deph}, as $\frac{r}{r'}<1$).
% and that $\mathsf{TO}(H_S,T)$ forms a semigroup (Proposition \ref{prop_1}).

Conversely, to see that \eqref{eq:TO_qubit_set} is a subset of $\mathsf{TO}(H_S,T)$ 
we have to show that
\begin{equation}\label{eq:extreme_points_not_in_TO}
\Psi_T^{-1}\begin{pmatrix}
\lambda\\e^{i\phi}\sqrt{(1-\lambda)(1-\lambda e^{-\Delta E/T})}
\end{pmatrix}\not\in \mathsf{TO}(H_S,T)
\end{equation}
for all $\lambda\in(0,1]$, $\phi\in[-\pi,\pi)$, $T\in(0,\infty)$ (proving \eqref{eq:extreme_points_not_in_TO} for $T=\infty$ is done analogously).

Assume to the contrary that 
\eqref{eq:extreme_points_not_in_TO} is false. Hence there exist
$m\in\mathbb N$ and $\alpha_0,\ldots,\alpha_{m-1}\in\mathbb N$ such that $\Phi_{T,m}
(H_B,U)=\Psi_T^{-1}(\lambda,e^{i\phi}\sqrt{(1-\lambda)(1-\lambda e^{-\Delta E/T})})$ for 
some energy-preserving unitary $U\in\mathsf U(2( \sum_{j=0}^{m-1}\alpha_j ))$ where $H_B:=\bigoplus_{j=0}^{m-1}j\Delta E\cdot\mathbbm1_{\alpha_j}$. The 
reason for choosing resonant $H_B$ is that $\Psi_T^{-1}(\lambda,e^{i\phi}\sqrt{(1-\lambda)(1-\lambda 
e^{-\Delta E/T})})$ is an extreme point of $\mathsf{EnTO}(H_S,T)$, and the proof of 
Proposition \ref{prop_1} (iv) (cf.~\eqref{eq:proof_lemma_iv_decomp}) shows that any thermal operation with a bath Hamiltonian 
which is \textit{not} of this form can be written as a convex combination of two 
thermal operations with bath Hamiltonians of the above form. But this contradicts the extreme point 
property so $H_B$ has to have resonant spectrum w.r.t.~$H_S$.

Due to $H_B$ being of spin form we may apply Lemma \ref{lemma_degenerate_spin_TO} to get an explicit form of $U$ and, more importantly, $\Psi_T(\Phi_{T,m}
(H_B,U))$. Define an inner product $\langle\cdot,\cdot\rangle_T$ on $\mathbb C^{\alpha_0\times\alpha_0}\times\ldots\times\mathbb C^{\alpha_{m-1}\times\alpha_{m-1}}$ via
\begin{align*}
\Big(\begin{pmatrix}
X_1\\\vdots\\X_m
\end{pmatrix},\begin{pmatrix}
Y_1\\\vdots\\Y_m
\end{pmatrix}\Big)&\mapsto \sum_{j=0}^{m-1}\langle X_j,Y_j\rangle_\mathsf{HS}\,e^{-j\Delta E/T}
\end{align*}
where $\langle A,B\rangle_\mathsf{HS}=\operatorname{tr}(A^*B)$ is the Hilbert-Schmidt inner product on complex square matrices of any dimension. Note that $\langle\cdot,\cdot\rangle_T$ is indeed an inner product because it is a sum of inner products with positive weights. This lets us rewrite $c$ from Lemma \ref{lemma_degenerate_spin_TO} as
$$
c=\frac{\langle (
D_1,\cdots,D_{m-1},U_m
),(
U_0,A_1,\cdots,A_{m-1}
) \rangle_T}{\sum_{j=0}^{m-1}\alpha_je^{-j\Delta E/T}}\,.
$$
In particular we can apply the Cauchy-Schwarz inequality to obtain
\begin{align*}
|c|&\leq \frac{\|(D_1,\ldots,D_{m-1},U_m)\|_T\|(U_0,A_1,\ldots,A_{m-1})\|_T}{\sum_{j=0}^{m-1}\alpha_je^{-j\Delta E/T}}\\
&= \frac{\sqrt{( \sum_{j=0}^{m-2}\|D_{j+1}\|_\mathsf{HS}^2e^{-j\Delta E/T}+\|U_m\|_\mathsf{HS}^2 e^{-\frac{(m-1)\Delta E}T} )(\|U_0\|_\mathsf{HS}^2+\sum_{j=1}^{m-1}\|A_j\|_\mathsf{HS}^2e^{-j\Delta E/T} )}}{\sum_{j=0}^{m-1}\alpha_je^{-j\Delta E/T}}\,.
\end{align*}
Now unitarity of $U$ comes into play: On the one hand $U_0,U_m$ are itself unitary so $\|U_0\|_\mathsf{HS}^2=\alpha_0$, $\|U_{m-1}\|_\mathsf{HS}^2=\alpha_{m-1}$. On the other hand unitarity of the blocks $U$ is made up of (i.e.~\eqref{eq:U_j_unitary} being unitary) implies $A_jA_j^*+B_jB_j^*=\mathbbm1_{\alpha_j}$ and $B_j^*B_j+D_j^*D_j=\mathbbm1_{\alpha_{j-1}}$ for all $j=1,\ldots,m-1$. Taking the trace yields
\begin{align*}
\|D_{j+1}\|_\mathsf{HS}^2=\alpha_j-\|B_{j+1}\|_\mathsf{HS}^2\qquad&\text{ for all }j=0,\ldots,m-2\,,\text{ and}\\
\|A_{j}\|_\mathsf{HS}^2=\alpha_j-\|B_{j}\|_\mathsf{HS}^2\qquad&\text{ for all }j=1,\ldots,m-1\,.
\end{align*}
With this the upper bound we found for $|c|$ is equal to
$$
\sqrt{\Big( 1- \frac{\sum_{j=0}^{m-2}\operatorname{tr}(B_{j+1}B_{j+1}^*)e^{-j\Delta E/T}}{\sum_{j=0}^{m-1}\alpha_j e^{-j\Delta E/T}} \Big)\Big( 1- \frac{\sum_{j=0}^{m-2}\operatorname{tr}(B_{j+1}B_{j+1}^*)e^{-j\Delta E/T}}{\sum_{j=0}^{m-1}\alpha_j e^{-j\Delta E/T}} e^{-\Delta E/T} \Big)}\,,
$$
that is, $|c|\leq\sqrt{(1-\lambda)(1-\lambda e^{-\Delta E/T})}$ with equality if and only if there is equality in Cauchy-Schwarz because that was the only estimate we used in our calculation. But it is well known that equality in Cauchy-Schwarz is equivalent to the two arguments being a scalar multiple of each other. Hence there exists $\xi\in\mathbb C$ such that
\begin{equation}\label{eq:lin_dep_matrix_vector}
\begin{pmatrix}
D_1\\\cdots\\D_{m-1}\\U_m
\end{pmatrix}=\xi\begin{pmatrix}
U_0\\A_1\\\cdots\\A_{m-1}\end{pmatrix}\,.
\end{equation}
Due to unitarity $\|U_0\|_\infty,\|U_m\|_\infty=1$ as well as $\|D_j\|_\infty,\|A_j\|_\infty\leq 1$ for all $j=1,\ldots,m$. Therefore
$
|\xi|=\|\xi U_0\|_\infty=\|D_1\|_\infty\leq 1
$
and
$
1=\|U_m\|_\infty=|\xi|\|A_{m-1}\|_\infty\leq |\xi|
$
so $|\xi|=1$. But with this \eqref{eq:lin_dep_matrix_vector} forces all $B_j$ to vanish: first $D_1^*D_1=|\xi|^2U_0^*U_0=\mathbbm1$ so $B_1=0$ and thus $A_1^*A_1=\mathbbm1$ because \eqref{eq:U_j_unitary} is unitary. Then considering the second element of \eqref{eq:lin_dep_matrix_vector} implies $D_2^*D_2=|\xi|^2A_1^*A_1=\mathbbm1$ so $B_2=0$; repeating this argument inductively shows $B_j=0$ for all $j=1,\ldots,m-1$. But this is problematic because then $\lambda=0$ (again by Lemma \ref{lemma_degenerate_spin_TO}) which contradicts our assumption that $\lambda\in(0,1]$.\medskip

(iii): This result is a truncated version of (iv) (i.e.~of Thm.~1 from \cite{Ding19}) so the proof we present is inspired by the arguments of Ding et al.
We showed in (ii) that for qubits every $S\in\mathsf{TO}(H_S,T)$ is the composition of a thermal operation generated by $H_{B,m}:=\operatorname{diag}(0,\Delta E,\ldots,(m-1)\Delta E)$ for some $m\in\mathbb N$ and a (partial) dephasing. So if $T>\frac{\Delta E}{\ln 2}$ it suffices to show that each partial dephasing can be implemented using some $H_{B,m}$.

For this note that given $m\in\mathbb N$ and arbitrary phases $\phi_1,\ldots,\phi_{m}\in[-\pi,\pi)$ the unitary $U:=\mathbbm1_m\oplus U_\phi$ with $U_\phi:=\operatorname{diag}(e^{i\phi_1},\ldots,e^{i\phi_{m}})$ is energy-preserving because it is diagonal. A straightforward computation yields
\begin{equation*}%\label{eq:dephasing_rel_phases}
\begin{split}
\Phi_{T,m}(H_{B,m},U_\phi)(A)=\begin{pmatrix}
a_{11}& a_{12}\big(\frac{\operatorname{tr}(e^{-H_{B,m}/T}U_\phi)}{\operatorname{tr}(e^{-H_{B,m}/T})}\big)^*\\
a_{21}\frac{\operatorname{tr}(e^{-H_{B,m}/T}U_\phi)}{\operatorname{tr}(e^{-H_{B,m}/T})}&a_{22}
\end{pmatrix}
\end{split}
\end{equation*}
for all $A\in\mathbb C^{2\times 2}$. Thus all we have to see is that for $m$ ``large enough'' the map $(\phi_1,\ldots,\phi_{m})\mapsto \frac{\operatorname{tr}(e^{-H_{B,m}/T}U_\phi)}{\operatorname{tr}(e^{-H_{B,m}/T})}$ maps surjectively onto the closed unit disk. The key observation here is that given numbers $0<c_2<c_1$ the (pointwise) sum of a circle with radius $c_2$ to a circle with radius $c_1$ both centered around the origin (i.e.~$\{c_1e^{i\phi_1}+c_2e^{i\phi_2}:\phi_1,\phi_2\in[-\pi,\pi)\}$) is equal to the annulus $\{re^{i\phi}:c_1-c_2\leq r\leq c_1+c_2,\phi\in[-\pi,\pi)\}$ with inner radius $c_1-c_2$ and outer radius $c_1+c_2$. We visualize this fact in Figure \ref{fig_circles} which makes a proof superfluous.
\begin{figure}[!htb]
\centering
\includegraphics[width=0.45\textwidth]{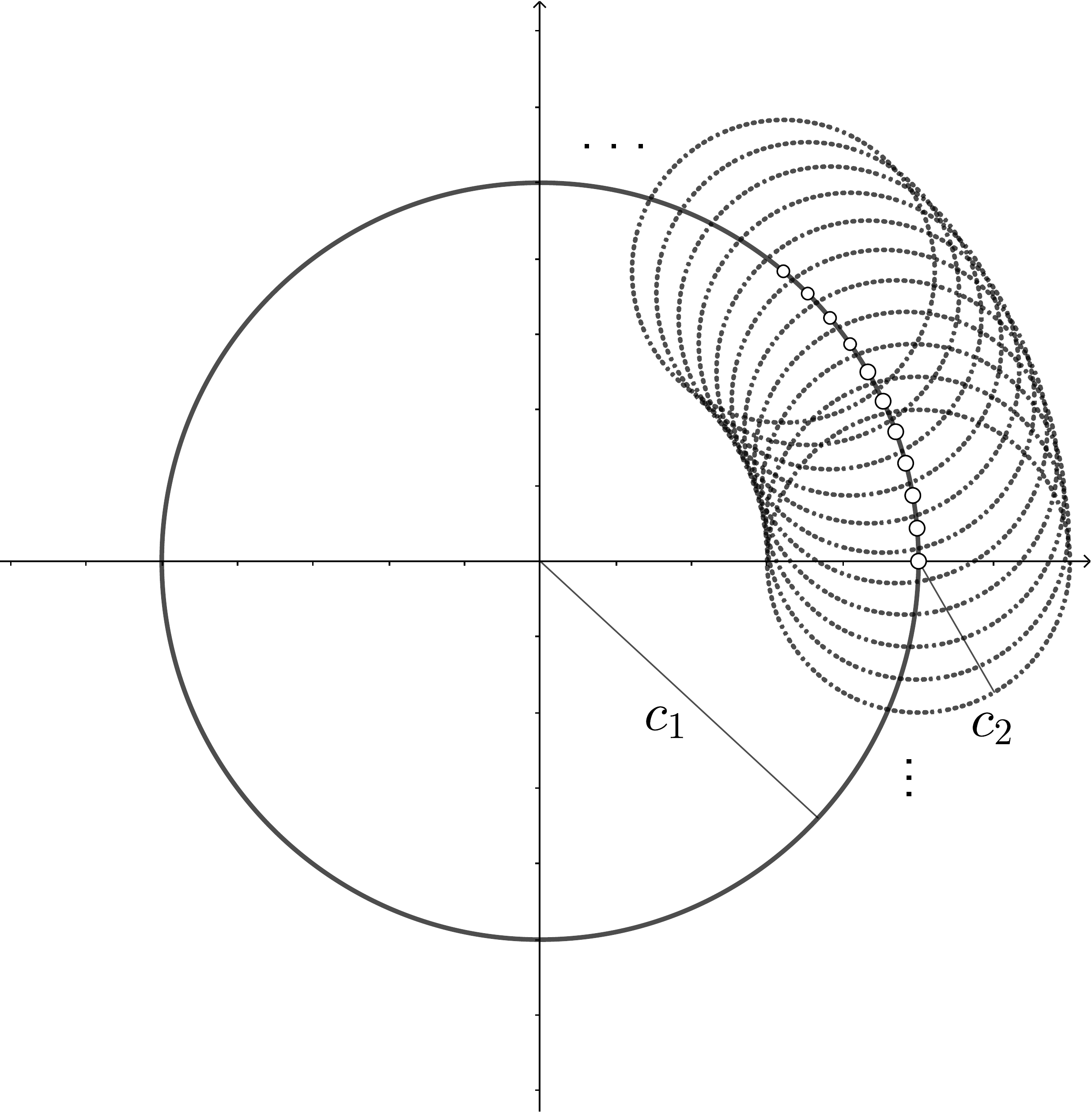}
\includegraphics[width=0.45\textwidth]{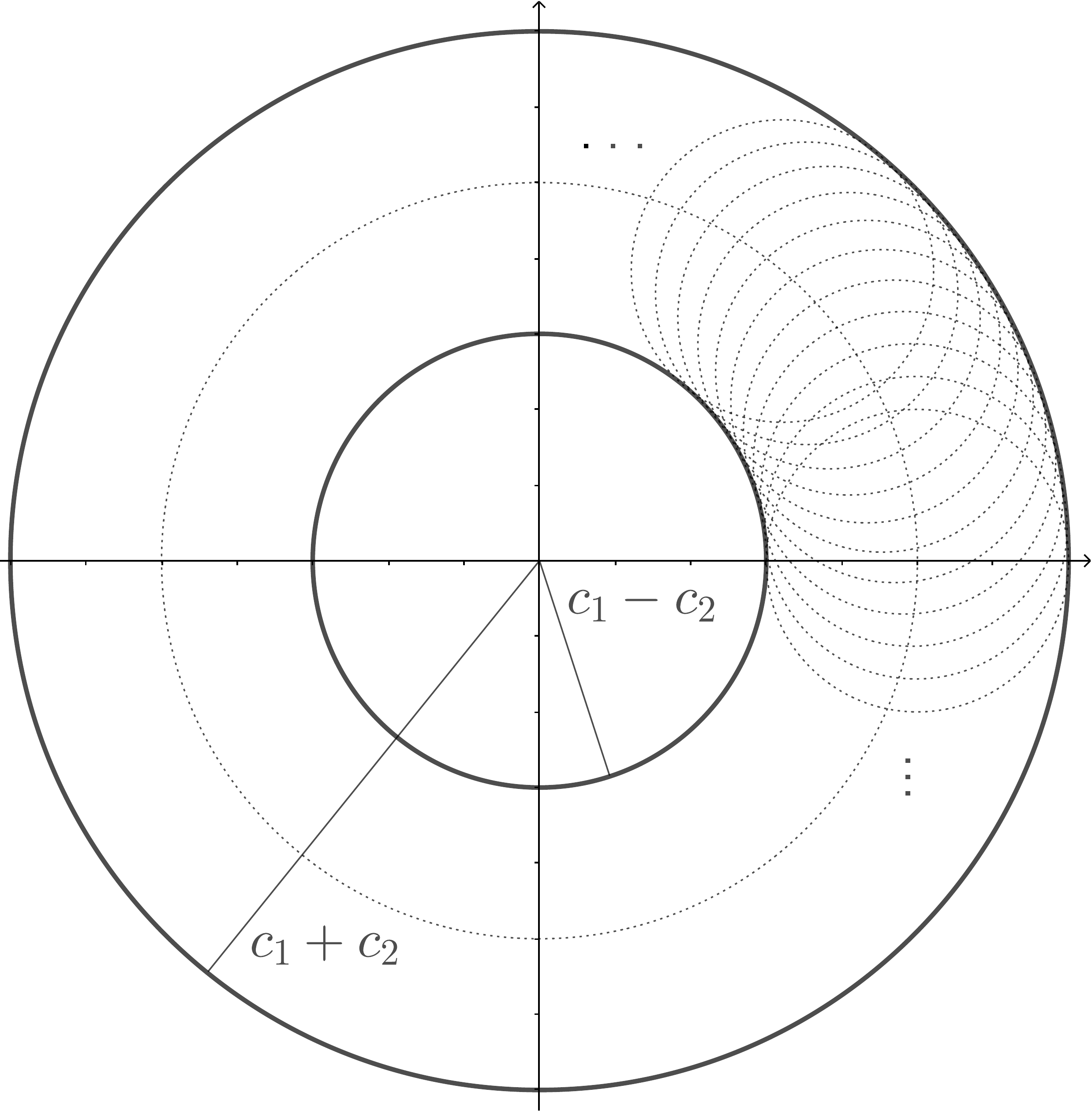}
\caption{
Visual proof of the equality of $\{c_1e^{i\phi_1}+c_2e^{i\phi_2}:\phi_1,\phi_2\in[-\pi,\pi)\}$ and $\{re^{i\phi}:c_1-c_2\leq r\leq c_1+c_2,\phi\in[-\pi,\pi)\}$ for all $0<c_2<c_1$. Left: Sketch of how each individual set $c_1e^{i\phi_1}+\{c_2e^{i\phi_2}:\phi_2\in[-\pi,\pi)\}$ looks like. Right: The union of these individual sets exhausts the full annulus.
}
\label{fig_circles}
\end{figure}

This implies that the expression
\begin{align*}
\frac{\operatorname{tr}(e^{-H_{B,m}/T}U_\phi)}{\operatorname{tr}(e^{-H_{B,m}/T})}
=\frac{\sum_{j=0}^{m-1}e^{-j\Delta E/T}e^{i\phi_j}}{\sum_{k=0}^{m-1}e^{-k\Delta E/T}}=\sum_{j=0}^{m-1}\frac{(1-e^{-\Delta E/T})e^{-j\Delta E/T}}{1-e^{-m\Delta E/T}}e^{i\phi_j}
\end{align*}
can take any value in the annulus $\{re^{i\phi}:\max\{r_m,0\}\leq r\leq 1,\phi\in[-\pi,\pi)\}$ where
$$
r_m=\frac{1-e^{-\Delta E/T}}{1-e^{-m\Delta E/T}}-\sum_{j=1}^{m-1}\frac{(1-e^{-\Delta E/T})e^{-j\Delta E/T}}{1-e^{-m\Delta E/T}}=2\frac{1-e^{-\Delta E/T}}{1-e^{-m\Delta E/T}}-1\,.
$$
But $\lim_{m\to\infty}r_m=1-2e^{-\Delta E/T}$ which is smaller than zero if and only if $T>\frac{\Delta E}{\ln 2}$; thus by assumption there exists $m\in\mathbb N$ such that $r_m<0$ so $(\phi_1,\ldots,\phi_{m})\mapsto \frac{\operatorname{tr}(e^{-H_{B,m}/T}U_\phi)}{\operatorname{tr}(e^{-H_{B,m}/T})}$ maps surjectively onto the closed unit disk. In other words for this $m$ all partial dephasings can be implemented via relative phases which is what we had to show.

Now if $T\leq\frac{\Delta E}{\ln 2}$ we will prove that the two sets in question do not coincide by showing that full dephasing is cannot be implemented using $H_{B,m}$. Indeed given arbitrary $m\in\mathbb N$ and any $U\in\mathsf U(2m)$ such that $[U,H_S\otimes\mathbbm1_B+\mathbbm1_S\otimes H_{B,m}]=0$, partitioning
$$
U=\begin{pmatrix} U_{11}&U_{12}\\U_{21}&U_{22} \end{pmatrix}
$$
with $U_{11},U_{12},U_{21},U_{22}\in\mathbb C^{m\times m}$ leads to
$$
\Psi_T\big(\Phi_{T,m}(H_{B,m},U)\big)=\frac{1}{\operatorname{tr}(e^{-H_{B,m}/T})}\begin{pmatrix}
\operatorname{tr}(U_{12}^*U_{12}e^{-H_{B,m}/T})\\
\operatorname{tr}(U_{22}^*U_{11}e^{-H_{B,m}/T})
\end{pmatrix}
$$
as is verified by direct computation. We want this expression to be equal to $(0,0)^\top$. This means $\operatorname{tr}(U_{12}^*U_{12}e^{-H_{B,m}/T})=0$ which implies $U_{12}=0$: this is due to the fact that $(A,B)\mapsto\operatorname{tr}(A^*Be^{-H_{B,m}/T})$ is an inner product on $\mathbb C^{m\times m}$ because $e^{-H_{B,m}/T}$ is positive definite. Thus $A\mapsto\operatorname{tr}(A^*Ae^{-H_{B,m}/T})$ is a norm on $\mathbb C^{m\times m}$ so it takes the value zero if and only if the input is zero. But as $U$ is unitary $U_{12}=0$ implies $U_{21}=0$ so $U=U_{11}\oplus U_{22}$ for some $U_{11},U_{22}\in\mathsf U(m)$. Moreover $U$ being energy-conserving yields $[U_{11},H_B]=[U_{22},H_B]=0$, and as $H_B$ is non-degenerate by assumption $U_{11},U_{22}$ (and thus $U$) have to be diagonal. However for diagonal $U$ we already showed that full dephasing can be implemented if and only if $T>\frac{\Delta E}{\ln 2}$ which concludes this part of the proof.

Finally the statement regarding the closure. By our previous argument (cf.~Figure \ref{fig_circles}) regardless of the temperature at least some range of dephasing maps can be implemented (e.g., using diagonal unitaries $U$), that is, for all $T>0$ there exists $r_0<1$ such that for all $c\in\mathbb C$, $|c|\in[r_0,1]$
\begin{align*}
S_c:\mathbb C^{2\times 2}&\to\mathbb{C}^{2\times 2}\\
\begin{pmatrix}
a_{11}&a_{12}\\a_{21}&a_{22}
\end{pmatrix}&\mapsto\begin{pmatrix}
a_{11}&c^*a_{12}\\ca_{21}&a_{22}
\end{pmatrix}
\end{align*}
is a thermal operation with
bath-Hamiltonian $\operatorname{diag}(0,\Delta E,\ldots,(m-1)\Delta E)$ for some 
$m$. On the other hand we know that every qubit thermal 
operation $S$ is the composition of a thermal operation $S_m$ with associated 
$H_{B,m}:=\operatorname{diag}(0,\Delta E,\ldots,(m-1)\Delta E)$ for some 
$m\in\mathbb N$, and a partial dephasing $D_S$ (i.e.~$\Psi_T(D_S)=(0,c_S)$ for 
some $c_S\in[0,1]$). But applying any $S_c$, $|c|<1$ enough times approximates 
any degree of
dephasing, i.e.~$\lim_{k\to\infty}\Psi_T(S_{r_0}^k)=\lim_{k\to\infty}(0,r_0^k)=(0,0)$.
Thus there are two 
cases: If $c_S>0$ then there exists $k\in\mathbb N_0$ and $c\in[r_0,1]$ such that 
$D_S=S_{r_0}^k\circ S_c$. Therefore $S=S_{r_0}^k\circ S_c\circ S_m$ so $S$ can be 
implemented exactly using finitely many truncated single-mode bosonic baths.
However if $c_S=0$ then this can (only) be done approximately, 
i.e.~$\lim_{k\to\infty}\|S-S_{r_0}^k\circ S_m\|=0$. Therefore $\mathsf{TO}(H_S,T)$ is a 
subset of the closure of the semigroup of thermal operations with bath-Hamiltonian 
$H_{B,m}:=\operatorname{diag}(0,\Delta E,\ldots,(m-1)\Delta E)$ (which itself is a subset of $\overline{\mathsf{TO}(H_S,T)}$) 
meaning the two sets coincide in the closure.\medskip

(iv): This is Theorem 1.(2) in \cite{Ding19}. Because our proof of (iii) (the ``truncated version'') is similar to their proof we will omit the details, and we simply refer to Appendix 2 in their paper.
\end{proof}
\bibliography{../../../../../../control21vJan20.bib} % name your BibTeX data base
\end{document}